\documentclass[12pt]{article}
\usepackage{amssymb,amsmath}
\usepackage{amsthm}
\newtheorem{proposition}{Proposition}
\makeatletter

\usepackage{arydshln}

\newcount\comment@nesting

\begingroup
\xdef\comment@begincomment{\string\\begin\string\{comment\string\}}
\xdef\comment@endcomment{\string\\end\string\{comment\string\}}
\endgroup

\begingroup
\catcode`\^^M=\active
\def\@temp{\endgroup\def\comment@processline##1^^M}%
\@temp{%
    \def\comment@curline{#1}%
    \let\@next=\comment@processline
    \ifx\comment@curline\comment@endcomment
        \ifnum\comment@nesting=0
            \def\@next{\end{comment}}%
        \else
            \global\advance\comment@nesting by -1
        \fi
    \else
        \ifx\comment@curline\comment@begincomment
            \global\advance\comment@nesting by 1
        \fi
    \fi
    \@next
}

\makeatother

\usepackage[dvipdfmx]{graphicx}
\usepackage{bm}
\usepackage{color}
\usepackage{here}
\usepackage{enumerate}
\usepackage{here}
\usepackage[vcentermath]{youngtab}
\usepackage{subfigure}
\usepackage{tikz}
\usepackage{hyperref}

\newcommand{\Zb}{\mathbb{Z}}

\numberwithin{equation}{section}

\allowdisplaybreaks[1]

\definecolor{mygreen}{rgb}{0,0.714,0.286}
\definecolor{purple}{rgb}{0.55,0.0,0.60}
\newif\ifshowchanges \showchangestrue
\ifshowchanges 
\else          \fi

\begin{document}

\thispagestyle{empty}
\begin{flushright}

\end{flushright}
\vskip1.5cm
\begin{center}
{\Large \bf 
%
Symplectic Chern-Simons theories \\
\bigskip
dual to Abelian rank-$0$ theories
}

\vskip1.5cm
Tadashi Okazaki\footnote{tokazaki@seu.edu.cn}

\bigskip
{\it School of Physics and Shing-Tung Yau Center, Southeast University,\\
Yifu Architecture Building, No.2 Sipailou, Xuanwu district, \\
Nanjing, Jiangsu, 210096, China
}

\bigskip
and
\\
\bigskip
Douglas J. Smith\footnote{douglas.smith@durham.ac.uk}

\bigskip
{\it Department of Mathematical Sciences, Durham University,\\
Upper Mountjoy, Stockton Road, Durham DH1 3LE, UK}

\end{center}

\vskip1cm
\begin{abstract}
We propose that the 3d symplectic Chern-Simons theory $USp(2n)$ at level
$k=n+\frac12+\ell$ with a single chiral multiplet in the fundamental representation is dual, for
every pair of positive integers $(n,\ell)$, to the Abelian rank-$0$ theory $\mathcal{T}_{\ell,n}$
of Creutzig, Garner and Kim.  The fundamental is pseudo-real, so the level is half-integral,
there is no non-Abelian flavor symmetry, and the supersymmetry is enhanced to $\mathcal{N}=4$.
Both sides have vanishing Higgs and Coulomb branches.  We match the counting of supersymmetric vacua under real mass deformations for all
$(n,\ell)$, and the supersymmetric indices in a number of cases.  We then study in detail the
case $\ell=1$, the minimal $USp(2n)$ Chern-Simons theory at level $n+\frac32$, whose dual is a
$U(1)^{n}$ theory with a tadpole level matrix and a bare monopole superpotential.
\end{abstract}
\newpage
\setcounter{tocdepth}{3}
\tableofcontents

\section{Introduction and conclusion}
Three-dimensional $\mathcal{N}=2$ supersymmetric gauge theories exhibit a rich and intricate landscape of Seiberg-like infrared dualities. 
The exploration of this landscape originated from the foundational discoveries of the Aharony duality for theories 
with vanishing Chern-Simons (CS) levels \cite{Aharony:1997gp,Karch:1997ux} and
the Giveon-Kutasov duality for CS-matter theories \cite{Giveon:2008zn,Niarchos:2008jb}. 
Strong evidence for these dualities was subsequently provided by supersymmetric localization techniques computing exact partition functions and supersymmetric indices
\cite{Kapustin:2009kz,Kim:2009wb,Imamura:2011su,Kapustin:2011jm}, 
which not only tested the dualities across various gauge groups and matter representations but also elucidated the precise mapping of monopole operators and global symmetries
\cite{Benini:2011mf,Kapustin:2011gh,Hwang:2011ht,Aharony:2011ci,Kim:2013cma,Park:2013wta}. 
In addition, a detailed field theoretic analysis of the vacuum structure of CS-matter theories under real mass deformations clarified 
how these dualities are related to each other along RG flows, in particular allowing one to derive the Aharony duality from the Giveon-Kutasov duality \cite{Intriligator:2013lca}. 
A unified perspective was further achieved by directly deriving the 3d dualities from 4d Seiberg dualities 
via circle compactifications supplemented with suitable real mass flows \cite{Aharony:2013dha,Aharony:2013kma}.

Building upon these foundations, the web of 3d dualities has been substantially enlarged in several directions: 
to theories with classical gauge groups involving generic matter content and CS levels \cite{Aharony:2014uya,Nii:2014jsa,Hwang:2015wna,Hwang:2018uyj}, 
and to exceptional gauge groups \cite{Benvenuti:2018bav,Nii:2019dwi}. 
A pivotal role in these developments has been played by dynamically generated monopole superpotentials \cite{Benini:2017dud} 
as well as by the closely related mechanisms of sequential confinement \cite{Benvenuti:2017lle,Benvenuti:2017kud}
and tensor deconfinement \cite{Benvenuti:2020gvy}. 
For further developments of the 3d duality web in these and related directions, see e.g.
\cite{Nii:2018bgf,Amariti:2018wht,Nii:2019qdx,Nii:2019wjz,Nii:2020eui,
Nii:2020xgd,Nii:2020ikd,Amariti:2020xqm,Amariti:2021snj,Kubo:2021ecs,
Benvenuti:2021nwt,Hwang:2024hhy,Jia:2025koz,Amariti:2025gca}.

A second and rather different class of infrared equivalences is provided by 3d mirror
symmetry \cite{Intriligator:1996ex}, which relates pairs of $\mathcal{N}=4$ theories in such a
way that the Higgs branch of one is the Coulomb branch of the other.  The map exchanges at the
same time the flavor and topological symmetries and the two factors $SU(2)_H$ and $SU(2)_C$ of
the R-symmetry, so that quantities protected on one branch of one theory are computed on the
other branch of its mirror.  It was understood soon afterwards in terms of brane
configurations, where it descends from S-duality of type IIB string theory
\cite{Hanany:1996ie,deBoer:1996mp,deBoer:1996ck}, and for Abelian theories the whole web can be generated
from the single basic equivalence between $U(1)$ with one flavor and a free hypermultiplet
\cite{Kapustin:1999ha}.

Much less is settled when only four supercharges are preserved.  There is then no
hyperkahler structure to organize the two branches, the superconformal R-symmetry is a single
$U(1)$ that can mix with the flavor symmetries, and monopole superpotentials are generally
needed, so the mirror map is not fixed by the geometry of the moduli spaces in the way it is
above.  For Abelian gauge theories the 3d $\mathcal{N}=2$ mirror symmetry is nevertheless
understood \cite{Aharony:1997bx,deBoer:1997kr}, and later work has gone beyond that case,
constructing from webs of $(p,q)$ five-branes mirror pairs whose Abelian member carries a
monopole superpotential \cite{Benvenuti:2016wet}.

The web of 3d dualities is not exhausted by these two classes. 
Equivalences are known that belong to neither, such as the duality \cite{Kapustin:2010xq} between 
the ADHM theory with one flavor \cite{deBoer:1996mp,deBoer:1996ck} and the ABJM theory at level one \cite{Aharony:2008ug}, 
or those relating $USp(2N)$ and $O(N)$ theories with rank-two matter to quiver Chern-Simons matter theories \cite{Gang:2011xp}.  
Uncovering further dualities that fall outside the classes established so far is an interesting problem in its own right.

In parallel, supersymmetric boundary conditions have provided a refined probe of 3d dualities \cite{Gadde:2013wq,Okazaki:2013kaa}.
Pairs of half-BPS $\mathcal{N}=(0,2)$ boundary conditions exchanged under the bulk dualities were proposed and tested through anomaly matching and the computation
of half-indices \cite{Dimofte:2017tpi}.
The same programme has since been carried out for other amounts of boundary supersymmetry
and for other gauge groups, e.g.\ $\mathcal{N}=(0,4)$ and $\mathcal{N}=(2,2)$ boundary
conditions in Abelian theories \cite{Okazaki:2019bok,Okazaki:2020lfy,Okazaki:2025rjl},
orthogonal and symplectic gauge groups \cite{Okazaki:2021pnc}, quivers \cite{Okazaki:2021gkk}
and exceptional gauge groups \cite{Okazaki:2023kpq}, and more recently the rank-$0$ theories
\cite{Creutzig:2024ljv}.
Extending this refined perspective, the
systematic exploration of boundary conditions and the matching of half-indices
have more recently led to proposals of new dualities for the bulk theories
themselves \cite{Okazaki:2021gkk,Okazaki:2023hiv}.

In this paper we propose an exact infrared duality between a family of non-Abelian
$USp(2n)$ Chern-Simons theories and the Abelian rank-$0$ theories of \cite{Creutzig:2024ljv}.
The symplectic side is a 3d $USp(2n)$ gauge theory coupled to \textit{minimal}
matter, a single chiral multiplet in the fundamental representation, at Chern-Simons level
$k=n+\frac12+\ell$ with $\ell$ a positive integer. The index
vanishes identically for $k<n+\frac12$, where supersymmetry is broken, and equals $1$ at
$k=n+\frac12$, where the theory has no local operators, so $\ell\ge1$ is required for a non-trivial theory.  We argue that this is dual to the Abelian theory
$\mathcal{T}_{\ell,n}$, which is $U(1)^{n(2\ell-1)}$, with a level matrix
$C(T_{n})^{-1}\otimes C(A_{2\ell-1})$ which we will explain later, with one chiral of unit charge at each node and a monopole
superpotential.  While general $\ell \ge 1$ is treated throughout, we include detailed examination of the case $\ell=1$, where the Abelian side is the $U(1)^{n}$ \textit{tadpole}
theory, whose level matrix is twice the inverse Cartan matrix of the tadpole graph $T_{n}$ and
whose superpotential is built from bare monopole operators.  It is that requirement, that bare
monopoles suffice, which singles out $\ell=1$. 

The correspondence proposed here appears to be of yet another kind of duality.  It is not a
rearrangement of the same matter content, the two sides being a non-Abelian and an Abelian
theory, and it is not an exchange of a Higgs branch for a Coulomb branch, since both branches
are points on either side.  We return to what can be said about its place, and to the relation
with the rank-$0$ theories of \cite{Gang:2018huc,Gang:2023rei}, in
section~\ref{sec_matching}.

Whatever its ultimate classification, the following features are what set the duality
apart.
\begin{enumerate}

\item \textbf{Pseudo-real matter:}
The fundamental of $USp(2n)$ is pseudo-real, and the symplectic theory carries a single chiral
multiplet in it.  An even number of fundamentals could be paired into a self-conjugate set, but a
single one cannot. 
This single fact organizes much of what follows: 
there is no non-Abelian flavor symmetry and the Chern-Simons level is half-integrally quantized.
We also note that the duality cannot be derived from the usual Giveon-Kutasov-type dualities for $USp(2n)_k$ with $2N_f$ fundamentals.
The Abelian side has the counterpart property, since each node carries a single chiral of
charge $+1$ with no conjugate partner and its matter cannot be paired either.  There is no
pseudo-reality to appeal to in the Abelian case, and the statement is simply that the matter is
chiral.

\item \textbf{A duality of spin theories:}
Both sides carry half-integrally quantized Chern-Simons levels and are therefore spin theories, well defined only after a choice of spin structure on the 3-manifold. 
The duality is a statement about spin theories, and the sign that the parity anomaly leaves in the localization formula must be matched on both sides.

\item \textbf{Non-Abelian/Abelian equivalence:}
To our knowledge, this provides a previously unknown example in quantum field theory where a
strongly coupled non-Abelian $USp(2n)$ gauge theory is completely equivalent in the infrared to
an Abelian gauge theory, and at $\ell=1$ to one of the same rank, for generic $n$.  One-parameter families of
non-Abelian theories can of course have Abelian Seiberg-like duals, but the rank of the Abelian
dual does not then vary along the family.  With a suitable matter content and a Chern-Simons
level depending on $N$, for instance, a $U(N)$ theory can be Seiberg dual to a $U(1)$ theory.  What is
unusual here is that the two ranks grow together.

\end{enumerate}

Furthermore, these specific theories have recently attracted significant interest from entirely different perspectives. 
For the $n=1$ case, the Abelian side is equivalent (up to orientation) to the minimal theory proposed by Gang and Yamazaki \cite{Gang:2018huc}. 
It has been proposed that these theories can emerge from M5-branes 
and are highly distinguished by the fact that their supersymmetry enhances to $\mathcal{N}=4$ in the infrared. 
For general $n$, the Abelian theories precisely coincide with the Abelian Chern-Simons matter theories studied by Gang, Kim and Stubbs \cite{Gang:2023rei}. 
In that context, it has been argued that when endowing these theories with $\mathcal{N}=(0,2)$ boundary conditions,  
the corresponding boundary chiral algebras realize the non-unitary Virasoro minimal models $M(2, 2n+3)$. 
On the symplectic side, the $USp(2n)$ Chern-Simons theory with a single fundamental chiral multiplet has also been actively investigated from the viewpoint of boundary physics. 
In particular, the $\mathcal{N}=(0,2)$ boundary conditions and the associated line defect half-indices for symplectic Chern-Simons theories 
were analyzed for $n=1$ in \cite{Okazaki:2024paq} (noting that $SU(2) \cong USp(2)$) and for general $n$ in \cite{Okazaki:2024kzo}. 
The duality \eqref{the_duality} at $n=1$ was in fact already conjectured there, up to a
parity transformation switching the sign of the $SU(2)$ Chern-Simons level, on the basis of
matching half-indices in the presence of a boundary \cite{Okazaki:2024paq}.  The present paper
establishes it at the level of the bulk indices, extends it to general $n$ and identifies the
dual as the tadpole theory.

\subsection{Structure}
The paper is organized as follows.
Section~\ref{sec_electric} develops the symplectic theory, $USp(2n)$ at level
$k=n+\frac12+\ell$ with a single fundamental chiral.  We fix the level quantization and examine
the sign that the parity anomaly leaves in the index (section~\ref{sec:electric_level}), count
the vacua under the two real mass deformations and show that $\ell\ge1$ is the condition for the
theory to be non-trivial (section~\ref{sec_electric_vacua}), show that there is no moduli space
(section~\ref{sec_electric_moduli}), work out the monopole operators
(section~\ref{sec_monopoles}), record two constraints that any dual must satisfy
(section~\ref{sec_dual_constraints}), write the full index (section~\ref{sec_index}),
show that the supersymmetry is enhanced to $\mathcal{N}=4$ and no further
(section~\ref{sec_electric_susy}) and record a relation among the symplectic indices at
different ranks (section~\ref{sec_level_pair}).
Section~\ref{sec:tadpole} develops the Abelian side.  Section~\ref{sec_tadpole_theory}
treats the tadpole $U(1)^{r}$ Chern-Simons theory, with one charge-one chiral per node and a
bare monopole superpotential, which is the dual at $\ell=1$.  Its subsections deliberately
mirror those of section~\ref{sec_electric}, so that the two theories may be read side by side,
and the important differences then stand out, in particular the absence on the Abelian side of
the pairing of weights that pseudo-reality supplies on the symplectic side.
Section~\ref{sec_general_rank0} then treats the general member $\mathcal{T}_{\ell,n}$ of the
family of \cite{Creutzig:2024ljv}, following the same order point by point.
Section~\ref{sec_matching} states the duality and assembles the evidence, organized by the
two constraints of section~\ref{sec_dual_constraints} and by the structure of the two indices.
Each of sections~\ref{sec_matching_symmetries} to \ref{sec_index_matching} first considers general
$\ell$ and then treats $\ell=1$ in detail, the index comparisons at $\ell>1$ being collected in
sections~\ref{sec_matching_12} to \ref{sec_matching_22} and the pattern of charge-conjugation
symmetry in section~\ref{sec_matching_conj}.
Appendix~\ref{app_bethe} contains the Bethe vacuum computation used in
section~\ref{sec_vacuum_counting}.

\subsection{Future works}
\label{sec_futureworks}

\begin{itemize}

\item \textbf{Boundary conditions and half-indices:} 
A bulk duality may descend to a correspondence between the $\mathcal{N}=(0,2)$ boundary conditions on the two sides and to identities between their half-indices. 
Since the boundary data look very different, a non-Abelian gauge group with a single fundamental
against $n(2\ell-1)$ Abelian nodes with a monopole superpotential, 
the resulting $q$-series identities are not manifest, and they would provide an independent test
of \eqref{family_main}. 
The same setting is also what a treatment of the boundary vertex operator algebras
\cite{Costello:2018fnz} of these theories would require.  Under \eqref{family_main} the
symplectic theory at level $n+\frac12+\ell$ inherits the algebras
$W^{\min}_{n-\frac12}(\mathfrak{sp}_{2\ell})$ and $L_{n}(\mathfrak{osp}_{1|2\ell})$ from its
Abelian dual.  To obtain them on the symplectic side one has to construct the boundary
conditions themselves and the line operators ending on them, neither of which we have done here, and
the $\mathcal{N}=(0,2)$ boundary conditions and line defect half-indices for symplectic
Chern-Simons theories analyzed in \cite{Okazaki:2024paq,Okazaki:2024kzo} are the natural
starting point.
We will report on this in upcoming work.

\item \textbf{Other graphs:} 
The requirements that select the tadpole graph in section~\ref{sec_magnetic_level_choice}, 
symmetry of the level matrix, integrality of $2C^{-1}$, and a single surviving global symmetry, 
are also satisfied for $C$ being the Cartan matrix of $A_1$, $D_{2k}$, $E_7$ or $E_8$. 
It would be interesting to know whether the corresponding Abelian theories admit non-Abelian duals of the present kind, 
and if so which gauge groups and matter contents appear.
The same question arises for general $\ell > 1$.  The level matrix of $\mathcal{T}_{\ell,n}$ is
$C(T_{n})^{-1}\otimes C(A_{2\ell-1})$, and \cite{Creutzig:2024ljv} note that Nahm sums built on
$C(Y)^{-1}\otimes C(X)$ are expected to be modular for other pairs of Dynkin diagrams, without
it being known which of them come from theories with supersymmetry enhancement.  Our
\eqref{family_main} is the statement that the pair $(T_{n},A_{2\ell-1})$ has a symplectic
non-Abelian dual, and one would like to know what, if anything, the other pairs correspond to.

\item \textbf{Other minimal theories:} 
More broadly, one may ask for dual descriptions of minimal Chern-Simons matter theories with other gauge groups. 
The organizing principle here was that the matter cannot be paired into a self-conjugate
set, a condition that can be imposed for orthogonal and exceptional groups as well, and for odd
numbers of fundamentals larger than one.  We note that the Gaiotto-Witten structure of
section~\ref{sec_electric_susy} is special to the symplectic case, since the fundamental identity
holds there because the generators span all of $\mathrm{Sym}^{2}V$.  These other minimal
theories should therefore be expected to be genuinely $\mathcal{N}=2$.

\item \textbf{4d uplift:} 
Many 3d dualities descend from 4d Seiberg dualities compactified on a circle, 
with monopole superpotentials generated along the way \cite{Aharony:2013dha,Aharony:2013kma}. 
A direct 4d parent for the symplectic side is obstructed in an interesting way. 
The 4d $USp(2n)$ gauge theory with an odd number of fundamental Weyl fermions is inconsistent by the Witten anomaly \cite{Witten:1982fp}, 
and the half-integral Chern-Simons level has no 4d ancestor. 
Any 4d origin of \eqref{the_duality} would therefore have to start from an anomaly-free parent, 
e.g. with an even number of fundamentals, and reach the minimal theory by a real mass flow decoupling all but one of them, 
with the parity anomaly generated in the process. 
It would be interesting to determine whether such a derivation exists.

\end{itemize}

\section{Minimal symplectic Chern-Simons theories}
\label{sec_electric}
We consider the 3d $\mathcal{N} = 2$ supersymmetric gauge theory with
\begin{itemize}
  \item symplectic gauge group $G = USp(2n)$;
  \item Chern--Simons level $k = n+\frac12+\ell$ with $\ell$ a positive integer;
  \item a single chiral multiplet $Q$ in the fundamental $\mathbf{2n}$;
  \item vanishing superpotential $\mathcal{W} = 0$.
\end{itemize}
We call the matter content \emph{minimal} because a single chiral multiplet in the
fundamental is the smallest amount of matter a symplectic gauge theory can carry, the
representation being pseudo-real.  The word refers to the matter alone and does not fix the
level, which is why the family is labelled by $\ell$ as well as by $n$.
The field content and the charges under the gauge group, the axial symmetry $U(1)_a$ and the R-symmetry $U(1)_R$ are
\begin{align}
\label{USp2n_Nf1_charges}
\begin{array}{c|c|c|c}
& USp(2n) & U(1)_a & U(1)_R \\ \hline
\textrm{VM}& {\bf Adj} & 0 & 0 \\[2pt]
Q & {\bf 2n} & 1 & r_Q
\end{array}
\end{align}
The trial R-charge $r_Q$ is a free parameter.  We call $U(1)_a$ axial because it
is what survives of the $U(1)_a \times SU(2N_f)$ flavor symmetry of $USp(2n)$ with $2N_f$ fundamentals, i.e.\ 
the single $U(1)_a$ acting on $Q$ is the $N_f = \frac12$ member of that family.

A distinguishing feature of (\ref{USp2n_Nf1_charges}) is that the fundamental of
$USp(2n)$ is pseudo-real, its weights coming in pairs $\pm e_i$.  An even number of fundamentals
could therefore be assembled into a self-conjugate set, but the single one here cannot be.
This has several consequences, in the quantization of the level, in the induced gauge charge of the bare monopole and in the form of the index. 
The most immediate one is that there is no non-Abelian flavor symmetry, 
and the usual Giveon-Kutasov-type dualities \cite{Giveon:2008zn,Kapustin:2011vz,Willett:2011gp} for $USp(2n)_k$ with $2N_f$ fundamentals do not by themselves produce an Abelian dual, although we will see in section~\ref{sec_level_pair} that they do continue to $2N_f=1$ and relate our theories to one another. 
Nevertheless, we will argue that the dual theory is an interesting Abelian Chern-Simons theory. 
Before presenting this argument, we first examine several properties of the theory.

The supersymmetry of \eqref{USp2n_Nf1_charges} is in fact larger than $\mathcal{N}=2$, and
we return to this in section~\ref{sec_electric_susy}, once the index is available.  The trial
R-charge is free because it is a choice of the $U(1)_R$ by which the index is graded, made in
the Lagrangian and available in any $\mathcal{N}=2$ theory with a flavor symmetry containing a $U(1)$.  What the larger supersymmetry fixes
is the \emph{superconformal} R-charge, which turns out to be $r_Q=\tfrac12$ and is determined by
the algebra rather than by F-maximization.  There is one theory here and a one-parameter family
of gradings of it, not a family of theories, since $\mathcal{W}=0$ leaves the Lagrangian
untouched by the choice.  Once enhancement to $\mathcal{N}=4$ is realised the point can be put more sharply
still.  The symmetry $U(1)_a$ is then not a flavor symmetry but the axial combination of the two
R-symmetry Cartans, so the trial R-symmetries $U(1)_R+r_Q\,U(1)_a$ all lie inside
$SU(2)_{H}\times SU(2)_{C}$, and varying $r_Q$ moves between Cartan directions of the enlarged
R-symmetry rather than mixing an R-symmetry with a flavor symmetry.

The theory (\ref{USp2n_Nf1_charges}) has been met before from a different angle. 
For $n=1$, noting $SU(2) \cong USp(2)$, it appears among the 3d $\mathcal{N}=2$ Chern-Simons theories 
whose $\mathcal{N}=(0,2)$ boundary conditions and line defect half-indices were analyzed in \cite{Okazaki:2024paq}, 
and for general $n$ the symplectic case is treated in \cite{Okazaki:2024kzo}, 
where the Neumann half-indices and the one-point functions of fundamental Wilson lines are organized into Rogers-Ramanujan type functions. 
A bulk duality should descend to a correspondence between the boundary conditions on the two
sides and to identities between their half-indices, which we will report on in upcoming work as noted in
section~\ref{sec_futureworks}.

\subsection{Level quantization and the parity anomaly}
\label{sec:electric_level}
An odd number of fundamental fermions induces a parity anomaly, 
so that the level must be half-integrally quantized,
\begin{align}
k+\frac{N_f}{2}\in\Zb
\qquad\Longrightarrow\qquad
k\in \Zb + \frac12
\quad (N_f=1),
\end{align}
which $k=n+\tfrac12+\ell$ satisfies for every integer $\ell$, or equivalently $2k=2n+1+2\ell$ is odd.
This shows up directly in the classical CS contribution to the supersymmetric full index through
\begin{align}
\label{CS_classical_sign}
(-s_i)^{2k m_i}
&=(-1)^{(2n+1+2\ell)m_i}\, s_i^{2k m_i}
=(-1)^{m_i}\, s_i^{2k m_i}\,,
\end{align}
so that the integrand carries an overall factor $(-1)^{\sum_i m_i}$.

The extra sign is the hallmark of a half-integral level. 
Where it sits is a matter of convention. 
It may be assigned to the phase of the matter 1-loop determinant 
or, as in (\ref{CS_classical_sign}), absorbed into the classical Chern-Simons factor, 
the difference between the two being a half-integral contact term \cite{Closset:2012vp,Closset:2019hyt}. 
What is not a convention is the presence of a sign factor 
that no choice of counterterm can remove. By definition this is a \emph{spin} theory, requiring a choice of spin structure on the 3-manifold \cite{Closset:2018ghr}. 

It is worth being precise about the status of this sign. 
Under the Weyl reflection $s_i \to s_i^{-1}$, $m_i \to -m_i$ 
the two factors in (\ref{CS_classical_sign}) are \emph{separately} invariant since $(-1)^{-m_i}=(-1)^{m_i}$. 
Likewise the convergence of the sum over GNO fluxes is insensitive to a sign. 
Dropping $(-1)^{m_i}$ therefore does not violate Weyl invariance. 
What one obtains that way is nevertheless not the index of some other consistent theory. 
The half-integral level is forced by the parity anomaly, and the spin structure comes with it. 
The sign is what that choice leaves behind in the localization formula. 
Removing it amounts to declaring the theory to be defined on oriented 3-manifolds carrying no spin structure, which is exactly what the anomaly forbids. 
We should be careful with the word anomaly here. 
The obstruction is to parity invariance at integral level, not to the consistency of the theory at half-integral level, 
which is perfectly well defined once a spin structure is chosen. 
Note also that the sign cannot be absorbed into $s_i \to -s_i$ 
because the 1-loop contribution of $Q$ contains odd powers $s_i^{\pm 1}$.

\subsection{Vacuum counts and the choice of level}
\label{sec_electric_vacua}
We now show that the range of $\ell \ge 1$ is determined by the counting of vacua.
First recall that for pure $\mathcal{N}=2$ $USp(2n)_{k_{\mathrm{eff}}}$ Chern-Simons theory 
the number of supersymmetric vacua is the number of integrable highest weights of $\widehat{\mathfrak{usp}}(2n)$ at level $\ell'= |k_{\mathrm{eff}}| - h^\vee$, with $h^\vee = n+1$ 
\cite{Witten:1999ds,Intriligator:2013lca,Closset:2023jiq}
\begin{align}
\label{vacuum_count}
\# \text{vacua}
= \binom{\ell'+n}{n}
= \binom{|k_{\mathrm{eff}}|-1}{n},
\qquad
|k_{\mathrm{eff}}| \ge n+1\,,
\end{align}
and supersymmetry is broken for $|k_{\mathrm{eff}}| < h^\vee = n+1$. 
Here $h^\vee$ is the dual Coxeter number. 
Note that the distinction between Coxeter and dual Coxeter matters for $USp(2n)$, where $h^\vee = n+1$ while the Coxeter number is $h = 2n$, unlike 
for the $SU(N)$ case usually quoted, for which the two coincide.

Turning on a real mass for $Q$ through $U(1)_a$ and integrating it out shifts the level by
$\pm\tfrac12$, so that writing $k=n+\tfrac12+\ell$,
\begin{align}
k_{\mathrm{eff}}^{\pm} 
= k \pm \frac{1}{2}
{}= \begin{cases}
n+1+\ell\,, &\cr 
n+\ell\,, &
\end{cases}
\end{align}
and so by \eqref{vacuum_count} the two deformations leave
\begin{align}
\label{two_vacuum_counts}
\# \text{vacua}\big|_{k_{\mathrm{eff}}=n+\ell}  &{}= \binom{n+\ell-1}{n},
\nonumber \\
\# \text{vacua}\big|_{k_{\mathrm{eff}}=n+1+\ell}  &{}= \binom{n+\ell}{n}.
\end{align}
Both counts are positive precisely when $\ell\ge1$, since $k_{\mathrm{eff}}^{-}=n+\ell$
falls below $h^{\vee}$ at $\ell=0$ and supersymmetry is then broken on one side.  At $\ell=1$ the
pair is $(1,\,n+1)$, and it is a precise detail that any proposed dual must reproduce.  We note
in passing that $\ell$ has an intrinsic meaning on this side of the duality.  Since
$h^{\vee}=n+1$, the affine level reached under the deformation with $k_{\mathrm{eff}}^{+}$ is
$\ell'=\ell$ itself, so the parameter labelling the family is the level of
$\widehat{\mathfrak{usp}}(2n)$ that the theory flows to under one of the two real mass
deformations, the other giving $\ell'=\ell-1$.

The range of $\ell$ is also seen from the index. Expanding \eqref{ind_USp2n_Nf1_full} at
half-integral $k$, we find that it vanishes identically for $\ell<0$, equals $1$ identically at
$\ell=0$, and is a non-trivial series for $\ell\ge1$, as we have checked at $n=1$ and $n=2$.  The
first regime is the supersymmetry breaking just described and the second is a theory with no
local operators at all, so $\ell\ge1$ is not a choice of level but the statement that the theory
is non-trivial.  What distinguishes the three regimes is invisible in the Gaiotto-Witten
description of section~\ref{sec_electric_susy}, in which the $\mathcal{N}=4$ supersymmetry is
present at every level.  The dual assigned to each $\ell$ in section~\ref{sec_matching} accounts
for the pattern, since the Abelian theory to which it is matched exists only for $\ell\ge1$.

\subsection{Global symmetries and the moduli space}
\label{sec_electric_moduli}
The global symmetry is
\begin{align}
U(1)_a \times U(1)_R\,,
\end{align}
with no non-Abelian factor, because there is only one fundamental. 
Since $\pi_1\big(USp(2n)\big)$ is trivial, there is no topological symmetry $U(1)_J$.
Nor is there a one-form symmetry, since $Q$ is charged under the center
$\Zb_2\subset USp(2n)$, 
so the would-be $\Zb_2$ one-form symmetry is completely broken and there is no discrete gauging ambiguity to keep track of.

The only candidate gauge invariant built from $Q$ alone is the meson
\begin{align}
\label{meson_vanishes}
M=J_{ab}Q^a Q^b\equiv 0,
\end{align}
which vanishes identically because a symplectic form $J_{ab}$, $a,b=1,\cdots,2n$ is antisymmetric 
while the components of $Q$ commute. 
There is therefore no Higgs branch.  
This should be distinguished from the \emph{real} moment map $\mu^{A} = Q^{\dagger} T^{A} Q$, $T^{A}\in\mathfrak{usp}(2n)$,
which need not vanish for $Q\neq 0$.  In the presence of a CS term the vacuum equations read
\begin{align}
\label{electric_vacuum_eqs}
\frac{k}{2\pi}\,\sigma^{A} + Q^{\dagger} T^{A} Q = 0, \qquad \sigma Q = 0,
\end{align}
where $\sigma = \sigma^{A}T^{A}$ is the real scalar of the $\mathcal{N}=2$ vector multiplet, 
valued in the adjoint of $\mathfrak{usp}(2n)$.
Contracting the first with $\sigma^{A}$ and using the second gives $\frac{k}{2\pi}|\sigma|^{2}=0$, hence $\sigma=0$ for $k\neq 0$. 
The remaining condition $\mu^{A}=0$ then forces $Q=0$ so the only solution is the origin. 

Since $Q$ contributes $2n$ fermions of unit $U(1)_a$ charge, the axial contact term is integrally quantized,
\begin{align}
\label{electric_contact}
\delta k_{aa} = \tfrac12\cdot 2n\cdot 1^2 = n \in \Zb\,,
\end{align}
so that $k_{aa}\in\Zb$. 
The mixed and gravitational contact terms $k_{aR},\,k_{RR},\,k_{gg}$ are additional data, 
not determined by the Lagrangian alone.

\subsection{Monopole operators}
\label{sec_monopoles}
Let $m = (m_1,\dots,m_n) \in \Zb^n$ be the GNO flux, valued in the coroot lattice of $USp(2n)$. 
The positive roots are $2e_i$ (long) and $e_i \pm e_j$, $i<j$ (short). 
The weights of the representation $\mathbf{2n}$ are $\pm e_i$. 
The Weyl group is the group of signed permutations, $W = \Zb_2^{\,n} \rtimes S_n$ of order $2^n n!$, 
so that fluxes may be restricted to the fundamental domain
\begin{align}
m_1 \ge m_2 \ge \dots \ge m_n \ge 0.
\end{align}

The quantum numbers of the bare monopole $V_m$ are
\begin{align}
R[V_m]&=
\tfrac{1}{2}(1-r_Q)\sum_{\rho \in \mathbf{2n}} |\rho(m)|-\sum_{\alpha > 0} |\alpha(m)|
\nonumber\\
&= (1-r_Q)\sum_i |m_i|-2\sum_i |m_i|-\sum_{i<j}\big( |m_i+m_j| + |m_i - m_j| \big),
\end{align}
\begin{align}
a[V_m] &= -\sum_i |m_i|,
\end{align}
and
\begin{align}
\label{monopole_gauge_charge}
\text{gauge charge under $i$-th Cartan } U(1)_i &: \quad 2k\, m_i
{}\;=\;(2n+1+2\ell)\,m_i \,.
\end{align}

Equation (\ref{monopole_gauge_charge}) deserves comment. 
In general the induced gauge charge of a bare monopole receives, besides the CS contribution $2k\,m_i$, a 1-loop piece
\begin{align}
-\frac12\sum_{\rho\in \mathbf{2n}}\rho_i\,|\rho(m)|+\frac12\sum_{\alpha\in\Delta}\alpha_i\,|\alpha(m)|.
\end{align}
Here both sums vanish identically. 
The roots come in pairs $\pm\alpha$, and the fundamental of $USp(2n)$ is pseudo-real, 
so its weights come in pairs $\pm e_i$. 
Hence the charge is \emph{purely} the Chern-Simons contribution, and it is a linear function of $m$.

Crucially, because of the CS term, the bare monopole is \emph{not} gauge invariant. 
It carries $2k\,m_i$ units of electric charge under the $i$-th Cartan. 
A genuine local operator must be dressed by modes of $Q$ and by $W$-boson/gaugino modes carrying the compensating charge.

\paragraph{Example}
For the minimal flux $m = (1,0,\dots,0)$ one finds
\begin{align}
R[V_m] = 1-r_Q-2n,
\qquad
a[V_m] = -1,
\qquad
\text{charge under } U(1)_1 = 2k=2n+1+2\ell,
\end{align}
so the bare operator is not gauge invariant and has a large negative R-charge. 
Dressing with $p$ modes of $Q$ of weight $-e_1$ adds $(-p,\,+p,\,+p\,r_Q)$ to the (gauge, $a$, $R$)-charges. 
Gauge invariance requires $p = 2k$ if only these modes are used, giving
\begin{align}
R = 1-2n+(2k-1)\, r_Q\,,
\qquad
a = 2k-1\,.
\end{align}
At $\ell=1$ these are $R=1-2n+(2n+2)r_Q$ and $a=2n+2$.

\subsection{Constraints on the dual}
\label{sec_dual_constraints}
Before constructing the dual we record two constraints that strongly limit the possibilities.
We state them at general $\ell$ and then specialize to $\ell=1$, which is the case we
construct in section~\ref{sec_tadpole_theory}.

\paragraph{(i) Vacuum counting}
Any dual must reproduce the pair $\big(\binom{n+\ell-1}{n},\binom{n+\ell}{n}\big)$ of
\eqref{two_vacuum_counts} under the two real mass deformations. 
We begin with the simplest guess, an Abelian theory with a \emph{single} gauge node. 
The answer we eventually propose has $n(2\ell - 1)$ nodes, so this is a warm-up rather than a derivation. 
However, it is worth going through because it produces the seed case $n = \ell = 1$, 
and because the precise way it fails for general $n$ is what forces the multi-node structure of the tadpole theory for $\ell = 1$ described in section~\ref{sec_tadpole_theory}. 
Consider then a $U(1)_{k'}$ gauge theory with chirals of charges $q_1,\dots,q_N$. 
The pure $\mathcal{N}=2$ $U(1)_{k}$ theory has $|k|$ vacua, and the mass deformations
(sending the mass of all chirals to either plus or minus infinity)
give $k'^{\pm}_{\mathrm{eff}} = k' \pm \tfrac12\sum_i q_i^2$. Matching the two counts
requires
\begin{align}
\label{dual_condition}
\sum_{i=1}^{N} q_i^{\,2} = \binom{n+\ell-1}{n-1}\,,\qquad
k' = \frac12\left[\binom{n+\ell}{n}+\binom{n+\ell-1}{n}\right],
\end{align}
which at $\ell=1$ is $\sum_i q_i^{2}=n$ and $k'=\frac{n+2}{2}$, in the case
$(k'^{-}_{\mathrm{eff}}, k'^{+}_{\mathrm{eff}}) = \big(\binom{n+\ell-1}{n},
\binom{n+\ell}{n}\big)$.\footnote{Writing $A$ and $B$ for the two entries of \eqref{two_vacuum_counts}, the other
sign choices give either the same up to the sign of $k'$, or the alternative
$\sum_i q_i^{\,2}=A+B$ with $k'=\tfrac12(B-A)$, which at $\ell=1$ reads
$\sum_i q_i^{\,2}=n+2$ with $k'=\pm\tfrac{n}{2}$.  There the special role passes from $n$ to
$n+2$ being a sum of squares, but the conclusion below is unchanged.}
One checks that the $U(1)$ parity anomaly is then automatic, since
$k' - \tfrac12\sum_i q_i^2 = 1 \in \Zb$ for every $n$.

Only the two uniform deformations have been used here, and it is worth saying why the others
are not available.  For $N>1$ one may also send some masses to plus infinity and others to
minus infinity, giving $k'_{\mathrm{eff}}=k'+\tfrac12\sum_i\epsilon_i q_i^{2}$ with
$\epsilon_i=\pm1$.  Such deformations are turned on through the spurious flavor symmetries
$U(1)^{N-1}$ acting on the chirals, which have no counterpart in the symplectic theory, whose
only flavor symmetry is the single $U(1)_a$.  Constraint (ii) below is precisely the requirement
that they be absent, and the mismatch of deformation spaces at $N>1$ is another symptom of the
same tension.  In the Abelian theories we eventually construct they are indeed absent, though
not by being lifted.  The charge matrix there is $\mathbf Q=\mathbb{I}$, so the flavor
symmetries are gauged away entirely, and what the monopole superpotential removes is the
topological symmetry, leaving the single factor that maps to $U(1)_a$.

\paragraph{(ii) Global symmetry}
The symplectic theory has only $U(1)_a\times U(1)_R$ global symmetry. 
On the Abelian side the topological symmetry $U(1)_J$ must map to $U(1)_a$, 
and \emph{no further} flavor symmetry may survive. 
A theory with $N$ chirals has at least $U(1)^{N}$ acting on them (or potentially a non-Abelian flavor symmetry $U(N)$ or a subgroup of $U(N)$), one combination of which is gauged, leaving at least $U(1)^{N-1}$ spurious factors. 
Hence either $N=1$, or a superpotential must lift the extra symmetries. 

If we assume the case $N=1$ in \eqref{dual_condition}, we have the proposal
\begin{align}
\label{N1_proposal}
USp(2n)_{n+\frac32} \ \text{with one } \mathbf{2n}
\quad\leftrightarrow\quad
U(1)_{\frac{n+2}{2}} \ \text{with one chiral of charge } q,\quad q^2=n\,,
\end{align}
which is already restricted to the case that $n$ is a perfect square. For $n=1$ this reads
\begin{align}
\label{n1_duality}
SU(2)_{\frac52}\ \text{with one doublet}
\quad\leftrightarrow \quad
U(1)_{\frac32}\ \text{with one charge-$1$ chiral},
\end{align}
with vacuum counts $(|{\tfrac32} - \tfrac12|, |{\tfrac32} + \tfrac12|) = (1, 2)$ on both sides. As we will explain, this is indeed the correct proposal, but for $n > 1$ we need a different generalization.

For $n$ not a perfect square the single node fails. 
Constraint (i) asks for $\sum_i q_i^2 = n$, which wants several chirals, 
while constraint (ii) allows only one unless something lifts the spurious symmetries. 
It is not clear how to achieve this, and anyway it would be surprising that the structure of the Abelian theory, in particular the number of chirals and their charges, was so sensitive to $n$ being a perfect square or not.
Enlarging the gauge group to $U(1)^r$ resolves the tension. 
The charges are then a matrix $q_{Ii}$ obeying the matrix analogue of (\ref{dual_condition}), 
and the extra symmetries can be removed, in a way that turns out to combine both options,
part of the work being done by the gauging itself and part by a superpotential built from
monopole operators. 
We construct that theory in the next section, where the two constraints are satisifed with $r=n$ nodes, one charge-one chiral per node and $r-1$ superpotential terms.
The two constraints are necessary but not sufficient once $\ell>1$, as we note at the end of
section~\ref{sec_electric_susy}, once the supersymmetry of \eqref{USp2n_Nf1_charges} has been
pinned down.

\subsection{The full index}
\label{sec_index}
For the theory (\ref{USp2n_Nf1_charges}) the full index reads
\begin{align}
\label{ind_USp2n_Nf1_full}
\mathcal{I}^{USp(2n)}
&=\frac{1}{2^n\, n!} \sum_{m_i \in \Zb} \oint 
\left( \prod_{i=1}^{n} \frac{ds_i}{2\pi i\, s_i}\, (-s_i)^{2k m_i} \right) 
q^{-\sum_i |m_i| \,-\, \frac12 \sum_{i<j} |m_i \pm m_j|}
\nonumber \\
& \left( \prod_i \big(1 - q^{|m_i|} s_i^{\pm 2}\big) \right)
\prod_{i<j}
\big(1 - q^{\frac{|m_i \pm m_j|}{2}} s_i s_j^{\pm}\big)
\big(1 - q^{\frac{|-m_i \pm m_j|}{2}} s_i^{-1} s_j^{\pm}\big)
\nonumber \\
&\big(q^{\frac{1-r_Q}{2}} a^{-1}\big)^{\sum_i |m_i|}
 \prod_i
 \frac{(q^{1 - \frac{r_Q}{2} + \frac{|m_i|}{2}} a^{-1} s_i^{\mp};q)_{\infty}}
 {(q^{\frac{r_Q}{2} + \frac{|m_i|}{2}} a\, s_i^{\pm};q)_{\infty}}.
\end{align}
Here $(x;q)_{\infty} = \prod_{\ell \ge 0}(1 - x q^{\ell})$, and a superscript
$\pm$ (or $\mp$) inside a product means that both signs are taken and the
corresponding factors are multiplied.

The prefactor $1/(2^n n!)$ is the inverse order of the Weyl group. 
Combined with the unrestricted sum over fluxes it implements the projection onto gauge singlets. 
Since $USp(2n)$ is simply connected, that sum runs over the full coroot lattice, $m \in \Zb^n$. 
The classical Chern-Simons contribution $(-s_i)^{2k m_i}$ is normalized with the trace in $\mathbf{2n}$, and carries the sign discussed above. 
The two remaining flux-dependent prefactors encode the quantum numbers of the bare monopole. 
The power of $q$,
\begin{align}
q^{-\sum_i|m_i| - \frac12\sum_{i<j}|m_i \pm m_j|}
= q^{-\frac12 \sum_{\alpha>0}|\alpha(m)|}\,,
\end{align}
is a contribution from the vector multiplet, 
while $\big(q^{\frac{1-r_Q}{2}} a^{-1}\big)^{\sum_i |m_i|}$ supplies the matter contribution to the R-charge together with the axial charge $a[V_m] = -\sum_i |m_i|$.
The 1-loop determinants then follow the usual pattern of one factor per
weight. For the vector multiplet each root $\alpha$ contributes
$\big(1 - q^{|\alpha(m)|/2} s^{\alpha}\big)$, so that the long roots $\pm 2 e_i$ give
the factors $\big(1 - q^{|m_i|} s_i^{\pm 2}\big)$, and the short roots
$\pm e_i \pm e_j$ give $q^{|m_i \pm m_j|/2}$ with the corresponding fugacity.
The ratio of $q$-Pochhammer symbols is the one-loop contribution of the fundamental chiral,
the two signs corresponding to the weights $\pm e_i$ of $\mathbf{2n}$.

Note that if we define
\begin{align}
\label{ahat_electric}
\hat{a}\equiv a q^{\frac{r_Q}{2}}, 
\end{align}
then the matter part of (\ref{ind_USp2n_Nf1_full}) becomes
\begin{align}
\big(q^{\frac12}\hat{a}^{-1}\big)^{\sum_i |m_i|}
\prod_i \frac{(q^{1+\frac{|m_i|}{2}}\hat{a}^{-1} s_i^{\mp};q)_{\infty}}{(q^{\frac{|m_i|}{2}}\hat{a}\, s_i^{\pm};q)_{\infty}},
\end{align}
so that \emph{all} dependence on the trial R-charge is absorbed into $\hat{a}$.
This has two consequences:
\begin{enumerate}
  \item The superconformal R-charge must be determined independently, e.g.\ by $F$-maximization \cite{Jafferis:2010un}.
        For the theory at hand it is fixed instead by the enlarged supersymmetry of
        section~\ref{sec_electric_susy}, at $r_Q=\tfrac12$.
  \item For computational purposes one may set $r_Q$ to any convenient value
        and restore it at the end by rescaling $a$ by an appropriate power of $q$. For the assignment $r_Q=0$ the purely $q$-graded expansion (i.e.\ the unflavored index where $a = 1$) degenerates
        and the refined truncation by $U(1)_a$ charge described in
        section~\ref{sec_index_matching} is required.
\end{enumerate}
The unit-circle contour separates the poles of
$1/(q^{\frac{|m|}{2}}\hat{a} s^{\pm};q)_{\infty}$ from their images provided
\begin{align}
\label{electric_contour}
  |q| < |\hat{a}| < 1.
\end{align}
This is the standard regime in which the index is a well-defined power series.

Take $m = (1,0,\dots,0)$.  
The monopole contribution is
\begin{align}
q^{-|m_1| - \frac12 \sum_{j>1}\big(|1+0| + |1-0|\big)}=q^{-1 - (n-1)}=q^{-n}\,,
\end{align}
and the flux axial factor is $q^{1/2}\hat{a}^{-1}$.  
One needs $s_1^{-2k}$ to cancel the CS phase.  
The sources of negative $s_1$ charge are:
\begin{itemize}
  \item the long-root factor $q\, s_1^{-2}$: two units at cost $q$;
  \item for each $j = 2,\dots,n$, the pair $q^{1/2} s_1^{-1} s_j$ and
        $q^{1/2} s_1^{-1} s_j^{-1}$: two units of $s_1^{-1}$ at cost $q$, with the
        $s_j$ charges cancelling between the two;
  \item modes of $Q$: one unit each at cost $q^{1/2}\hat{a}$.
\end{itemize}
Using the vector multiplet maximally gives $s_1^{-2n}$ at cost $q^{n}$, and the
remaining $2k-2n=2\ell+1$ units must come from $Q$, at cost
$q^{\ell+\frac12}\hat{a}^{2\ell+1}$.  
The total is
\begin{align}
\label{electric_m1_min}
q^{-n} \cdot q^{\frac12}\hat{a}^{-1} \cdot q^{n} \cdot\big(q^{\frac12}\hat{a}\big)^{2\ell+1}
&=q^{1+\ell}\,\hat{a}^{2\ell}.
\end{align}
At the other extreme, taking all $2k$ units from $Q$ and none from the vector
multiplet gives
\begin{align}
\label{electric_m1_max}
q^{-n} \cdot q^{\frac12}\hat{a}^{-1} \cdot \big(q^{\frac12}\hat{a}\big)^{2k}
=q^{1+\ell}\hat{a}^{2n+2\ell}.
\end{align}
So, we see that trading vector multiplet dressing for matter dressing does not change the $q$-degree.  
Both extremes sit at $q^{1+\ell}$, which is $q^{2}$ at $\ell=1$, so the rank $n$ is
invisible at this order for every level.
\footnote{We have derived this here for the $m=(1,0,\dots,0)$ sector but we believe this is true for all sectors, and numerical expansion of the index for low values of $n$ supports this.}
We stress that \eqref{electric_m1_min} and \eqref{electric_m1_max} bound only the
range of $\hat a$-charges that the $m=(1,0,\dots,0)$ sector can contribute at order
$q^{1+\ell}$.  Whether a given contribution survives in the index is determined by the full
expansion, in which cancellations against other contributions of the same weight can
and do occur, as we will see in section~\ref{sec_index_matching}.
\subsection{Supersymmetry enhancement}
\label{sec_electric_susy}

The consequences of pseudo-reality noted below \eqref{USp2n_Nf1_charges} are all restrictions.
Stated positively, the same fact makes something available. 
Because $\mathbf{2n}$ is pseudo-real, a single chiral multiplet in it is a \emph{half}-hypermultiplet, 
and \eqref{USp2n_Nf1_charges} is then of the form to which the construction of Gaiotto and Witten \cite{Gaiotto:2008sd} applies (see also \cite{Hosomichi:2008jd}). 
The input of that construction is a symplectic space $V$, that is the representation in which
the matter transforms, equipped with a non-degenerate antisymmetric form $\omega_{AB}$, together
with a gauge group embedded in the group $USp(V)$ of transformations preserving $\omega$, and a
level.  Lowering the second index of the generators with $\omega$
\begin{align}
\label{t_lowered}
t^{m}_{AB}=\omega_{AC}\,(t^{m})^{C}{}_{B}\,,
\end{align}
gives objects symmetric in $AB$, since preserving $\omega$ is the statement that the
antisymmetric part vanishes, so the $t^{m}_{AB}$ live in $\mathrm{Sym}^{2}V$.
We take $V=\mathbf{2n}$ with the gauge group all of $USp(2n)$, for which
$\mathrm{Sym}^{2}V\cong\mathfrak{usp}(2n)$ and the embedding is the identity. 
The condition under which the resulting Lagrangian has $\mathcal{N}=4$ supersymmetry is the fundamental identity 
$k_{mn}t^{m}_{(AB}t^{n}_{C)D}=0$ with $A,B,C$ symmetrized over cyclic permutations. 
For a simple factor $k_{mn}\propto\delta_{mn}$, and the $t^{m}_{AB}$ are a basis of
$\mathrm{Sym}^{2}V$, 
so that $\sum_{m}t^{m}_{AB}t^{m}_{CD}\propto\omega_{AC}\omega_{BD}+\omega_{AD}\omega_{BC}$ and the three cyclic terms cancel in pairs,
\begin{align}
(\omega_{AC}+\omega_{CA})\,\omega_{BD}
+(\omega_{BC}+\omega_{CB})\,\omega_{AD}
+(\omega_{AB}+\omega_{BA})\,\omega_{CD}=0\,,
\nonumber
\end{align}
which is the antisymmetry of $\omega$ and nothing more.  The Lagrangian the construction then
produces is ours.  Its holomorphic $\mathcal{N}=2$ superpotential
$k_{mn}\mu^{m}_{\mathbb C}\mu^{n}_{\mathbb C}$, with $\mu^{m}_{\mathbb C}=t^{m}_{AB}Q^{A}Q^{B}$,
is proportional to $(J^{ab}Q_aQ_b)^{2}$ and vanishes by \eqref{meson_vanishes}, in agreement with
$\mathcal{W}=0$, and its sextic scalar potential is the one generated here by integrating out
$\sigma$.  We conclude that \eqref{USp2n_Nf1_charges} has $\mathcal{N}=4$ supersymmetry,
manifestly and at every level, in the sense that the eight supercharges act on the fields of
\eqref{USp2n_Nf1_charges} itself and no infrared dynamics is invoked.  The $\mathcal{N}=2$
R-symmetry is then the one contained in $SU(2)_{H}\times SU(2)_{C}$, so that $r_Q=\tfrac12$ is
fixed by the algebra and not by F-maximization.  The eight supercharges are symmetries of the
Lagrangian and are not anomalous, so this is a statement about the quantum theory, and one that
requires no infrared input, the supercharges being visible already in the ultraviolet.

We stress that this is an application of the construction and not an appeal to a
classification.  The superalgebra associated with $\big(\mathfrak{usp}(2n),\mathbf{2n}\big)$ is
$\mathfrak{osp}(1|2n)$, whose orthogonal part is trivial, so the case is the degenerate $N=1$
member of the $OSp(N|M)$ family and is not among the examples written out explicitly in
\cite{Gaiotto:2008sd,Hosomichi:2008jd}.  Nothing degenerates in the field theory, however,
because the input is the pair $(G,V)$ and here only $USp(2n)$ is gauged, so that the usual
requirement of equal and opposite levels on two factors has nothing to constrain and the level
is free.  The degeneracy is a feature of the labelling by a superalgebra, not of the theory.
One caveat does remain, that
the half-integral quantization forced by the parity anomaly is a quantum effect not treated in
those references.

That construction fixes the amount of supersymmetry at $\mathcal{N}=4$ exactly, and no
more.  Orthosymplectic Chern-Simons matter theories can enhance further,
\cite{Hosomichi:2008jb} finding $\mathcal{N}=5$ for $O(N)\times USp(2M)$ and $\mathcal{N}=6$ for
$O(2)\times USp(2M)$.  This does not happen for \eqref{USp2n_Nf1_charges}, and the reason is a
counting of matter.  The
supercharges sit in the vector of the R-symmetry $SO(\mathcal{N})$, so at $\mathcal{N}=4$ one has
$SO(4)=SU(2)_H\times SU(2)_C$ and the matter splits into hypermultiplets acted on by $SU(2)_H$
and twisted hypermultiplets acted on by $SU(2)_C$, in representations $R_1$ and $R_2$.  At
$\mathcal{N}=5$ the scalars must fill the spinor $\mathbf{4}$ of $SO(5)$, which decomposes under
$SO(4)$ as $(\mathbf{2},\mathbf{1})\oplus(\mathbf{1},\mathbf{2})$, so a hypermultiplet and a
twisted hypermultiplet in the \emph{same} gauge representation are needed.  This is the
condition $R_1=R_2$ of \cite{deMedeiros:2009eq}, and it means that the $\mathcal{N}=5$ theories
carry twice the matter of the corresponding $\mathcal{N}=4$ ones.  The pattern is familiar from
the unitary case, where the Gaiotto-Witten theory for $\mathfrak{su}(N|M)$ has one bifundamental
hypermultiplet and is $\mathcal{N}=4$, while the ABJM theory has two and is $\mathcal{N}=6$.
Our theory has a single chiral multiplet, which is half of a hypermultiplet and no twisted
hypermultiplet at all, so $R_2=0$.  The $\mathcal{N}=5$ theory at $N=1$ would be $USp(2n)$ with
two fundamental chirals, a different theory.

The index of section~\ref{sec_index} confirms this, and it is worth spelling out how, since
we use the same reasoning in section~\ref{sec_index_matching}.  Supersymmetry enhancement is visible in the
$\mathcal{N}=2$ index because the currents responsible for it sit in short multiplets whose
quantum numbers are fixed by the superconformal algebra, so that each such multiplet contributes
to a definite coefficient of \eqref{ind_USp2n_Nf1_full}.  This is the diagnostic used in \cite{Gang:2018huc} and in
\cite{Creutzig:2024ljv}, and we follow it here.  Two entries of the dictionary are
relevant.  A conserved flavor current multiplet, whose bottom component is a real scalar of
dimension one and vanishing R-charge, contributes $-q$, weighted by the flavor fugacity of the
current.  This is the standard way in which an enlarged global symmetry is read off from the
index, and it is how the single $-q$ of our expansions is identified below
\eqref{rank1_expansion_rQhalf} with the one $U(1)$ current of the theory.  The multiplet
containing the supercharges by which $\mathcal{N}=2$ is enhanced contributes instead at
$q^{3/2}$, weighted by the fugacity of the $U(1)$ that becomes part of the enlarged R-symmetry,
and its appearance is the usual index diagnostic of enhancement
\cite{Gang:2018huc,Creutzig:2024ljv}.  For the general structure of superconformal multiplets
see \cite{Cordova:2016emh}, whose emphasis is on $d>3$ but whose tables of current multiplets
are the ones invoked in this context, and \cite{Evtikhiev:2017heo} for when such an index
argument can fail.

Applied to $\mathcal{N}=4$, whose R-symmetry is $SO(4)=SU(2)_{H}\times SU(2)_{C}$, this
predicts a single flavor-current-like contribution $-q$ from the axial combination of the two
Cartans, which on this side is $U(1)_a$, together with a pair of enhancing supercurrents of
$U(1)_a$ charge $\pm2$ at $q^{3/2}$.  That is exactly the opening
$1-q+(a^{2}+a^{-2})q^{3/2}$ of the expansions of \eqref{ind_USp2n_Nf1_full}, and the two
contributions at $q^{3/2}$ are the
gauge invariants $J_{ab}Q^{a}\partial Q^{b}$ and $J^{ab}\bar\psi_{a}\bar\psi_{b}$, the meson
itself vanishing by \eqref{meson_vanishes}.

An enhancement to $\mathcal{N}=5$ would add to both.  The R-symmetry adjoint grows from
$\mathbf{6}=(\mathbf{3},\mathbf{1})\oplus(\mathbf{1},\mathbf{3})$ of $SO(4)$ to
$\mathbf{10}=\mathbf{6}\oplus\mathbf{4}$ of $SO(5)$, and writing the charges of the extra
$\mathbf{4}=(\mathbf{2},\mathbf{2})$ as $(R,A)=(J_{H}+J_{C},J_{H}-J_{C})$ gives $(\pm1,0)$ and
$(0,\pm1)$.  The last two are conserved currents of vanishing R-charge and $U(1)_a$ charge
$\pm2$, so by the first entry of the dictionary they contribute at order $q$ weighted by
$a^{\pm2}$, and an $\mathcal{N}=5$ theory would begin $1-(1+a^{2}+a^{-2})q+\dots$ rather than
$1-q+\dots$.  The supercharges likewise grow from the $\mathbf{4}$ to the $\mathbf{5}$ of the
R-symmetry, and since $\mathbf{5}\to\mathbf{4}\oplus\mathbf{1}$ the extra one is an $SO(4)$
singlet, so by the second entry it would contribute at $q^{3/2}$ with no $a$ dependence.
Neither term is present.  The coefficient of $q$ is exactly $-1$ and that of $a^{0}q^{3/2}$
vanishes, at $n=1$, $2$ and $3$ alike.

One caveat applies to arguments of this kind, and we state it once here.  The index is a
signed count, so a contribution can in principle be cancelled by another short multiplet with
the same quantum numbers, and an absence is conclusive only when such multiplets are excluded.
This is the reason the enhancement to $\mathcal{N}=4$ was checked semiclassically in
\cite{Gang:2018huc} rather than read off from the index alone.  Here
the conclusion we draw is the negative one, that $\mathcal{N}=5$ does not occur, and for that
the argument of the previous paragraph is a consistency check on the counting of matter given
above rather than an independent proof.

This bears on the constraints of section~\ref{sec_dual_constraints}, which are necessary
but not sufficient once $\ell>1$, so that the tightness of the argument there is special to
$\ell=1$.  Take $n=1$ and $\ell=2$, where \eqref{dual_condition} asks for
$\sum_i q_i^{2}=\binom{2}{0}=1$ and $k'=\tfrac52$.  A single node with one chiral of unit charge
meets this, and having one chiral it carries no spurious flavor symmetry either, so it satisfies
both constraints.  Its index, however, is
\begin{align}
1-q+\big(\zeta^{-1}-2\big)q^{2}+\dots\,,
\nonumber
\end{align}
in terms of the fugacity $\zeta$ for its own topological symmetry, with no term at order
$q^{3/2}$ at all, whereas the symplectic theory at $k=\tfrac72$ has the pair
$(a^{2}+a^{-2})q^{3/2}$ that signals the enhancement.  The comparison needs no dictionary
between the two fugacities, since what distinguishes the two series is the presence or absence
of the order $q^{3/2}$ altogether.  The two theories are therefore
different.  What the constraints do not encode is that the dual must itself be a rank-$0$ theory
with $\mathcal{N}=4$ supersymmetry, and it is that requirement which the three conditions of
section~\ref{sec_magnetic_level_choice} implement at $\ell=1$ and which the construction of
\cite{Creutzig:2024ljv} implements in general.

\subsection{A relation among the symplectic theories}
\label{sec_level_pair}

Before turning to the dual we record a property of the indices \eqref{ind_USp2n_Nf1_full}
that involves only the theories of this section.  Write the level as $k=n+\tfrac12+\ell$ as
above.  The combination $n+\ell$ is what fixes the level, and it is symmetric under exchanging
the rank with the level parameter, so the two theories
\begin{align}
USp(2n)_{\,k}\quad\text{and}\quad USp(2\ell)_{\,k}\,,
\qquad k=n+\ell+\tfrac12\,,
\nonumber
\end{align}
each with a single fundamental, sit at one and the same Chern-Simons level and carry the same
amount of matter, differing only in the rank of the gauge group.  We find that their indices
agree up to charge conjugation,
\begin{align}
\label{level_pair}
\mathcal{I}^{USp(2n)}\big(q;\hat y\big)\Big|_{k=n+\ell+\frac12}
=\mathcal{I}^{USp(2\ell)}\big(q;\hat y^{-1}\big)\Big|_{k=n+\ell+\frac12}\,,
\end{align}
where $\hat y=q^{-1/2}\hat a^{2}$, with $\hat a$ the combination \eqref{ahat_electric} in which
the R-charge dependence is carried, so that the inversion acts on the fugacity of
\eqref{ind_USp2n_Nf1_full} as $a^{2}\to a^{-2}q^{1-2r_Q}$.  The notation $\hat y$ anticipates
section~\ref{sec:tadpole}, where the same symbol denotes the corresponding combination on the
Abelian side, the two being identified by the dictionary of section~\ref{sec_matching}.

The first non-trivial case is $(n,\ell)=(1,2)$, where both sides are at $k=\tfrac72$, and
there we find $\mathcal{I}^{USp(2)}(q;\hat{y})=\mathcal{I}^{USp(4)}(q;\hat{y}^{-1})$ through order $q^{5}$ at
$r_{Q}=\tfrac12$ and through order $q^{3}$ at $r_{Q}=\tfrac13$.  Note that this relates a rank-one
to a rank-two gauge theory, so no field redefinition can be responsible.

The self-conjugate case $n=\ell$ of \eqref{level_pair} says that the index of
$USp(2n)_{2n+\frac12}$ is invariant under $\hat y\to\hat y^{-1}$ by itself, and that no other
member of the family is.  This too we find, at $r_{Q}=\tfrac12$,
\begin{align}
\begin{array}{c|ccc}
 & \ell=1 & \ell=2 & \ell=3 \\ \hline
n=1 & \text{symmetric} & q^{5/2} & q^{5/2} \\
n=2 & q^{5/2} & \text{symmetric} & q^{4}
\end{array}
\nonumber
\end{align}
where an entry records the order at which the invariance first fails.  Only the diagonal is
symmetric.  The theories treated in the rest of this paper sit at $\ell=1$, which meets the
diagonal only at $n=1$, and this is the origin of the charge-conjugation symmetry of the rank-one
index that we will meet in section~\ref{sec_matching_rank1}.

We arrived at \eqref{level_pair} from the duality proposed in
section~\ref{sec_matching} together with the mirror symmetry of \cite{Creutzig:2024ljv}, as
explained there.  It is, however, nothing other than the Giveon-Kutasov duality
\cite{Giveon:2008zn,Kapustin:2011vz,Willett:2011gp} continued to a single fundamental.  For
$USp(2N)_{k}$ with $2N_{f}$ fundamentals that duality gives a dual $USp(2\tilde N)_{-k}$ with
\begin{align}
\label{GK_rank}
\tilde N = N_{f}+|k|-N-1\,,
\end{align}
and setting $2N_{f}=1$, $N=n$ and $k=n+\ell+\tfrac12$ gives
$\tilde N=\tfrac12+n+\ell+\tfrac12-n-1=\ell$, which is \eqref{level_pair}, the reversal of the
level being the charge conjugation.  The self-dual case is $\tilde N=N$, that is $\ell=n$.

The absence of a superpotential, which would be surprising for a Seiberg-like duality, is
forced by the same counting.  The Giveon-Kutasov duality carries mesons
$M_{ij}=J_{ab}Q^{a}_{i}Q^{b}_{j}$ and a superpotential $\mathcal{W}=M\tilde q\tilde q$, and for a
symplectic gauge group the mesons are antisymmetric in the flavor indices, so there are
$\binom{2N_{f}}{2}$ of them.  At $2N_{f}=1$ this is $\binom{1}{2}=0$.  There are no mesons to
write, which is the statement \eqref{meson_vanishes} that the only candidate gauge invariant
built from $Q$ alone vanishes identically, and hence no superpotential.  What is unusual about
\eqref{level_pair} among Seiberg-like dualities is therefore not that a superpotential is
missing but that the meson sector is empty, but these are actually the same fact.

\section{Abelian rank-$0$ Chern-Simons theories}
\label{sec:tadpole}
The Abelian theories that will appear as duals are the two-parameter family
$\mathcal{T}_{\ell,n}$ of \cite{Creutzig:2024ljv}, one member for each pair of positive integers
$(\ell,n)$.  We begin with the member $\ell=1$, the tadpole theory, and develop it in full.
There are two reasons for considering it first.  It is the one member for which every step can be
carried out explicitly, since its superpotential is built from \emph{bare} monopole operators
and its Bethe vacua are available in closed form, and it is the case already met in the
literature, its rank-one member being the minimal theory of \cite{Gang:2018huc} and its higher
ranks the theories of \cite{Gang:2023rei}.  Section~\ref{sec_general_rank0} then treats general
$\ell$, following the same order point by point, so that where the tadpole discussion
generalizes, and where it does not, can be read off directly.

\subsection{The tadpole theory, $\ell=1$}
\label{sec_tadpole_theory}
We consider the 3d $\mathcal{N} = 2$ supersymmetric gauge theory with
\begin{itemize}
  \item Abelian gauge group $G = U(1)^r$;
  \item mixed Chern--Simons level matrix
        $K_r = 2\,C(T_r)^{-1}$, i.e.\ $(K_r)_{IJ} = 2\min(I,J)$,
        with $C(T_r)$ the Cartan matrix of the \emph{tadpole graph}
        $T_r = A_{2r}/\Zb_2$ \cite{Gang:2023rei,Gaiotto:2024ioj};
  \item $r$ chiral multiplets $\Phi^{(1)},\dots,\Phi^{(r)}$ with diagonal charge matrix
        $\mathbf{Q} = \mathbb{I}$, i.e.\ $\Phi^{(I)}$ has charge $+1$ under
        $U(1)_I$ and is neutral under the others.  We write
        $\Phi^{(I)} = (\phi^{(I)},\psi^{(I)})$ for its scalar and fermion components;
  \item monopole superpotential
        $\mathcal{W} = \sum_{I=1}^{r-1} V_{\mathbf{m}_I}$ with
        $\mathbf{m}_I = (\mathbf{0}_{I-2},-1,2,-1, \mathbf{0}_{r-I-1})$,
        i.e.\ $\mathbf{m}_I = C(T_r)\,\mathbf{e}_I$, the leading $-1$ being absent
        for $I=1$, where $\mathbf{m}_1 = (2,-1,\mathbf{0}_{r-2})$.
\end{itemize}
The field content and the charges under the gauge group, 
the surviving topological symmetry $U(1)_y$ and the R-symmetry $U(1)_R$ are
\begin{align}
\label{tadpole_charges}
\begin{array}{c|c|c|c}
& U(1)_1\times\cdots\times U(1)_r & U(1)_y & U(1)_R \\ \hline
\textrm{VM} & \textrm{level } K_r & 0 & 0 \\[2pt]
\Phi^{(I)} & \mathbf{e}_I & 0 & R^{(I)} \\[2pt]
V_{\mathbf m} & q_I[V_{\mathbf m}] & \vec A\cdot\mathbf m & R[V_{\mathbf m}]
\end{array}
\qquad
\vec A = (1,2,\dots,r).
\end{align}
We will discuss the charges of $V_{\mathbf m}$ later in this section and will see that the marginality of $\mathcal{W}$ forces
\begin{align}
\label{Rassignment}
R^{(1)}=\dots=R^{(r-1)}=0\,,
\qquad
R^{(r)} \ \text{unconstrained}\,,
\end{align}
so that the trial R-charge $R^{(r)}$ is a free parameter, exactly as $r_Q$ is on the symplectic side. 
We denote the theory by
\begin{align}
\label{tadpole_name}
U(1)^r_{K_r}\,,
\qquad
K_r = 2\,C(T_r)^{-1}\,,
\end{align}
with the $r$ charge-one chirals and the superpotential
\begin{align}
\label{W_tadpole}
\mathcal{W} \;=\; \sum_{I=1}^{r-1} V_{\mathbf m_{I}}\,,
\qquad
\mathbf m_{I} = C(T_{r})\,\mathbf{e}_{I}\,,
\end{align}
understood.  The sum stops at $r-1$.  Marginality of $V_{\mathbf m_{I}}$ forces $R^{(I)}=0$, as
we verify below, so including the term $I=r$ would force $R^{(r)}=0$ as well and remove the free
parameter of \eqref{Rassignment}.  We call it the \emph{tadpole theory} of rank $r$.

The name is meant literally. 
Every piece of the defining data is read off from the tadpole graph $T_r$.\footnote{Throughout, ``tadpole'' refers to the graph $T_r$, never to a one-point function or to a tadpole cancellation condition.}
The level matrix is $K_r = 2\,C(T_r)^{-1}$. 
The superpotential fluxes are the columns $\mathbf m_I = C(T_r)\,\mathbf{e}_I$. 
The surviving flavor symmetry is generated by the vector $\vec A$ obeying $C(T_r)\vec A = \mathbf{e}_r$ 
and the single node at which $C(T_r)$ fails to be a genuine Cartan matrix, $C(T_r)_{rr} = 1$ rather than $2$, 
is the node carrying the free R-charge $R^{(r)}$. 
These are not independent observations. 
They all follow from the single relation $K\,\mathbf m_I = 2\mathbf{e}_I$, as shown below. 
The monopole superpotential is left out of the name but is always understood. 
Without it the topological symmetry would be the full $U(1)^r$ rather than the single $U(1)_y$ of (\ref{tadpole_charges}).

The tadpole theory (\ref{tadpole_charges}) has a distinguishing feature which
is the Abelian counterpart of the one noted below (\ref{USp2n_Nf1_charges}): 
each node carries a single chiral of charge $+1$ and no conjugate partner, so here too
the matter cannot be paired into a self-conjugate set.  There is no pseudo-reality to appeal to
in the Abelian case, and the statement is simply that the matter is chiral.  
Consequently there is no flavor symmetry acting on the matter at all, 
and the usual Abelian mirror symmetries for $U(1)_k$ with $N_f$ flavors have no immediate counterpart here.
In addition, as we will discuss in detail in section~\ref{sec_LevQuAbel}, since the charge is odd, the CS level is half-integrally quantized.
In section~\ref{sec_matching} we will argue that this theory is the Abelian dual of the $USp(2n)$ theory of the previous section, with the identification $r = n$. 
In the present section we develop the theory on its own terms.

\paragraph{Relation to the literature}
Particular members of the family have been studied before, and it is worth locating them.
At $r=1$, where $K_1=2$ and no superpotential is possible, 
one has $U(1)$ at half-integral level with a single charge-one chiral, 
which up to orientation is the minimal theory of Gang and Yamazaki \cite{Gang:2018huc}. 
The level there is $-\tfrac32$ whereas $\widetilde K_1 = +\tfrac32$, so the two are related by parity. 
In the notation of section~\ref{sec_general_rank0} this is $\mathcal{T}_{1,1}$, the one
member for which both parameters take their smallest value.
For $r>1$ the level matrix $K_r$ and the fluxes $\mathbf m_I$ coincide with those of the Abelian Chern-Simons matter theories of Gang, Kim and Stubbs \cite{Gang:2023rei}, 
whose boundary chiral algebras are argued to realize the non-unitary Virasoro minimal models $M(2,2r+3)$. 
These are the members $\mathcal{T}_{1,r}$, the whole $\ell=1$ line of the family.  The
Abelian theories at $\ell>1$, which we will need in section~\ref{sec_matching}, do not appear
there.
See also \cite{Gang:2021hrd} for the general rank-$0$ framework and \cite{Dedushenko:2023cvd} for a review.

All of these are rank-$0$ theories.
The term was introduced in \cite{Gang:2021hrd} for 3d $\mathcal{N}=4$ SCFTs whose
Coulomb and Higgs branches are both points.
The Abelian theories of this section are a priori only $\mathcal{N}=2$, for which the
invariant statement is simply that the moduli space of supersymmetric vacua is a point.
This is the sense in which the term applies to the Abelian theories of \cite{Gang:2023rei},
and the conjectured infrared enhancement to $\mathcal{N}=4$ is what relates the two
characterizations.  On the symplectic side no such conjecture is needed, the enhancement
being manifest by section~\ref{sec_electric_susy}.
We recover that condition from the Lagrangian in section~\ref{sec_magnetic_moduli}. 
The absence of a Higgs branch follows because $\mathbf{Q}=\mathbb{I}$ admits no gauge invariant monomial in the $\Phi^{(I)}$, 
and the absence of a Coulomb branch follows because the Chern-Simons term gives the vector multiplets a topological mass, 
the monopole superpotential playing no role in either statement. 
It is worth stressing that rank-$0$ is a strong condition, and that the $USp(2n)$ theory of the previous section inherits it if the duality of section~\ref{sec_matching} holds.

What differs here is the treatment of the R-symmetry. 
In
\cite{Gang:2021hrd, Gang:2023rei} a reference R-charge assignment is fixed and a mixing parameter $\nu$ with the topological symmetry is then introduced, 
its superconformal value being determined afterwards by $F$-maximization \cite{Jafferis:2010un}. 
We find such a parameter to be redundant. 
The consistency conditions already leave a one-parameter family of admissible assignments, in the direction $K_r^{-1}\vec A$, 
and since $U(1)_y$ is the only global symmetry available this is precisely the direction along which mixing can act. 
Concretely, for the gauge invariant operators of section~\ref{sec_tadpole_monopoles} the shift $R^{(r)}\to R^{(r)}+\delta$ changes the R-charge by $-2\delta A$, 
which is exactly the effect of mixing with $\nu=-2\delta$ and in the index the entire dependence on $R^{(r)}$ enters through the single factor $q^{-R^{(r)} A}$, 
again as a mixing would. 
The two descriptions parametrize the same one-dimensional freedom, 
i.e.\ $R^{(r)}$ is not an extra choice made on top of $\nu$ but instead of it. 
We treat it as part of the defining data of the theory. 

There is a further reason to prefer this bookkeeping. 
The index is defined for any $\mathcal{N}=2$ theory equipped with a choice of $U(1)_R$, 
with no reference to a superconformal point, and it is that choice which the index grades by. 
Fixing a reference assignment and then hunting for $\nu$ inside the index is therefore not extracting new data. 
It is relabelling which $U(1)_R$ one has graded by, a choice already made in the Lagrangian. 
Keeping $R^{(r)}$ free from the outset makes the whole family manifest at once, 
and the superconformal point is then one value of $R^{(r)}$ among others rather than a prerequisite for writing the index down. 
Which value it is does not follow from anything in this section, and we do not determine it
here.  Section~\ref{sec_matching} fixes it from the symplectic side. 

The topological twist is likewise not additional structure from the present point of view, but a specialization. 
A topological twist requires the R-charge to be integrally quantized on all local operators, 
which by the grading discussed in section~\ref{sec_tadpole_index} happens exactly at $R^{(r)} \in \Zb+\tfrac12$. 
The two values $R^{(r)} = \pm\tfrac12$ are the ones appearing in \eqref{trivial_index}. 
We will not use the twisted theories. 
The comparison in section~\ref{sec_matching} is carried out for the untwisted $\mathcal{N}=2$ theories at generic $R^{(r)}$, 
which is strictly more information than any single twisted specialization.

\subsubsection{Level quantization and the parity anomaly}
\label{sec_LevQuAbel}
The integral matrix $K_r$ appearing in the Lagrangian is a choice of quantization scheme, 
and it is convenient to fix the scheme by declaring what $K_r$ measures: 
$K_r$ is the effective level matrix of the pure Chern-Simons theory obtained when every chiral $\Phi^{(I)}$ is given a large positive real mass and integrated out. 
Integrating out a single fermion of charge vector $\mathbf{Q}_I$ and real mass $m$ shifts the level matrix by $(1/2) \mathrm{sign}(m) \mathbf{Q}_I \mathbf{Q}_I^T$, 
so with this definition the two effective level matrices reachable by uniform real mass deformations are
\begin{align}
K_r \quad (\textrm{all masses positive}) 
\qquad\text{and}\qquad
K_r - \mathbf{Q}\mathbf{Q}^T = K_r - \mathbb{I} \quad(\textrm{all masses negative})
\end{align}
and their midpoint, the parity-symmetric level, is
\begin{align}
\label{Ktilde}
\widetilde K_r= K_r - \frac12\,\mathbf{Q}\mathbf{Q}^{T}=K_r - \frac12\,\mathbb{I},
\qquad
(\widetilde K_r)_{II} = 2I-\frac12\in\Zb+\frac12.
\end{align}
That the parity-symmetric level is half-integral, and so not a legitimate level
on its own, is the parity anomaly.  Equivalently $2(\widetilde K_r)_{II} = 4I-1$
is odd, while the off-diagonal entries $2\min(I,J)$ are integral.  The same
half-unit shift will reappear in section~\ref{sec_matching} as the deformation
that maps to a real mass on the symplectic side.

We quote the level as the integral matrix $K_r$, which is the convention in which every
exponent in the index integrand is an integer.  The alternative is to quote the parity-symmetric
$\widetilde K_r$ of \eqref{Ktilde}, whose entries are half-integral precisely because each node
carries a single charge-one fermion.  That is the convention in which a single charge-one
fermion coupled to $U(1)$ carries a half-integral level, as for the $U(1)_{-1/2}$ theory of
\cite{Seiberg:2016gmd}, and it is also the one in which the minimal theory of
\cite{Gang:2018huc} was quoted at level $-\tfrac32$ above.  The two ways of writing the index
integrand, and the conventions we adopt, are compared below \eqref{mplus_split}.
The price is the bare monopole factor
\begin{align}
\label{bare_monopole_factor}
\Big(\frac{-\,q^{(1-R_\Phi)/2}}{s}\Big)^{m^{+}},
\qquad
m^{+} \equiv \max(m,0) = \tfrac12\big(|m|+m\big),
\end{align}
so that the integrand carries an overall factor $(-1)^{\sum_I (m^{(I)})^{+}}$.

As in section~\ref{sec:electric_level}, the extra sign is the hallmark of a half-integral level, and the decomposition
\begin{align}
\label{mplus_split}
m^{+} = \frac12|m| + \frac12 m
\end{align}
shows exactly how this arises. 
The parity-symmetric piece $\tfrac12|m|$ is common to both ways of writing the theory. 
The piece $\tfrac12 m$, linear in the flux, is precisely the half-unit shift relating $K_r$ to $\widetilde K_r$, 
and is the supersymmetric counterpart of the familiar rule 
that integrating out a fermion of charge $Q$ and real mass $m$ shifts the level by $\tfrac12\,\mathrm{sign}(m)Q^{2}$ \cite{Dimofte:2011py,Beem:2012mb}. 
The sign is carried by the classical Chern-Simons factor if the linear piece is kept, and by the matter 1-loop determinant if it is not. 
It is present either way, and the theory is a \emph{spin} theory.

\paragraph{Appearance of $m^{+}$}
The asymmetry between the two sections is not a change of convention but a property of the matter, 
and it is the same property that made the 1-loop induced gauge charge vanish in section~\ref{sec_monopoles}. 
The fundamental of $USp(2n)$ is pseudo-real, so its weights come in pairs $\pm e_i$, and for each pair
\begin{align}
\big(\rho(m)\big)^{+} + \big(-\rho(m)\big)^{+} \;=\; |\rho(m)|\,.
\end{align}
The linear halves of \eqref{mplus_split} cancel \emph{within} the multiplet, 
the zero-point exponent is $\sum_i |m_i|$ whatever one does, and there is simply nothing to choose. 
The symplectic 1-loop determinant is parity symmetric on its own, and the entire anomaly is carried by the half-integrality of $k$, i.e.\ by the sign from $(-s_i)^{2k m_i}$. 
Here the matter is a single charge-$(+1)$ chiral per node with no conjugate partner, 
the cancellation does not happen, and the linear piece survives.

Two clarifications are in order. 
First, the two ways of writing are not two different integrands but the same one. 
The exponent of $s^{(I)}$ is
\begin{align}
(K_r\mathbf m)^{(I)} - (m^{(I)})^{+}=(\widetilde K_r\mathbf m)^{(I)} - \tfrac12\big|m^{(I)}\big|,
\end{align}
and it is an integer in both. 
No square root survives in the product, and the supersymmetric index is of course the same. 
Second, the convention standard in the supersymmetric index literature is the parity-symmetric one 
\cite{Kim:2009wb,Imamura:2011su,Kapustin:2011jm}, in which every factor is written with $\tfrac12|\rho(m)|$. 
The form with $m^{+}$ is the one natural to effective Chern-Simons levels, 
to the 3d index of \cite{Dimofte:2011py} and to holomorphic blocks \cite{Beem:2012mb}. 
We use the latter since each factor is then separately single valued, 
which makes the bookkeeping of the $|A|=1$ sectors in section~\ref{sec_tadpole_index} easier to follow. 
A reader who prefers the standard convention may restore it by the substitution above.

As we did for the $USp(2n)$ theory, we comment on the relevance of the sign in the index.
There is no Weyl group here, so no Weyl-invariance argument is available but the sign is not removable in any case. 
It cannot be absorbed into $s^{(I)} \to -s^{(I)}$, because the 1-loop determinant of $\Phi^{(I)}$ contains odd powers $(s^{(I)})^{\pm1}$, nor is the convergence of the sum over fluxes sensitive to it. 
Dropping $(-1)^{\sum_I (m^{(I)})^{+}}$ therefore does not spoil any manifest symmetry of the integrand. 
It is nevertheless not a legitimate operation. 
A local counterterm contributes a number of factors of $-1$ linear in the flux, whereas $(-1)^{(m^{(I)})^{+}}$ is not. 
It equals $(-1)^{m^{(I)}}$ for $m^{(I)}>0$ and $+1$ for $m^{(I)}<0$, so it agrees with no $(-1)^{c\,m^{(I)}}$. 
Therefore, the sign cannot be removed by a change of convention, only relocated between the two factors as above, 
and an expression with the sign deleted is not the index of any 3d $\mathcal{N}=2$ theory. 

\subsubsection{The choice of level matrix}
\label{sec_magnetic_level_choice}
The particular matrix $K_r$ is not arbitrary. 
It is fixed, essentially uniquely, by demanding a marginal gauge invariant monopole superpotential of maximal rank. 

Anticipating (\ref{tadpole_monopole_charges}) and (\ref{tadpole_monopole_gauge_charge}) below, 
a bare monopole of flux $\mathbf m$ carries gauge charge $q_I[V_{\mathbf m}] = (K\mathbf m)^{(I)} - (m^{(I)})^{+}$ 
and R-charge $R[V_{\mathbf m}] = \sum_I(1-R^{(I)})(m^{(I)})^{+}$ for $\mathbf{Q}=\mathbb{I}$. 
Suppose we ask for a set of fluxes $\{\mathbf m_I\}$ so that the bare monopole is simultaneously gauge invariant and marginal, 
and whose only positive entry is a single $+2$ in the $I$-th slot.
This last condition is a minimality requirement rather than an extra assumption. Gauge invariance forces $(K\mathbf m)^{(J)} = (m^{(J)})^{+}$, so that $K\mathbf m$ is a non-negative vector whose entries are exactly the positive parts of the flux.  Marginality reads
$\sum_J (1-R^{(J)})(m^{(J)})^{+} = 2$, and since the R-charges of the chirals are non-negative
each term is at most $(m^{(J)})^{+}$, and hence
\begin{align}
\sum_J (m^{(J)})^{+} \ \ge\ 2\,,
\end{align}
with equality precisely when $R^{(J)} = 0$ at every node carrying positive flux.  A total
positive weight of $2$ is therefore the smallest one available, and the only one that does not
require some $R^{(J)}$ to be non-zero and tuned.  Two patterns realize it, $K\mathbf m = 2\mathbf e_I$
and $K\mathbf m = \mathbf e_I + \mathbf e_J$ with $I\neq J$.  We take the former, which assigns one
superpotential term to each node, makes the $\mathbf m_I$ automatically linearly independent
because $K$ is invertible, and is what \emph{maximal rank} means above.  The latter is a
different branch, labelled by pairs of nodes rather than by nodes, which we do not pursue here.
Marginality then reads $R[V_{\mathbf m_I}] = 2(1-R^{(I)}) = 2$, i.e.\ $R^{(I)} = 0$, and gauge invariance reads
\begin{align}
\label{K_defining}
K\mathbf m_I =2\mathbf{e}_I\,,
\end{align}
where we should note that $I \in \{1, 2, \cdots, r-1\}$ (i.e.\ the case $I = r$ is not included) since the monopole superpotential is a sum of $r-1$ (rather than $r$) terms.
The minimal such fluxes are the columns of a Cartan matrix,
$\mathbf m_I = C\,\mathbf{e}_I$, whence $K = 2C^{-1}$. 
Three further requirements select $C$:
\begin{itemize}
  \item \emph{Symmetry}: 
  $K$ is a Chern-Simons level matrix, so $K \in \mathrm{Sym}^2(\Zb^r)$ and $C$ must be symmetric. 
  This restricts $C$ to the simply-laced cases and excludes the Cartan matrix of $C_r$ and the other non-simply-laced Cartan matrices.
  \item \emph{Integrality} of $K = 2C^{-1}$: 
  since $C^{-1} = (\det C)^{-1}\operatorname{adj}C$ with $\operatorname{adj}C$ integral, 
  $2C^{-1}$ is integral precisely when $2\Zb^{r} \subseteq C\Zb^{r}$.
  (This is weaker than $\det C \in \{1,2\}$, since $D_4$ has $\det C = 4$ and yet $\Zb^{4}/C\Zb^{4} = \Zb_2\times\Zb_2$, so $2C^{-1}$ is integral.)
  \item \emph{A single surviving flavor symmetry} requires the $\mathbf m_I$ to span a codimension-one sublattice of $\Zb^r$, 
  i.e.\ we impose the $r-1$ linearly independent conditions \eqref{K_defining}.
\end{itemize}
The tadpole graph $T_r$ satisfies all three requirements: 
$C(T_r)$ is the $A_r$ Cartan matrix except for $C_{rr} = 1$, hence it is symmetric. 
It is positive definite with $\det C(T_r) = 1$, so $\Zb^{r}/C(T_r)\Zb^{r}$ is trivial. 
Its inverse is the integral matrix $\big(C(T_r)^{-1}\big)_{IJ} = \min(I,J)$. 
Hence
\begin{align}
K_r = 2\,C(T_r)^{-1},
\qquad
(K_r)_{IJ} = 2\min(I,J),
\qquad
\det K_r = 2^r\,.
\end{align}
The omitted relation $I=r$ is exactly what leaves one unbroken topological symmetry.
The vectors of topological charges left unbroken by the superpotential, i.e. the solutions of 
$\vec{A}\cdot m_I=0$ for $I=1,\cdots, r-1$, form a rank-one sublattice generated by $\vec A = (1,2,\dots,r)$, 
equivalently $C(T_r)\vec A = (0,\dots,0,1)^{T}$, and correspondingly $R^{(I)} = \lambda\,\delta_{Ir}$ leaves precisely one free R-charge, namely $R^{(r)}$, 
in agreement with \eqref{Rassignment}.

The tadpole graph is the naive $\Zb_2$ folding of $A_{2r}$, 
the one in which the last column is not rescaled, so that $C(T_r)_{rr} = 1$ rather than $2$. 
Rescaling that column by $2$ would instead produce the Cartan matrix of $C_r$, 
which is not symmetric and is therefore excluded by the first requirement above. 
The failure of $C(T_r)$ to be a genuine Cartan matrix \footnote{Although we loosely refer to this as a Cartan matrix, it has a diagonal element $C_{rr} = 1 \ne 2$.} is exactly what keeps it symmetric. 
We stress that this is a statement about the level matrix only. 
The folding that produces $C_r = \mathfrak{usp}(2r)$ is $A_{2r-1}\to C_r$, not $A_{2r}\to T_r$, 
and nothing here should be read as evidence for a relation to $USp(2r)$.

Let us also note that $T_r$, corresponding to $A_{2r}$ ($\Zb_2$), is not the only graph passing these tests. 
Among the positive-definite simply-laced Cartan matrices, 
$2\Zb^{r} \subseteq C\Zb^{r}$
for $A_1$ ($\Zb_2$), $D_{2k}$ ($\Zb_2\times\Zb_2$), $E_7$ ($\Zb_2$) and $E_8$ (trivial), 
while $A_{r\ge2}$ ($\Zb_{r+1}$), $D_{2k+1}$ ($\Zb_4$) and $E_6$ ($\Zb_3$) fail.
The label ``tadpole'' therefore does genuinely specify the family. 
We do not consider the alternatives here.

\subsubsection{Global symmetries and the moduli space}
\label{sec_magnetic_moduli}
The global symmetry is
\begin{align}
U(1)_y \times U(1)_R,
\end{align}
with no non-Abelian factor and, crucially, no factor acting on the chiral multiplets. 
Naively a global $U(1)^r$ acts on the $\Phi^{(I)}$, but with $\mathbf{Q} = \mathbb{I}$ all of it is gauged. 
The topological symmetry $U(1)_{J_1}\times\dots\times U(1)_{J_r}$ is broken 
by the superpotential to those combinations $\vec A\cdot\vec M$ with $\vec A\cdot\mathbf m_I = 0$, 
and since the $\mathbf m_I$ span a codimension-one sublattice a single $U(1)_y$ survives, with charge $A = \sum_{I} I\,m^{(I)}$. 
The chiral multiplets are $A$-neutral so $U(1)_y$ acts only on disorder operators. 
There is no one-form symmetry either. 
The line operators of $U(1)^r_{K_r}$ modulo screening would form $\Zb^r/K_r\Zb^r$, of order $|\det K_r| = 2^r$, 
but $\mathbf{Q} = \mathbb{I}$ generates the full charge lattice $\Zb^r$, so the matter screens everything 
and there is no discrete gauging ambiguity to keep track of.

There is no gauge invariant polynomial in the $\Phi^{(I)}$ at all.  
A superpotential or chiral ring generator built from the chirals would be a monomial $\prod_I (\Phi^{(I)})^{n_I}$ with $n_I \in \Zb_{\ge0}$, 
carrying charge $n_I$ under $U(1)_I$. 
Because $\mathbf{Q}$ is diagonal with $Q_I \neq 0$, 
the charges of distinct fields cannot cancel against one another, and gauge invariance forces $n_I = 0$ for every $I$. 
There is therefore no Higgs branch.

The statement should be distinguished from the \emph{real} moment maps $\mu^{(I)} = |\phi^{(I)}|^2$, 
which need not vanish for $\phi^{(I)} \neq 0$. 
In the presence of the CS term the vacuum equations read
\begin{align}
\label{magnetic_vacuum_eqs}
\frac{1}{2\pi}(K_r)_{IJ}\,\sigma^{(J)} + |\phi^{(I)}|^2 = 0,
\qquad
\sigma^{(I)}\,\phi^{(I)} = 0 \quad (\text{no sum over $I$}).
\end{align}
Contracting the first with $\sigma^{(I)}$ and using the second gives $\frac{1}{2\pi}\,\sigma^{T}K_r\,\sigma = 0$. 
Since $C(T_r)$ is positive definite, so is $K_r = 2C(T_r)^{-1}$, hence $\sigma = 0$. 
The remaining condition $|\phi^{(I)}|^2 = 0$ then forces $\phi = 0$. 
The only solution is the origin.

It is worth separating the two mechanisms at work here, because they are easy to conflate. 
Without a Chern-Simons term the Coulomb branch of $U(1)^r$ would have complex dimension $r$. 
Each photon is dual to a compact scalar $\gamma^{(I)}$, and the branch is parametrized by the chiral superfields 
whose lowest components are $\sigma^{(I)}+i\gamma^{(I)}$, 
equivalently by the monopole operators $V_{\pm\mathbf m}$. 
A Chern-Simons term removes this description altogether. 
At level $K_r$ the gauge fields acquire a topological mass of order $g^{2}K_r$, 
so there is no massless photon left to dualize and hence no dual photon available as a modulus. 
Supersymmetry gives the same mass to the partners $\sigma^{(I)}$, in accordance with the vacuum equations above, 
which force $\sigma = 0$ as soon as $K_r$ is non-degenerate, and here it is positive definite. 
Whatever vacua remain are isolated. 
The monopole operators of section~\ref{sec_tadpole_monopoles} survive, but as ordinary massive local
disorder operators rather than as coordinates on a branch.

The monopole superpotential is therefore \emph{not} what removes the Coulomb branch. 
By the time $\mathcal{W}$ is switched on there is nothing left to lift.
Its role is the different one of breaking $U(1)_{J}^{r}$ down to $U(1)_y$ and imposing relations in the chiral ring of dressed monopoles. 
The absence of moduli is thus a statement about $K_r$, while the symmetry content is a statement about $\{\mathbf m_I\}$.

Because $U(1)_y$ is topological, its contact term is governed by the inverse level matrix. 
Using $C(T_r)\vec A = \mathbf{e}_r$ and $\vec A^{(r)} = r$,
\begin{align}
\label{A_contact}
\vec A^{T} K_r^{-1}\vec A = \frac12\vec A^{T} C(T_r)\vec A
&= \frac12\vec A^{(r)} = \frac{r}{2}.
\end{align}
The mixed and gravitational contact terms $k_{AR},k_{RR},k_{gg}$ are additional data, not determined by the Lagrangian alone.

\subsubsection{Monopole operators}
\label{sec_tadpole_monopoles}
Let $\mathbf m = (m^{(1)},\dots,m^{(r)})\in\Zb^r$ be the magnetic flux through a small sphere surrounding the insertion point. 
The gauge group is Abelian, so there is no Weyl group and no restriction to a fundamental domain. 
The sum over fluxes runs over all of $\Zb^r$ and every $\mathbf m$ labels a distinct sector. 

The quantum numbers of the bare monopole $V_{\mathbf m}$ are, with the level written as $K_r$,
\begin{align}
\label{tadpole_monopole_charges}
R[V_{\mathbf m}] &= \sum_I (1-R^{(I)})\,(m^{(I)})^{+}
= \sum_{I} (m^{(I)})^{+} - R^{(r)}\,(m^{(r)})^{+}, \\
A[V_{\mathbf m}] &= \sum_{I=1}^{r} I\,m^{(I)},
\end{align}
and
\begin{align}
\label{tadpole_monopole_gauge_charge}
\text{gauge charge under } U(1)_I: \quad
q_I[V_{\mathbf m}]&= (K_r\mathbf m)^{(I)} - (m^{(I)})^{+}
\nonumber\\
&= 2\sum_J \min(I,J)\,m^{(J)} - (m^{(I)})^{+}.
\end{align}

Equation (\ref{tadpole_monopole_gauge_charge}) deserves comment. 
Besides the Chern-Simons contribution $(K_r\mathbf m)^{(I)}$ there is a 1-loop piece, and here it does \emph{not} vanish. 
The matter is genuinely chiral, a single charge-$(+1)$ field per node with no conjugate partner, 
so the weights do not come in pairs $\pm\rho$ and the 1-loop sum leaves $-(m^{(I)})^{+}$. 
It is precisely this term that is not linear in $\mathbf m$, and precisely its linear part that is the half-unit shift of (\ref{Ktilde}).
 
Because of the CS term, the bare monopole is \emph{not} gauge invariant. 
It carries $2\sum_J\min(I,J)m^{(J)} - (m^{(I)})^{+}$ units of electric charge under the $I$-th $U(1)$. 
A genuine local operator must be dressed by modes carrying the compensating charge. 
There are no $W$-bosons, so the only available dressings are matter modes:
\begin{itemize}
  \item modes of $\phi^{(I)}$ carry charge $+1$ under $U(1)_I$;
  \item modes of $\bar\psi^{(I)}$ carry charge $-1$ under $U(1)_I$.
\end{itemize}
Consequently a bare monopole with $q_I[V_{\mathbf m}] < 0$ is dressed with $|q_I|$ modes of $\phi^{(I)}$, 
while one with $q_I[V_{\mathbf m}] > 0$ requires $q_I$ modes of $\bar\psi^{(I)}$.

\paragraph{Example}
For a negative flux $\mathbf m = -\mathbf{e}_c$ one finds
\begin{align}
q_I[V_{\mathbf m}] = -2\min(I,c),
\qquad
R[V_{\mathbf m}] = 0,
\qquad
A[V_{\mathbf m}] = -c,
\end{align}
so the bare operator is not gauge invariant. 
Dressing with $2\min(I,c)$ modes of $\phi^{(I)}$ at each node restores gauge invariance, at a total cost of
\begin{align}
\sum_{I=1}^{r} 2\min(I,c) = c\,(2r-c+1)
\end{align}
matter modes.  
Since only $\Phi^{(r)}$ carries a non-zero R-charge, the dressed operator has
\begin{align}
R = 2c\,R^{(r)},
\qquad
A = -c.
\end{align}
For a positive flux $\mathbf m = +\mathbf{e}_c$ 
one has instead $q_I = 2\min(I,c) - \delta_{Ic} > 0$ for all $I$, 
which must be soaked up by $\bar\psi^{(I)}$ modes. 
These are more expensive in $q$, and correspondingly the $A>0$ side of the index is populated by different, more economical fluxes. 
The full set of admissible dressings is most reliably extracted from the supersymmetric index, and we do so in the next subsection.

\subsubsection{The full index}
\label{sec_tadpole_index}
The full index of the tadpole theory $U(1)^r_{K_r}$ is
\begin{align}
\label{ind_tadpole_full}
\mathcal{I}_r
&= \sum_{\mathbf m\in\Zb^{r}} \oint
\left( \prod_{I=1}^{r} \frac{ds^{(I)}}{2\pi i\,s^{(I)}}\;
(s^{(I)})^{\,(K_r\mathbf m)^{(I)}} \right)
\Big(q^{-1/2}y\Big)^{A}
\nonumber \\
&\qquad \prod_{I=1}^{r}
\Big(\frac{-\,q^{(1-R^{(I)})/2}}{s^{(I)}}\Big)^{(m^{(I)})^{+}}
\frac{\big(q^{\,1-\frac{R^{(I)}}{2}+\frac{|m^{(I)}|}{2}}(s^{(I)})^{-1};q\big)_{\infty}}
     {\big(q^{\,\frac{R^{(I)}}{2}+\frac{|m^{(I)}|}{2}}s^{(I)};q\big)_{\infty}}\,,
\end{align}
where
$(K_r\mathbf m)^{(I)} = 2\sum_J\min(I,J)m^{(J)}$,
the R-charge is fixed by the assignment \eqref{Rassignment}, and $y$ is the $U(1)_y$ fugacity. 

The sum over $\mathbf m$ runs over the full magnetic lattice $\Zb^r$, 
with no Weyl factor and no identification of fluxes. 
The projection onto gauge singlets is implemented by the contour integrals alone. 
There is likewise no vector multiplet 1-loop determinant, since an Abelian vector multiplet has no charged modes. 
The entire flux dependence beyond the classical term comes from the matter.
The classical Chern-Simons contribution is $\prod_I (s^{(I)})^{(K_r\mathbf m)^{(I)}}$, 
and the sign of section~\ref{sec:electric_level} resides in the bare monopole factor (\ref{bare_monopole_factor}). 
The exponent $(K_r\mathbf m)^{(I)} - (m^{(I)})^{+}$ of $s^{(I)}$ is an integer, 
so the integrand is single valued, and the sign $(-1)^{(m^{(I)})^{+}}$ is independent of $R^{(I)}$. 
This sign is often described as counting fermionic zero modes, of which there are $|m^{(I)}|$. 
The description is a useful mnemonic but not exact, 
since $(m^{(I)})^{+}$ and $|m^{(I)}|$ differ precisely by the linear term relating $K_r$ to $\widetilde K_r$ in (\ref{Ktilde}). 
The factor $(q^{-1/2}y)^{A}$ supplies the topological charge of the bare monopole together with its universal R-charge shift.
The ratio of $q$-Pochhammer symbols is the one loop contribution from $\Phi^{(I)}$. 
The numerator supplies $\bar\psi^{(I)}$ modes of charge $-1$ at cost $q^{1-\frac{R^{(I)}}{2}+\frac{|m^{(I)}|}{2}}$ 
and the denominator $\phi^{(I)}$ modes of charge $+1$ at cost $q^{\frac{R^{(I)}}{2}+\frac{|m^{(I)}|}{2}}$.

The R-charge parameter enters in only one combination. 
Setting $\zeta^{(I)} = s^{(I)} q^{R^{(I)}/2}$, i.e.\ $\zeta^{(I)} = s^{(I)}$ for $I<r$ and
$\zeta^{(r)} = s^{(r)} q^{R^{(r)}/2}$, removes $R^{(r)}$ from the matter factors and
the residual Chern-Simons factor is a pure shift of the $U(1)_y$ fugacity:
\begin{align}
\label{ind_tadpole_zeta}
\mathcal{I}_r
= \sum_{\mathbf m\in\Zb^{r}} \oint \prod_{I=1}^{r}
\frac{d\zeta^{(I)}}{2\pi i\,\zeta^{(I)}}
(\zeta^{(I)})^{\,(K_r\mathbf m)^{(I)}}
\Big(\frac{-q^{1/2}}{\zeta^{(I)}}\Big)^{(m^{(I)})^{+}}
\frac{\big(q^{1+\frac{|m^{(I)}|}{2}}(\zeta^{(I)})^{-1};q\big)_{\infty}}
     {\big(q^{\frac{|m^{(I)}|}{2}}\zeta^{(I)};q\big)_{\infty}}
\Big(q^{-\frac12-R^{(r)}}y\Big)^{A}.
\end{align}
Defining
\begin{align}
\label{ahat_magnetic}
\hat y \equiv y\,q^{-R^{(r)}}\,,
\end{align}
the last factor is $(q^{-1/2}\hat y)^{A}$ and \emph{all} dependence on the trial R-charge is absorbed into $\hat y$. 
Two consequences:
\begin{enumerate}
  \item The index carries no information about the superconformal R-charge.  It
        must be determined independently, e.g.\ by $F$-maximization \cite{Jafferis:2010un}.
        If the duality of section~\ref{sec_matching} holds it predicts the value, as we
        note there.
  \item For computational purposes one may set $R^{(r)}$ to any convenient value and
        restore it at the end by rescaling $y$ by an appropriate power of $q$. For the assignment $R^{(r)}=+\tfrac12$ the purely $q$-graded expansion
        degenerates and the refined truncation described in
        section~\ref{sec_index_matching} is required.
\end{enumerate}
The grading is nonetheless sensitive to $R^{(r)}$ as a formal series. 
Since $\frac12\sum_I (m^{(I)})^{+}$ contributes powers of $q$ in $\frac12\Zb$ correlated with $A$, one has
\begin{align}
(\text{power of }q) \equiv \left(\frac12-R^{(r)}\right)A \pmod 1,
\end{align}
so for $\frac12-R^{(r)} = p/N$ with $p$ and $N$ coprime, $\mathcal{I}_r$ is a series in $q^{1/N}$ 
and the coefficient of $q^{k/N}$ contains only $y^{A}$ with $pA \equiv k \ (\mathrm{mod}\ N)$.
In particular the grading is integral exactly at $R^{(r)}\in\Zb+\tfrac12$.

At the two such points closest to the origin the index collapses completely in the limit $y \to 1$, i.e.,
\begin{align}
\label{trivial_index}
\mathcal{I}_r\big|_{R^{(r)}=\pm\frac12}\big(q;y=1\big) = 1,
\end{align}
so that no operator other than the identity contributes. 
This is the index counterpart of the statement of section~\ref{sec_magnetic_moduli} that the
vacuum equations have only the origin as a solution. 
It confirms independently that there is no continuous moduli space, and it does so at the two values of $R^{(r)}$ at which the theory can be topologically twisted.

Regarding convergence, the contour is $|\zeta^{(I)}| = 1$ with the denominators expanded in positive powers of $\zeta^{(I)}$, 
which separates the poles of $1/(q^{|m^{(I)}|/2}\zeta^{(I)};q)_\infty$ from their images for $|q|<1$, 
except for the single factor $1/(1-\zeta^{(I)})$ arising at $m^{(I)}=0$ for the nodes with $R^{(I)}=0$, 
whose pole sits on the contour.  For it, the positive-power expansion \emph{is} the
prescription, 
equivalent to approaching the unitarity bound from above, $R^{(I)}\to0^{+}$, 
which moves the pole to $|\zeta^{(I)}|=q^{-R^{(I)}/2}>1$. 
No condition on $y$ arises as the chiral multiplets are $A$-neutral, so the fugacity $y$ appears only in the flux prefactor and never inside a $q$-Pochhammer symbol. 

It is instructive to identify the cheapest contributions with $A = \pm1$.  Working with
(\ref{ind_tadpole_zeta}), at node $I$ the integrand carries
$(\zeta^{(I)})^{\,(K_r\mathbf m)^{(I)} - (m^{(I)})^{+}}$, and gauge invariance requires this to
be compensated.  The available sources are:
\begin{itemize}
  \item $\phi^{(I)}$ modes: one unit of $+1$ charge each, at cost
        $q^{|m^{(I)}|/2}$;
  \item $\bar\psi^{(I)}$ modes: one unit of $-1$ charge each, at cost
        $q^{1+|m^{(I)}|/2}$;
  \item the bare monopole factor $(-q^{1/2}/\zeta^{(I)})^{(m^{(I)})^{+}}$, which is fixed by
        the flux and contributes $q^{\frac12\sum_I (m^{(I)})^{+}}$ and the sign.
\end{itemize}

The cheapest negative flux is $\mathbf m = -\mathbf{e}_1$. 
Then $(m^{(I)})^{+} = 0$ for all $I$ and $(K_r\mathbf m)^{(I)} = -2$ at every node, so two $\phi^{(I)}$ modes are needed at each node. 
These cost $q^{|m^{(I)}|}$, i.e.\ $q$ at node $1$ and nothing at nodes $I \ge 2$ where $m^{(I)} = 0$. 
Together with $[q^{-1/2}\hat y]^{-1} = q^{1/2}\hat y^{-1}$ the total is
\begin{align}
\label{magnetic_Am1}
q^{1}\cdot q^{\frac12}\hat y^{-1} = q^{\frac32}\hat y^{-1}.
\end{align}

Here the cheapest positive flux is not $\mathbf{e}_1$, which would require $(\zeta^{(I)})^{-2}$ at every node $I \ge 2$ at a cost $q^{2}$ apiece, 
but $\mathbf m = \mathbf{e}_r - \mathbf{e}_{r-1}$. 
For this flux $(K_r\mathbf m)^{(I)} = 2[\min(I,r)-\min(I,r-1)]$ vanishes for $I < r$ and equals $2$ at $I = r$, while $(m^{(I)})^{+} = \delta_{Ir}$. 
Only node $r$ is unbalanced, carrying $(\zeta^{(r)})^{\,2-1}$, and a single $\bar\psi^{(r)}$ mode at cost $q^{1+\frac12} = q^{3/2}$ suffices. 
With the bare monopole factor $q^{1/2}$ and $[q^{-1/2}\hat y]^{+1}$ the total is
\begin{align}
\label{magnetic_Ap1}
q^{\frac12}\cdot q^{\frac32}\cdot q^{-\frac12}\hat y = q^{\frac32}\hat y,
\end{align}
the two minus signs 
(from the bare monopole factor and from the leading term of the numerator $q$-Pochhammer) cancelling. 
For $r=1$ the same flux degenerates to $\mathbf m = \mathbf{e}_1$ and gives the same answer. 

Neither leading contribution depends on $r$. 
The extra nodes are either at zero flux, where $\phi^{(I)}$ modes are free, or already balanced. 
The rank is thus invisible in the $|A| = 1$ sector and first appears at higher order. 
The two contributions sit at the same $q$-degree and are exchanged by $\hat y\to\hat y^{-1}$.
At $R^{(r)} = 0$ they combine into $(y+y^{-1})q^{3/2}$, which is indeed common to $\mathcal{I}_1$, $\mathcal{I}_2$ and $\mathcal{I}_3$. 

\subsection{The general rank-$0$ theory $\mathcal{T}_{\ell,n}$}
\label{sec_general_rank0}

We now return to general $\ell$ and go through the same points, indicating in each case
where the tadpole discussion generalizes and where it does not.

\subsubsection{Field content and level matrix}
\label{sec_general_content}

For positive integers $\ell$ and $n$ the theory $\mathcal{T}_{\ell,n}$ has gauge group
$U(1)^{n(2\ell-1)}$, with the factors labelled by a pair $I=(i,a)$, $i=1,\dots,n$ and
$a=1,\dots,2\ell-1$, and a single chiral multiplet $\phi_{ia}$ of unit charge at each node, so
that the charge matrix is again $\mathbf Q=\mathbb{I}$.  The level matrix is
\begin{align}
\label{kappa_family}
\kappa \;=\; C(T_{n})^{-1}\otimes C(A_{2\ell-1})\,,
\qquad
\big(C(T_{n})^{-1}\big)_{ij}=\min(i,j)\,,
\end{align}
with $C(A_{2\ell-1})$ the Cartan matrix of $A_{2\ell-1}$.  It is symmetric and integral, its
diagonal entries being $2i$.  The tadpole graph sits in the first factor and an $A$-type diagram
in the second.  At $\ell=1$ one has $C(A_{1})=(2)$ and $\kappa$ reduces to $K_{n}$
of \eqref{tadpole_charges}, so $\mathcal{T}_{1,n}$ is the tadpole theory of
section~\ref{sec_tadpole_theory}.

Two remarks on names and labels.  The subscripts are those of \cite{Creutzig:2024ljv}, our
$\mathcal{T}_{\ell,n}$ being their $\mathcal{T}_{N,k}$ with $N=\ell$ and $k=n$, so that the
second subscript is the one carried by the tadpole graph.  And we reserve the term \emph{tadpole
theory} for $\ell=1$, where the level matrix is built from $T_{n}$ alone.  The tadpole graph
does not disappear for $\ell>1$, since $C(T_{n})^{-1}$ remains one factor of \eqref{kappa_family},
but it is then only part of the level matrix.  It is worth noting which parameter each factor
carries.  Under the duality of section~\ref{sec_matching} the tadpole factor $T_{n}$ tracks the
rank of the symplectic gauge group and the factor $A_{2\ell-1}$ tracks its level, so the dual of
$USp(2n)$ is built on $T_{n}$ for every $\ell$.

\subsubsection{Level quantization}
\label{sec_general_level}

The discussion of section~\ref{sec_LevQuAbel} applies verbatim.  Each node carries one
chiral of unit charge, so the parity anomaly makes the parity-symmetric level
$\widetilde\kappa=\kappa-\tfrac12\mathbb{I}$ half-integral on the diagonal while $\kappa$ itself
is integral, and every one of these theories is a spin theory.  In the index the sign appears
exactly as in \eqref{CS_classical_sign}, through the bare monopole factor, and gives an overall
$(-1)^{\sum_{I}(m^{(I)})^{+}}$.

\subsubsection{The superpotential and the role of $\ell=1$}
\label{sec_general_W}

Here the family departs from section~\ref{sec_magnetic_level_choice}, and the way it does so
is worth stating precisely, because it is what distinguishes $\ell=1$.  Gauge invariance of a
bare monopole requires $(\kappa\mathbf m)^{(I)}=(m^{(I)})^{+}$, and the minimal solutions are
those with a single $+2$, so that $\kappa\mathbf m=2\mathbf e_{I}$ and
$\mathbf m=2\kappa^{-1}\mathbf e_{I}$.  Integral bare monopole fluxes therefore exist precisely
when $2\kappa^{-1}$ is integral, which is the second of the three requirements of
section~\ref{sec_magnetic_level_choice}.  For the family
$\kappa^{-1}=C(T_{n})\otimes C(A_{2\ell-1})^{-1}$, and $C(A_{2\ell-1})^{-1}$ has entries in
$\tfrac{1}{2\ell}\Zb$, so the requirement holds only at $\ell=1$.  We have checked directly at
$\ell=2$, $n=1$, where $\kappa=C(A_{3})$, that the condition
$(\kappa\mathbf m)^{(I)}=(m^{(I)})^{+}$ has no non-zero integral solution at all.

Accordingly \cite{Creutzig:2024ljv} use monopoles dressed by the chirals.  Their
superpotential is a sum of $n(2\ell-1)-1$ terms,
\begin{align}
\label{W_general}
\mathcal{W}_{\ell,n}
=\sum_{i=1}^{n-1}\sum_{a=1}^{2\ell-1}\phi_{i(a-1)}\phi_{i(a+1)}V_{\mathfrak m_{ia}}
+\sum_{a=1}^{2\ell-2}\phi_{n(a-1)}\phi_{n(a+2)}V_{\mathfrak m_{a}}\,,
\end{align}
with $\phi_{i0}=\phi_{i(2\ell)}=1$ and with fluxes
$(\mathfrak m_{ia})_{jb}=(2\delta_{ij}-\delta_{(i-1)j}-\delta_{(i+1)j})\delta_{ab}$ and
$(\mathfrak m_{a})_{jb}=(\delta_{jn}-\delta_{j(n-1)})(\delta_{ab}+\delta_{(a+1)b})$.  The number
of terms is $(n-1)(2\ell-1)+(2\ell-2)=n(2\ell-1)-1$, one fewer than the number of nodes.  At
$\ell=1$ the second sum is empty, the dressing factors are absent, and one recovers the bare
monopole superpotential \eqref{W_tadpole} with its $n-1$ terms.  The comparison is worth making
in detail, since it is the only place where the two descriptions of the superpotential meet.  At
$\ell=1$ the node label $(i,a)$ reduces to $(i,1)$ and $\phi_{i1}$ is the chiral called
$\Phi^{(i)}$ in section~\ref{sec_tadpole_theory}.  The dressing factors are
$\phi_{i0}=\phi_{i2}=1$ by the convention above, so they disappear.  The flux
$(\mathfrak m_{i1})_{j1}=2\delta_{ij}-\delta_{(i-1)j}-\delta_{(i+1)j}$ is the $i$-th column of
$C(T_{n})$, which is $\mathbf m_{i}$.  And the first sum runs over $i=1,\dots,n-1$, which is the
range in \eqref{W_tadpole}.

Marginality of \eqref{W_general} fixes the R-charges of the chirals up to a single
parameter, as it does at $\ell=1$, and we record the solution since it enters the duality in
section~\ref{sec_matching}.  For the terms of the first sum the positive part of
$\mathfrak m_{ia}$ is a single $2$ at the node $(i,a)$, so
$R[V_{\mathfrak m_{ia}}]=2\big(1-R^{(ia)}\big)$, the R-charge of a bare
monopole being $\sum_{i,a}\big(1-R^{(ia)}\big)(m_{ia})^{+}$ as in
\eqref{tadpole_monopole_charges} and as recorded at general $\ell$ in
section~\ref{sec_general_monopoles}, and $R[\mathcal{W}]=2$ reads
\begin{align}
\label{Rcharge_bulk}
2R^{(ia)} = R^{(i,a-1)}+R^{(i,a+1)}\,,
\qquad i=1,\dots,n-1\,,
\end{align}
with $R^{(i0)}=R^{(i,2\ell)}=0$ from the convention $\phi_{i0}=\phi_{i(2\ell)}=1$.  A solution of
\eqref{Rcharge_bulk} is linear in $a$, and the two boundary values force it to vanish, so
$R^{(ia)}=0$ for every $i<n$.  For the terms of the second sum the positive part of
$\mathfrak m_{a}$ is a $1$ at each of the nodes $(n,a)$ and $(n,a+1)$, and marginality reads
\begin{align}
\label{Rcharge_last}
R^{(n,a-1)}+R^{(n,a+2)} = R^{(na)}+R^{(n,a+1)}\,,
\qquad a=1,\dots,2\ell-2\,.
\end{align}
Writing $g(a)=R^{(na)}-R^{(n,a-1)}$ this says $g(a+2)=g(a)$, so $g$ takes one value on odd
arguments and another on even ones, and $R^{(n,2\ell)}=0$ makes the two opposite.  Hence
\begin{align}
\label{Rcharge_solution}
R^{(ia)}=0 \quad (i<n)\,,
\qquad
R^{(na)}=\begin{cases} \alpha\,, & a \text{ odd}\,,\\ 0\,, & a \text{ even}\,,\end{cases}
\end{align}
with $\alpha$ the one free trial R-charge, carried by the $\ell$ odd nodes of the last row.  At
$\ell=1$ the last row is the single node $a=1$ and \eqref{Rcharge_solution} reduces to
$R^{(I)}=0$ for $I<n$ with $R^{(n)}=\alpha$ free, which is the assignment of
section~\ref{sec_tadpole_theory}.  Setting $\alpha=0$ gives the frame of
\cite{Creutzig:2024ljv} used in \eqref{ind_general_full}.

\subsubsection{Global symmetries and the moduli space}
\label{sec_general_moduli}

The counting of section~\ref{sec_magnetic_moduli} is unchanged in form.  Since
$\mathbf Q=\mathbb{I}$, the $n(2\ell-1)$ flavor symmetries that would act on the chirals are
entirely gauged away, and the $n(2\ell-1)$ topological symmetries are broken by the
$n(2\ell-1)-1$ superpotential terms to the single factor generated by
\begin{align}
\label{S_general}
S=\sum_{i=1}^{n}\sum_{a=1}^{2\ell-1}(-1)^{a+1}\,i\,M_{ia}\,,
\end{align}
with $M_{ia}$ the topological charge of $U(1)_{ia}$ and with $i$ the row label of
section~\ref{sec_general_content}.  Both factors in \eqref{S_general} are forced.  Since $V_{\mathfrak m}$ carries
topological charge $\mathfrak m$, the surviving symmetry must annihilate every flux occurring in
\eqref{W_general}.  For the first family, $\mathfrak m_{ia}$ is a discrete Laplacian in the row
direction, and
\begin{align}
S\cdot\mathfrak m_{ia}=(-1)^{a+1}\big[\,2i-(i-1)-(i+1)\,\big]=0
\nonumber
\end{align}
holds precisely because the coefficient is linear in $i$.  For the second family,
\begin{align}
S\cdot\mathfrak m_{a}=\big[\,n-(n-1)\,\big]\big[\,(-1)^{a+1}+(-1)^{a+2}\,\big]=0
\nonumber
\end{align}
holds precisely because the sign alternates in $a$.  At $\ell=1$ the second family is empty and
\eqref{S_general} is the charge $\vec A=(1,2,\dots,n)$ of section~\ref{sec_magnetic_moduli},
where the same linearity was what made $\vec A$ orthogonal to the columns of $C(T_{n})$.  The
global symmetry is therefore
$U(1)_{S}\times U(1)_{R}$, matching the $U(1)_{a}\times U(1)_{R}$ of the symplectic side, and
the moduli space is a point for the same reason as there, the Chern-Simons term forcing
$\sigma=0$ and the D-term then forcing the chirals to vanish.

\subsubsection{Monopole operators}
\label{sec_general_monopoles}

Let $\mathbf m=(m_{ia})\in\Zb^{n(2\ell-1)}$ be the flux.  The gauge group is Abelian, so
every $\mathbf m$ labels a distinct sector and there is no Weyl group to quotient by, exactly as
in section~\ref{sec_tadpole_monopoles}.  The quantum numbers of the bare monopole are
\begin{align}
\label{general_monopole_charges}
R[V_{\mathbf m}] &= \sum_{i,a} \big(1-R^{(ia)}\big)\,(m_{ia})^{+}\,,
\nonumber\\
S[V_{\mathbf m}] &= \sum_{i=1}^{n}\sum_{a=1}^{2\ell-1} (-1)^{a+1}\,i\,m_{ia}\,,
\end{align}
and
\begin{align}
\label{general_monopole_gauge_charge}
q_{I}[V_{\mathbf m}] = (\kappa\mathbf m)^{(I)} - (m^{(I)})^{+}\,,
\end{align}
the one-loop piece being $-(m^{(I)})^{+}$ for the same reason as at $\ell=1$, namely that
$\mathbf Q=\mathbb{I}$ so that each node sees exactly one chiral of unit charge.  These are
\eqref{tadpole_monopole_charges} and \eqref{tadpole_monopole_gauge_charge} with $K_{r}$ replaced
by $\kappa$ and $\vec A$ by \eqref{S_general}.

Inserting the R-charge assignment \eqref{Rcharge_solution} makes the first line explicit.
Only the odd nodes of the last row carry a non-vanishing R-charge, so
\begin{align}
\label{general_monopole_R}
R[V_{\mathbf m}]
= \sum_{i,a} (m_{ia})^{+} \;-\; \alpha \sum_{a\ \mathrm{odd}} (m_{na})^{+}\,,
\end{align}
which at $\ell=1$ is $\sum_{I}(m^{(I)})^{+}-\alpha\,(m^{(n)})^{+}$, the tadpole expression with
$\alpha=R^{(r)}$.  The dressed operators of \eqref{W_general} have $R=2$ by construction, since
that is what fixed \eqref{Rcharge_solution} in the first place, and the free parameter $\alpha$
drops out of the marginality condition, appearing only in the quantum numbers of the other
monopoles.

What changes is which monopoles are gauge invariant.  At $\ell=1$ the minimal solutions of
$q_{I}[V_{\mathbf m}]=0$ are the fluxes $\mathbf m_{I}$ of section~\ref{sec_tadpole_theory}, and
the superpotential is built from the bare operators $V_{\mathbf m_{I}}$.  For $\ell>1$ there are
none, by the argument of section~\ref{sec_general_W}, and gauge invariance is restored instead by
dressing, the operator $\phi_{i(a-1)}\phi_{i(a+1)}V_{\mathfrak m_{ia}}$ of \eqref{W_general}
carrying $q_{I}=0$ because the two chiral modes supply the two missing units of charge.  The
spectrum of gauge-invariant dressed monopoles is correspondingly richer, and we have not
enumerated it.

\subsubsection{The full index}
\label{sec_general_index}

The index of $\mathcal{T}_{\ell,n}$ is \eqref{ind_tadpole_full} with the same two
replacements and with the R-charges of \eqref{Rcharge_solution}.  Writing it out,
\begin{align}
\label{ind_general_full}
\mathcal{I}\big[\mathcal{T}_{\ell,n}\big]
&= \sum_{\mathbf m\in\Zb^{\,n(2\ell-1)}} \oint
\left( \prod_{i=1}^{n}\prod_{a=1}^{2\ell-1} \frac{dz_{ia}}{2\pi i\,z_{ia}}\;
z_{ia}^{\,(\kappa\mathbf m)_{ia}} \right)
\Big(q^{-1/2}y\Big)^{\sum_{i,a}(-1)^{a+1}i\,m_{ia}}
\nonumber \\
&\qquad \prod_{i=1}^{n}\prod_{a=1}^{2\ell-1}
\left(\frac{-\,q^{\frac{1-R^{(ia)}}{2}}}{z_{ia}}\right)^{(m_{ia})^{+}}
\frac{\big(q^{\,1-\frac{R^{(ia)}}{2}+\frac{|m_{ia}|}{2}}z_{ia}^{-1};q\big)_{\infty}}
     {\big(q^{\,\frac{R^{(ia)}}{2}+\frac{|m_{ia}|}{2}}z_{ia};q\big)_{\infty}}\,,
\end{align}
with the contour the product of unit circles and $(x;q)_{\infty}=\prod_{p\ge0}(1-xq^{p})$.  By
\eqref{Rcharge_solution} every $R^{(ia)}$ vanishes except on the odd nodes of the last row, where
it equals $\alpha$, so the parameter enters through $\ell$ of the $n(2\ell-1)$ factors and
nowhere else.  At $\ell=1$ this is \eqref{ind_tadpole_full}, and setting $\alpha=0$ gives the
frame of \cite{Creutzig:2024ljv}.  The fugacity $\eta$ of
\cite{Creutzig:2024ljv} is related to ours by $y=-\eta$ up to the inversion discussed below, the
prefactors $(-q^{-1/2}\eta)^{S}$ and $(q^{-1/2}y)^{S}$ agreeing under that substitution.

A word on the sign of \eqref{S_general} is needed, since it is not a convention we are free
to choose.  Reversing it sends $y\to y^{-1}$, which by \eqref{parity_conjugate} is the parity
conjugation, so it exchanges $\mathcal{T}_{\ell,n}$ with $\overline{\mathcal{T}}_{\ell,n}$ and
these are inequivalent.  We have fixed the sign by the requirement that the resulting index
matches the symplectic one, and the opposite sign fails already at order $q^{5/2}$ for
$\mathcal{T}_{2,1}$.  The specialization of the topological fugacities used in
\cite{Creutzig:2024ljv} carries the exponent $(-1)^{a}i$ rather than the $(-1)^{a+1}i$ of
\eqref{S_general}, so their $\mathcal{T}_{\ell,n}$ and ours may well differ by exactly this
parity conjugation.  Since \cite{Creutzig:2024ljv} note the same ambiguity in their own
comparison of fusion rings, we do not attempt to fix it here, and every statement we make about
$\mathcal{T}_{\ell,n}$ refers to the theory with \eqref{S_general} as written.

Three remarks on \eqref{ind_general_full}.  The overall sign
$(-1)^{\sum_{ia}(m_{ia})^{+}}$ carried by the bare monopole factor is the parity anomaly of
section~\ref{sec_general_level}, and cannot be removed.  The parameter $\alpha$ is a choice of frame rather than a
restriction.  Absorbing it by $z_{na}\to z_{na}q^{-\alpha/2}$ at the odd nodes of the last row
removes it from the factors just described and leaves it only
in the classical factor, as $q^{-\frac{\alpha}{2}\sum_{a\ \mathrm{odd}}(\kappa\mathbf m)_{na}}$,
and since $\sum_{a\ \mathrm{odd}}C(A_{2\ell-1})_{ab}=2(-1)^{b+1}$ and $\min(n,j)=j$ this
exponent is $-\alpha S$.  All dependence on the trial R-charge therefore sits in
\begin{align}
\label{yhat_general}
\hat y \equiv y\,q^{-\alpha}\,,
\end{align}
which is \eqref{ahat_magnetic} with $\alpha=R^{(r)}$ at $\ell=1$, and it is this that makes the
comparison \eqref{family_index_identity} independent of the frame.  And the factors
$1/(q^{|m_{ia}|/2}z_{ia};q)_{\infty}$ have a pole on the contour at $m_{ia}=0$, exactly as in
section~\ref{sec_tadpole_index}, for which the expansion in positive powers of $z_{ia}$ is the
prescription.

Two properties of $\mathcal{T}_{\ell,n}$ established in \cite{Creutzig:2024ljv} will be used
in section~\ref{sec_matching}.  The first is that it has $\binom{\ell+n}{n}$ Bethe vacua at
generic fugacity.  Their Bethe equations are Nahm's equations for the bilinear form $\kappa$,
and the count is obtained there by matching it to the number of simple modules of
$L_{n}(\mathfrak{osp}_{1|2\ell})$.  At $\ell=1$ the count is $n+1$, which is the result proved
by elimination in appendix~\ref{app_bethe}, where the solutions are also obtained in closed
form.  That derivation is independent of the vertex algebra input, and we know of no closed form
beyond $\ell=1$.

The second is a 3d mirror symmetry within the family, which we state carefully because it
involves a parity conjugation and because we use it repeatedly.  For any theory $\mathcal{T}$ of
the present kind we write $\overline{\mathcal{T}}$ for its \emph{parity conjugate}, obtained by
reversing the sign of every Chern-Simons level.  Since the topological charge is parity odd, the
index of the conjugate is the original one with the fugacity of the surviving topological
symmetry inverted,
\begin{align}
\label{parity_conjugate}
\mathcal{I}\big[\overline{\mathcal{T}}\big](\hat y) = \mathcal{I}\big[\mathcal{T}\big](\hat
y^{-1})\,.
\end{align}
The bar is not removable here.  Each node of $\mathcal{T}_{\ell,n}$ carries a single chiral of
unit charge, hence an odd number of charged fermions, so parity is broken by the anomaly of
section~\ref{sec_general_level} and $\mathcal{T}_{\ell,n}$ is not equivalent to
$\overline{\mathcal{T}}_{\ell,n}$.

The statement of \cite{Creutzig:2024ljv} is then that $\mathcal{T}_{\ell,n}$ and
$\overline{\mathcal{T}}_{n,\ell}$ are 3d mirror to one another, equivalently that
\begin{align}
\label{cgk_mirror}
\mathcal{I}\big[\mathcal{T}_{\ell,n}\big](\hat y)
= \mathcal{I}\big[\mathcal{T}_{n,\ell}\big](\hat y^{-1})\,.
\end{align}
The two sides are genuinely different gauge theories whenever $\ell\neq n$.  Their gauge groups
are $U(1)^{n(2\ell-1)}$ and $U(1)^{\ell(2n-1)}$, whose ranks agree only at $\ell=n$, and their
level matrices $C(T_{n})^{-1}\otimes C(A_{2\ell-1})$ and $C(T_{\ell})^{-1}\otimes C(A_{2n-1})$
are unrelated.  At $\ell=n$ the two theories coincide and \eqref{cgk_mirror} says that
$\mathcal{T}_{n,n}$ is mirror to its own parity conjugate, so that its index is invariant under
$\hat y\to\hat y^{-1}$.  This is what we call self-mirror, and it does not mean that parity is
unbroken, only that the theory is carried to its parity image by the mirror map.

The two topological twists of $\mathcal{T}_{\ell,n}$ are conjectured in
\cite{Creutzig:2024ljv} to carry the boundary vertex operator algebras
$W^{\min}_{n-\frac12}(\mathfrak{sp}_{2\ell})$ and $L_{n}(\mathfrak{osp}_{1|2\ell})$, in the $H$-
and $C$-twists respectively, using the terminology of \cite{Costello:2018fnz} in which the two
twists are labelled by the branch whose chiral ring they compute.  The first reduces at $\ell=1$
to the minimal model $M(2,2n+3)$ quoted in the introduction.  Consistently with
\eqref{cgk_mirror}, the pair of algebras assigned to $\mathcal{T}_{n,\ell}$ is the same pair in
the opposite order.


\section{The duality}
\label{sec_matching}
We now come to the proposal of this paper.  For every pair of positive integers $(n,\ell)$
we propose that
\begin{align}
\label{family_main}
USp(2n)_{\,n+\frac12+\ell} \ \text{with one } \mathbf{2n}
\qquad\longleftrightarrow\qquad
\mathcal{T}_{\ell,n}\,,
\end{align}
with $\mathcal{T}_{\ell,n}$ the Abelian rank-$0$ theory of section~\ref{sec_general_rank0}, under
the identification
\begin{align}
\label{family_dictionary}
\alpha = \tfrac12 - r_Q\,,
\qquad
U(1)_{a}\ \longleftrightarrow\ U(1)_{S}\,,
\qquad
y = a^{2}\,,
\end{align}
where $r_{Q}$ is the trial R-charge of \eqref{USp2n_Nf1_charges} and $\alpha$ the one of
\eqref{Rcharge_solution}, and where the axial symmetry of section~\ref{sec_electric_moduli} is
identified with the surviving topological symmetry \eqref{S_general}.  Equivalently, and more
invariantly, $\hat y=q^{-1/2}\hat a^{2}$ in the combinations $\hat a$ of \eqref{ahat_electric}
and $\hat y$ of \eqref{yhat_general} that carry all the dependence on the trial R-charges.  At
$\ell=1$ the parameter $\alpha$ is $R^{(r)}$ and \eqref{family_dictionary} becomes the
dictionary \eqref{the_dictionary} recorded below.  Note that the same relation $\alpha=\tfrac12-r_{Q}$ holds at every
$\ell$, although $\alpha$ is carried by $\ell$ of the Abelian chirals and $r_{Q}$ by the single
chiral $Q$.  That there is exactly one parameter to match on each side is worth remarking on,
since the two sides arrive at it in opposite ways.  On the symplectic side no superpotential can
be written, the meson vanishing by \eqref{meson_vanishes} and the bare monopoles not being gauge
invariant, so $r_{Q}$ is unconstrained.  On the Abelian side the $n(2\ell-1)$ chiral R-charges
are cut down by the $n(2\ell-1)-1$ marginality conditions of section~\ref{sec_general_W} to the
single $\alpha$ of \eqref{Rcharge_solution}.  The two countings agree because both sides carry
the same global symmetry $U(1)\times U(1)_{R}$, with one Abelian factor available to mix with
the R-symmetry.  One consequence is worth recording, and it is a consequence of the
duality rather than an independent determination.  The superconformal R-charge on the symplectic
side is $r_Q=\tfrac12$ by section~\ref{sec_electric_susy}, where it follows from supersymmetry
that is manifest in the Lagrangian, so if \eqref{family_main} holds then
\eqref{family_dictionary} places the superconformal point of $\mathcal{T}_{\ell,n}$ at
$\alpha=0$, that is at vanishing R-charge for every chiral.  Nothing in
section~\ref{sec:tadpole} fixes that value.  Reaching it from the Abelian side would require
either $F$-maximization \cite{Jafferis:2010un} or the conjectured infrared enhancement of
\cite{Gang:2018huc,Gang:2023rei}, so agreement between $\alpha=0$ and the value obtained by
those routes would be a test of \eqref{family_main} rather than a check of something already
known.
The level on the left is not a choice, since by
section~\ref{sec_electric_vacua} the index vanishes identically for $\ell<0$ and equals $1$ at
$\ell=0$, so $\ell\ge1$ is the statement that the theory is non-trivial, and
$\mathcal{T}_{\ell,n}$ likewise exists only there.  The rest of this
section treats the case $\ell=1$, where $\mathcal{T}_{1,n}$ is the tadpole theory of
section~\ref{sec_tadpole_theory} and the statement reads
\begin{align}
\label{the_duality}
USp(2n)_{n+\frac32} \ \text{with one } \mathbf{2n}
\qquad\longleftrightarrow\qquad
U(1)^r_{K_r} \ \text{with } r \text{ chirals of charge } 1,
\end{align}
under the identification
\begin{align}
\label{the_dictionary}
n = r\,,
\qquad
R^{(r)} = \tfrac12 - r_Q\,,
\qquad
U(1)_a \longleftrightarrow U(1)_y\,,
\qquad
y = a^{2}\,.
\end{align}
The symplectic axial symmetry maps to the Abelian topological symmetry, with the fugacities related by $y=a^{2}$, 
and the two trial R-charges are identified up to the shift in \eqref{the_dictionary},
which is the same shift that makes the two combined fugacities of section~\ref{sec_index} correspond. 
In the three subsections that follow we assemble the evidence for \eqref{the_duality},
organized by the two constraints of section~\ref{sec_dual_constraints} and by the structure of
the two indices.  This is the case we can treat in most detail, because the Abelian side has a
bare monopole superpotential and its Bethe vacua are available in closed form.
General $\ell$ is treated alongside, the index
comparisons being collected in sections~\ref{sec_matching_12} to \ref{sec_matching_22}, and the pattern of
charge-conjugation symmetry in section~\ref{sec_matching_conj}.

\subsection{Global symmetries and the chiral ring}
\label{sec_matching_symmetries}
Constraint (ii) of section~\ref{sec_dual_constraints} was the more restrictive of the two,
and it is met throughout the family in the same way.  The symplectic theory has only
$U(1)_{a}\times U(1)_{R}$, and on the Abelian side the $n(2\ell-1)$ flavor symmetries that would
act on the chirals are entirely gauged away, since $\mathbf Q=\mathbb{I}$, while the
$n(2\ell-1)$ topological symmetries are broken by the $n(2\ell-1)-1$ superpotential terms of
\eqref{W_general} to the single factor $U(1)_{S}$ of \eqref{S_general}.  The mapping of the two
theories is summarized in Table~\ref{table_family_mapping}.
\begin{table}[h!]
\begin{center}
\begin{tabular}{l|l|l}
& $USp(2n)_{\,n+\frac12+\ell}$ & $\mathcal{T}_{\ell,n}$ \\ \hline
gauge group & $USp(2n)$ & $U(1)^{n(2\ell-1)}$ \\
level & $k=n+\frac12+\ell$ & $\kappa=C(T_{n})^{-1}\otimes C(A_{2\ell-1})$ \\
matter & one chiral in $\mathbf{2n}$ & $n(2\ell-1)$ chirals of charge $1$ \\
superpotential & $\mathcal{W}=0$ & $n(2\ell-1)-1$ dressed monopoles \\
global symmetry & $U(1)_{a}\times U(1)_{R}$ & $U(1)_{S}\times U(1)_{R}$ \\
flavor versus topological & $U(1)_{a}$ & $U(1)_{S}$ of \eqref{S_general} \\
moduli space & origin only & origin only \\
vacua under the two masses & $\binom{n+\ell-1}{n},\ \binom{n+\ell}{n}$ &
$\binom{n+\ell-1}{n},\ \binom{n+\ell}{n}$
\end{tabular}
\end{center}
\caption{Duality mappings at general level}
\label{table_family_mapping}
\end{table}

We now set $\ell=1$ and follow the two sides in detail. 
The symplectic theory has only $U(1)_a\times U(1)_R$, 
so on the Abelian side the topological symmetry must map to $U(1)_a$ and \emph{no further} flavor symmetry may survive. 
This is what forced $N=1$ in the naive proposal \eqref{N1_proposal}, and hence restricted $n$ (or $n+2$) to being a perfect square. 

In $U(1)^r_{K_r}$ the constraint is met for every $r$ by a different mechanism.
The $r$ spurious $U(1)$'s that would act on the $\Phi^{(I)}$ are entirely gauged away by $\mathbf{Q} = \mathbb{I}$, 
while the $r$ topological $U(1)$'s are broken to a single $U(1)_y$ by the $r-1$ monopole superpotential terms. 
This is not quite the second option left open in section~\ref{sec_dual_constraints}, where
we envisaged a superpotential lifting the flavor symmetries themselves.  Here the flavor
symmetries never appear, being gauged away, and the superpotential acts on the topological ones
instead. 
By the argument of section~\ref{sec:tadpole} monopole operators are the only superpotential terms available, so the mechanism is essentially forced. 

Sections~\ref{sec_electric_moduli} and \ref{sec_magnetic_moduli} may now be read side by side, and the parallel is close, as summarized in Table~\ref{table_dualitymapping}.

\begin{table}[h!]
\begin{center}
\begin{tabular}{l|l|l}
& $USp(2n)_{n+\frac32}$ & $U(1)^r_{K_r}$ $=\mathcal{T}_{1,r}$ \\ \hline
global symmetry & $U(1)_a\times U(1)_R$ & $U(1)_y\times U(1)_R$ \\
absent factor & $U(1)_J$, as $\pi_1(USp(2n))=1$ & flavor acting on matter, all gauged \\
gauge invariant from matter & $M=J_{ab}Q^aQ^b\equiv0$ & $\prod_I(\Phi^{(I)})^{n_I}$ never neutral \\
Higgs branch & none & none \\
vacuum equations & $\tfrac{k}{2\pi}|\sigma|^2=0\Rightarrow\sigma=0$ &
$\tfrac{1}{2\pi}\sigma^{T}K_r\sigma=0\Rightarrow\sigma=0$ \\
Coulomb branch & lifted by the CS term & lifted by the CS term \\
moduli space & origin only & origin only \\
contact term & $\delta k_{aa}=n$ & $\vec A^{T}K_r^{-1}\vec A=r/2$
\end{tabular}
\end{center}
\caption{Duality mappings}
\label{table_dualitymapping}
\end{table}

The two mechanisms differ in detail but agree in outcome. 
For example, the meson of the symplectic theory vanishes for the group-theoretic reason that $J_{ab}$ is antisymmetric while the components of $Q$ commute, 
whereas on the Abelian side no gauge-invariant monomial exists 
because $\mathbf{Q}=\mathbb{I}$ permits no cancellation between distinct fields. 
Likewise $\sigma=0$ follows on the symplectic side from $k\neq0$ and on the Abelian side from the positive definiteness of $K_r$. 

This last point is a structural parallel rather than a consistency requirement. 
A duality is a statement about the infrared, and there is no fundamental obstruction to the two sides lifting their flat directions by different mechanisms.
For instance, Seiberg-like dualities routinely match a classical constraint on one side to a quantum-generated one on the other. 
It is nevertheless reassuring that here the mechanisms line up exactly. 
The symplectic theory has $\mathcal{W} = 0$, so the Chern-Simons term is the only available lifting mechanism there. 
Had the Coulomb branch of the Abelian theory been removed by the monopole superpotential instead, the parallel would have failed.
As discussed in section~\ref{sec_magnetic_moduli}, this is not what happens since on the Abelian side too, the Chern-Simons term does the lifting on its own, 
and $\mathcal{W}$ only fixes the symmetry content. 
We note that both statements are classical and independent of the R-charge assignment, 
since the vacuum equations \eqref{electric_vacuum_eqs} and \eqref{magnetic_vacuum_eqs} do not involve $U(1)_R$. 
What does depend on the assignment is the unflavored index diagnostic \eqref{trivial_index}, 
which collapses to $1$ only at the two values $R^{(r)} = \pm \tfrac12$.
At generic admissible $R^{(r)}$ the unflavored index is a non-trivial series, 
but its content is the discrete spectrum of dressed monopole operators rather than bosonic moduli, in agreement with the classical statement.

Since the Abelian theory is one of the rank-$0$ theories of \cite{Gang:2018huc,Gang:2023rei},
the duality \eqref{the_duality} implies that $USp(2n)_{n+\frac32}$ with a single
fundamental provides a Lagrangian, non-Abelian realization of the same rank-$0$ SCFT.
These theories are argued to enhance to $\mathcal{N}=4$ in the infrared, and the duality then
predicts the corresponding enhancement on the symplectic side.  As explained in
section~\ref{sec_electric_susy}, on the symplectic side that enhancement is instead manifest in
the Lagrangian.

This is also the point at which the question left open in the introduction, whether
\eqref{the_duality} is a form of mirror symmetry, can be addressed. The tadpole theory is
$\mathcal{T}_{1,r}$ in the notation of section~\ref{sec_general_rank0}, and a mirror of that
family, $\overline{\mathcal{T}}_{r,1}$, is conjectured in \cite{Creutzig:2024ljv}.  What
\eqref{the_duality} supplies is then a non-Abelian Lagrangian for a theory known so far only in
an Abelian presentation.

The point of view of section~\ref{sec:tadpole} makes this concrete, and the information runs
in the opposite direction to what one might expect. 
Supersymmetry enhancement to $\mathcal{N}=4$ is not a property of a theory alone but of a theory together with a choice of $U(1)_R$, 
the enhanced $\mathcal{N}=4$ R-symmetry containing a particular $U(1)_R$, 
and in the $\nu$ language one says that enhancement occurs at a particular value of the mixing parameter. 
Since $\nu$ is nothing but $R^{(r)}$, the statement is that enhancement occurs at a particular value $R^{(r)}_\star$ in the family \eqref{Rassignment}. 
On the Abelian side that value is not determined by anything in section~\ref{sec:tadpole},
and \cite{Gang:2018huc,Gang:2023rei} obtain it from the infrared.  On the symplectic side there
is no such freedom, the enhancement being manifest in the Lagrangian by
section~\ref{sec_electric_susy} and fixing $r_Q=\tfrac12$.  Since \eqref{the_duality} matches the
two one-parameter families as a whole rather than two isolated points, the dictionary
\eqref{the_dictionary} transports that determination and gives $R^{(r)}_{\star}=0$.  What the
duality supplies is therefore not a prediction of enhancement on the symplectic side, which is
already known there, but a prediction of where in the family \eqref{Rassignment} the Abelian
theory enhances.
The content of the statement is not the genericity of the R-charge, which by
\eqref{yhat_ahat} below carries no information beyond a single assignment, but the dictionary itself, since
\eqref{the_dictionary} fixes which combination of $U(1)_R$ and $U(1)_a$ on the symplectic side
corresponds to a given combination of $U(1)_R$ and $U(1)_y$ on the Abelian side, and hence
which Abelian R-charge is the one sitting inside the enhanced R-symmetry.

The contact terms \eqref{electric_contact} and \eqref{A_contact} agree at $n=r$ up to the normalization fixed by $y=a^{2}$. 
The mixed and gravitational contact terms $k_{aR},k_{RR},k_{gg}$ are not fixed by this argument and must be specified as part of the duality map. 
They are what allows the overall powers of $q$ in the two indices to differ, as discussed below.

\subsection{Vacuum counting}
\label{sec_vacuum_counting}
Constraint (i) of section~\ref{sec_dual_constraints} requires the pair
\eqref{two_vacuum_counts} to be reproduced under the two real mass deformations, and at general
$\ell$ this is immediate on the Abelian side.  The larger count, $\binom{n+\ell}{n}$, is the
number of Bethe vacua that \cite{Creutzig:2024ljv} assigns to $\mathcal{T}_{\ell,n}$, so the
match holds throughout the family, and it is the first of the checks of \eqref{family_main} that
is available at every $(n,\ell)$.  What the rest of this subsection adds is the mechanism, which
we can follow only at $\ell=1$, where the Bethe vacua are available in closed form.

At $\ell=1$ constraint (i) requires the pair $(1,\,n+1)$ of \eqref{two_vacuum_counts} to be
reproduced under the two real mass deformations.

Under \eqref{the_dictionary} a real mass for the symplectic $U(1)_a$ maps to an FI deformation of the Abelian theory along the direction $\xi^{(I)} = \zeta\,\vec A^{(I)} = \zeta\,I$. 
The corresponding vacuum equation $\frac{1}{2\pi}K_r\sigma = -\xi$ gives, using $K_r^{-1} = \tfrac12 C(T_r)$ and $C(T_r)\vec A = \mathbf{e}_r$,
\begin{align}
\sigma^{(I)} = -\pi\,\zeta\,\delta_{Ir},
\end{align}
so that the deformation switches on a vev only for the last node, exactly the node carrying the free R-charge. 
This is a non-trivial structural check. 
The distinguished node of the tadpole graph is the one the duality singles out. 

At $r=1$ the counting is immediate. 
There $K_1 = 2$, no superpotential is possible, and in the parity-symmetric scheme $\widetilde K_1 = \tfrac32$. 
The effective levels are $\widetilde K_1 \pm \tfrac12 \in \{1, 2 \}$ and the vacuum counts $\{1, 2 \}$ agree with $(1,\,n+1) = (1,2)$ at $n=1$. 
Equivalently, $r=1$ reproduces the level $k' = \tfrac{n+2}{2}$ with $\sum_i q_i^{2}=n$ demanded by \eqref{dual_condition}, i.e.\ the duality \eqref{n1_duality}.

For $r \ge 2$ the deformed theory is not a pure Chern-Simons theory, 
and the quantity $|\det K_{\rm eff}|$ that counted vacua above is simply not the relevant one. 
Two ingredients are missing from it. 
The FI deformation gives a vev only to $\sigma^{(r)}$, so only $\Phi^{(r)}$ becomes massive, 
and the chirals $\Phi^{(I)}$ with $I<r$ remain massless and charged. 
The monopole superpotential survives the deformation as well, 
since each $V_{\mathbf m_I}$ carries $A = \vec A\cdot\mathbf m_I = 0$ and is neutral under the symmetry being deformed. 
A determinant of a level matrix knows about neither.

The tool that does accommodate both is the effective twisted superpotential, 
and we now use it to reproduce the count $r+1$.

\paragraph{Bethe vacua}
We recall the framework briefly. 
Compactifying a 3d $\mathcal{N}=2$ theory on $\mathbb{R}^{2}\times S^{1}_{\beta}$ gives an effective 2d $\mathcal{N}=(2,2)$ theory on its Coulomb branch, 
governed by an effective twisted superpotential $\widetilde{\mathcal{W}}(u)$ of the dimensionless variables $u_I = a_I + i\beta\sigma^{(I)}$ \cite{Nekrasov:2009uh}. 
Its critical points, the solutions of the Bethe equations
\begin{align}
\label{bethe_general}
\Pi_I(u) \equiv \exp\Big(2\pi i \frac{\partial\widetilde{\mathcal{W}}}{\partial u_I}\Big) = 1,
\qquad I = 1,\dots,r,
\end{align}
are the Bethe vacua, and whenever the Witten index is well defined it equals the number of gauge inequivalent solutions of \eqref{bethe_general} \cite{Closset:2016arn}. 
Equivalently, the twisted index on $\Sigma_g\times S^1$ at $g=1$ is a UV computation of that index. 
The conventions are the ones we have already adopted, in particular the quantization in which the level matrix is the integral $K_r$ \cite{Closset:2023vos}.

Writing $x_I = e^{2\pi i u_I}$, the ingredients are standard. 
The Chern-Simons term at level $K_r$ contributes $\prod_J x_J^{\,(K_r)_{IJ}}$ to $\Pi_I$, 
a chiral multiplet of charge $+1$ under the $I$-th node contributes $(1-x_I)^{-1}$, 
and an FI parameter contributes its fugacity $\eta_I$. 
The R-charges do not enter, as they should not for a count of vacua.

The monopole superpotential does enters in one way. 
While it does not contribute to $\widetilde{\mathcal{W}}$, which is computed from the perturbative Kaluza-Klein spectrum on the Coulomb branch, 
it does is constrain the FI fugacities. 
$V_{\mathbf m_I}$ carries topological charge $\mathbf m_I$, so \eqref{W_tadpole} is neutral only if
\begin{align}
\prod_J \eta_J^{\,(\mathbf m_I)^{(J)}} = 1
\qquad (I = 1,\dots,r-1)\,,
\end{align}
whose general solution is $\eta_J = \eta^{\,\vec A^{(J)}} = \eta^{\,J}$. 
This is the same statement as the breaking of $U(1)_{J}^{r}$ to $U(1)_y$, now expressed in terms of the FI parameters, 
and it leaves the single parameter $\eta$. 
Consistently, it is the same direction $\xi^{(I)} = \zeta\,I$ that the dictionary
\eqref{the_dictionary} assigns to the symplectic real mass.  The two statements are one and the
same, both expressing that the superpotential leaves only $U(1)_y$ unbroken.

Assembling these, the Bethe equations of the tadpole theory read
\begin{align}
\label{bethe_tadpole}
\Pi_I = \eta^{I} \frac{\prod_J x_J^{(K_r)_{IJ}}}{1-x_I} = 1,
\qquad I = 1,\dots,r\,,
\end{align}
and the tadpole level matrix makes the monomial factorize,
\begin{align}
\prod_J x_J^{\,2\min(I,J)} \;=\; \prod_{c=1}^{I} w_c^{2}\,,
\qquad
w_c \;\equiv\; \prod_{J\ge c} x_J\,,
\end{align}
since $\min(I,J)$ counts the $c$ with $c\le I$ and $c\le J$. 
Equation \eqref{bethe_tadpole} is a system of $r$ Laurent polynomial equations, 
and we count its solutions in $(\mathbb{C}^{*})^{r}$.
In appendix~\ref{app_eliminant} we solve the system in closed form and prove that for generic $\eta$
\begin{align}
\label{bethe_count}
\#\,\text{Bethe vacua}\big[U(1)^r_{K_r}\big] = r+1
\qquad \text{for all } r\ge1\, .
\end{align}
An explicit Gr\"obner basis computation, as in \cite{Closset:2023jiq,Closset:2023vos},
independently confirms the count for $r\le4$, see appendix~\ref{app_counting}.

This is the second of the two symplectic counts, $\binom{n+1}{n} = n+1$ of \eqref{two_vacuum_counts} at $\ell=1$ and $n=r$: 
the one that grows with the rank, and hence the one that actually tests the identification $n=r$. 
Three descriptions of the same integer, with no free parameter left to adjust, now agree: 
the $n+1$ integrable weights of $\widehat{\mathfrak{usp}}(2n)_{\ell=1}$ on the symplectic side, 
the $r+1$ Bethe vacua of \eqref{bethe_count}, 
and the $r+1$ primaries of the minimal model $M(2,2r+3)$ that the twist of the Abelian theory is argued to produce \cite{Gang:2023rei}. 
We regard \eqref{bethe_count} as the main evidence for \eqref{the_duality} outside the index. 
We note also that the theory \eqref{tadpole_charges} and its Bethe vacua appear in \cite{Gaiotto:2024ioj}, 
where the level matrix is likewise written as $2\,C(T_r)^{-1}$ in terms of the tadpole diagram.

The remaining symplectic number, the $1$ of \eqref{two_vacuum_counts}, is not a count at
generic $\eta$, that count being $r+1$ for every generic value, but a statement about the
opposite infinite-mass limit.  Under \eqref{the_dictionary} the two symplectic real mass
deformations correspond to $\zeta\to\mp\infty$, that is to $\eta\to0$ and $\eta\to\infty$, and in
these limits solutions of \eqref{bethe_tadpole} run off to the boundary of
$(\mathbb{C}^{*})^{r}$, and deciding which of them still describe vacua of the deformed theory
requires a care that we do not attempt here.
The closed-form solution of appendix~\ref{app_eliminant} is at least suggestive.
As $\eta\to0$ exactly one solution, $x_I\to1$ for all $I$, remains at an interior point of the
torus, while all the others escape to the boundary, see the remark at the end of the appendix.
This does not affect the index comparison below, which tests the undeformed theories directly and at generic R-charge.

\subsection{Matching of the full indices}
\label{sec_index_matching}
The identity to be checked is
\begin{align}
\label{family_index_identity}
\mathcal{I}^{USp(2n)}\big(q;\hat a\big)\Big|_{k=n+\frac12+\ell}
\;=\;
\mathcal{I}\big[\mathcal{T}_{\ell,n}\big]\big(q;\hat y\big)\,,
\qquad
\hat y=q^{-\frac12}\hat a^{\,2}\,,
\end{align}
The two sides are \eqref{ind_USp2n_Nf1_full} at trial R-charge $r_{Q}$ and \eqref{ind_general_full} at
trial R-charge $\alpha$, and each depends on its parameter only through the combination shown,
by \eqref{ahat_electric} and \eqref{yhat_general}.  The identity is therefore one statement, not
a family of statements indexed by the frame, and we may check it in whichever frame is
convenient.

In this section we spell out the index test of the duality \eqref{the_duality} (for $\ell = 1$) rank by rank,
writing the explicit expressions to be compared for $n=r=1,2,3$
and recording the structural features of each pair of theories that the comparison probes.
These are the cases $(n,\ell)=(1,1)$, $(2,1)$ and $(3,1)$ of the family
\eqref{family_main}, and the analysis here is the most detailed we are able to give, since it
uses the explicit form of the tadpole index of section~\ref{sec_tadpole_index}.  The
corresponding tests at $\ell>1$ are collected in sections~\ref{sec_matching_12} to \ref{sec_matching_22}.

\paragraph{The identity to be checked}
The duality \eqref{the_duality} with the dictionary \eqref{the_dictionary} predicts,
for every admissible value of the trial R-charge,
\begin{align}
\label{index_identity}
\mathcal{I}^{USp(2n)}\big(q;a\big)\Big|_{r_Q}
\;=\;
\mathcal{I}_{r}\big(q;y\big)\Big|_{R^{(r)}=\frac12-r_Q\,,\;\;y\,=\,a^{2}}\,,
\qquad n=r\,.
\end{align}
This is \eqref{family_index_identity} at $\ell=1$, written in $a$ and $y$
rather than in the combinations $\hat a$ and $\hat y$, with the relationship given in \eqref{yhat_ahat} below.
From here on we use $n$ for the common rank, restoring the separate symbol $r$ only where
the two families of indices are compared off the diagonal $n=r$.
As functions, the two sides depend on their trial R-charges only through the combinations
\eqref{ahat_electric} and \eqref{ahat_magnetic}, which the dictionary relates as
\begin{align}
\label{yhat_ahat}
\hat y \;=\; y\, q^{-R^{(r)}}
\;=\; a^{2}\, q^{\,r_Q-\frac12}
\;=\; q^{-\frac12}\,\hat a^{2}\,,
\end{align}
so that the identities \eqref{index_identity} at two values $r_Q$ and $r_Q'$ are related by the
substitution $a \to a\, q^{(r_Q'-r_Q)/2}$, and agreement at one admissible value implies agreement at all.

\paragraph{Choice of expansion scheme}
Equation \eqref{yhat_ahat} might suggest eliminating the R-charges altogether and comparing the two
indices as series in $(q,\hat a)$ and $(q,\hat y)$.
We deliberately do \emph{not} do this, for the following reason.
Expanding in powers of $q$ at fixed $\hat a$ is, by the definition $\hat a = a\,q^{r_Q/2}$,
nothing but the expansion at the assignment $r_Q=0$.
At that assignment the scalar letters of $Q$ descend to weight $q^{0}$,
and the power of $q$ alone no longer truncates the mode sums entering the localization
formula.  Each order in $q$ receives contributions from arbitrarily many such letters, and no
finite computation returns an exact coefficient.
The assignment itself is nevertheless not pathological.
Each letter of $Q$ carries $U(1)_a$ charge $+1$.  Grading in addition by that charge is
legitimate in the regime \eqref{electric_contour}, where every geometric sum is an expansion in
positive powers of $a\,s^{\pm1}$, and it restores a finite computation at every order.
In this refined scheme the $r_Q=0$ expansion exists, its coefficients are Laurent polynomials in
$a$, and it coincides with the Abelian series at $R^{(r)}=+\tfrac12$, as the duality dictates.
For the systematic comparison we therefore keep the R-charge dependence explicit and quote,
for every rank, the expansions at one and the same set of three assignments:
the two interior values $r_Q=\tfrac12$ and $r_Q=\tfrac13$, i.e.\ $R^{(r)}=0$ and $R^{(r)}=\tfrac16$,
for which the $q$-grading alone truncates every sum,
together with the boundary value $r_Q=0$, i.e.\ $R^{(r)}=+\tfrac12$, in the refined scheme.
Using the same assignments at every rank independently tests the implementation and the grading
structure described next, and makes the rank-by-rank comparison of coefficients directly readable
at each of them.
\paragraph{Grading}
For $r_Q = p/N$, with $p$ and $N$ coprime, both sides of \eqref{index_identity} are series in $q^{1/N}$, with
\begin{align}
\label{grading_rule}
(\text{power of } q) \;\equiv\; r_Q\, A \pmod 1\,,
\qquad
A \;=\; \tfrac12\times\big(U(1)_a \text{ charge}\big) \;=\; U(1)_y \text{ charge}\,,
\end{align}
which is the grading rule of section~\ref{sec_tadpole_index} written at $R^{(r)}=\tfrac12-r_Q$.
The pattern of fractional powers in the symplectic expansion therefore exhibits the map of gradings
directly, independently of the values of the coefficients, and does so identically for all
ranks.  In particular, for the assignment $r_Q=0$ all powers are integral.

\paragraph{Contact terms}
The identity \eqref{index_identity} is claimed as a strict equality.
The relative contact terms left unfixed in section~\ref{sec_matching} could in principle act as an
overall power of $q$ (from $k_{RR}$, $k_{gg}$) or as a shift $y \to q^{c}\, y$ (from $k_{aR}$).
Both sides begin with $1+\dots$, which fixes the overall power to zero,
and the leading $|A|=1$ coefficients line up with $c=0$, as we exhibit at $n=r=1$ below.
Any residual mismatch found in the expansion should be attributed to these counterterms
and would then become part of the duality map.

\paragraph{A corollary}
The trivial-index property \eqref{trivial_index} of the Abelian theory at the two values
$R^{(r)}=\mp\tfrac12$, at which the R-charge is integrally quantized and the theory admits a
topological twist, translates, through the dictionary, into the all-order predictions
\begin{align}
\label{electric_trivial_prediction}
\mathcal{I}^{USp(2n)}\big(q;\,a^{2}=1\big)\Big|_{r_Q=1} \;=\; 1\,,
\qquad
\mathcal{I}^{USp(2n)}\big(q;\,a^{2}=1\big)\Big|_{r_Q=0} \;=\; 1\,,
\end{align}
the second understood in the refined scheme described above.
We have checked both directly: at $n=1$ to order $q^{7}$ and at $n=2,3$ to order $q^{5}$,
the specialization $a^{2}=1$ of the symplectic expansions at both $r_Q=1$ and $r_Q=0$
indeed collapses to $1$.
For $r_Q=0$ this is manifest in the boundary series quoted in the following subsections.

For the expansions themselves, only finitely many flux sectors contribute at any fixed order in $q$,
since the minimal $q$-degree of a sector grows with the flux,
as illustrated by \eqref{electric_m1_min} on the symplectic side and by
\eqref{magnetic_Am1}, \eqref{magnetic_Ap1} on the Abelian side.

\subsubsection{$SU(2)_{\frac52}$ with one doublet versus the minimal Abelian theory}
\label{sec_matching_rank1}
\paragraph{The symplectic theory}
Using $USp(2)\cong SU(2)$, the symplectic theory is $SU(2)_{5/2}$ with a single doublet $Q$ and $\mathcal{W}=0$.
The Weyl group is $\Zb_2$, the GNO flux is a single integer $m$,
and the only positive root is the long root $2e_1$,
so the vector multiplet contributes the two factors $(1-q^{|m|}s^{\pm2})$
and the monopole power $q^{-|m|}$.
The level satisfies $2k=5$, odd, so the classical factor is
$(-s)^{5m}=(-1)^{m}s^{5m}$: the theory is a spin theory, and the sign $(-1)^{m}$
is irremovable in the sense of section~\ref{sec:electric_level}.
The two mass deformations give $k_{\mathrm{eff}}\in\{2,3\}$ and, by \eqref{two_vacuum_counts},
the vacuum counts $(1,2)$.
The full index \eqref{ind_USp2n_Nf1_full} specializes to
\begin{align}
\label{ind_SU2}
\mathcal{I}^{USp(2)}
&= \frac{1}{2}\sum_{m\in\Zb} \oint \frac{ds}{2\pi i\, s}\,
(-s)^{5m}\;
q^{-|m|}\,
\big(1-q^{|m|}s^{2}\big)\big(1-q^{|m|}s^{-2}\big)
\nonumber\\
&\qquad\times
\big(q^{\frac{1-r_Q}{2}} a^{-1}\big)^{|m|}\,
\frac{(q^{1-\frac{r_Q}{2}+\frac{|m|}{2}}a^{-1}s^{-1};q)_{\infty}\,(q^{1-\frac{r_Q}{2}+\frac{|m|}{2}}a^{-1}s;q)_{\infty}}
     {(q^{\frac{r_Q}{2}+\frac{|m|}{2}}a\, s;q)_{\infty}\,(q^{\frac{r_Q}{2}+\frac{|m|}{2}}a\, s^{-1};q)_{\infty}}\,,
\end{align}
with the contour and expansion regime \eqref{electric_contour}.

\paragraph{The Abelian theory}
The rank-one tadpole theory is $U(1)_{K_1}$ with $K_1=2$ and a single chiral $\Phi$ of charge $+1$.
There is no monopole superpotential, since the sum in $\mathcal{W}$ has $r-1=0$ terms,
so the R-charge $R\equiv R^{(1)}$ is the free parameter and the topological symmetry is the full $U(1)_y$
with charge $A=m$.
The parity-symmetric level is $\widetilde K_1 = \tfrac32$,
half-integral as required by the parity anomaly of a single charge-one fermion,
and the effective levels $\widetilde K_1 \pm \tfrac12 \in \{1,2\}$ give
$|k_{\mathrm{eff}}|$ vacua each, reproducing the symplectic pair of vacuum counts $(1,2)$ directly.
Up to parity this is the minimal $\mathcal{N}=2$ theory of \cite{Gang:2018huc}.
The full index \eqref{ind_tadpole_full} specializes to
\begin{align}
\label{ind_tadpole1}
\mathcal{I}_1
= \sum_{m\in\Zb} \oint \frac{ds}{2\pi i\,s}\;
s^{2m}\,
\big(q^{-\frac12}y\big)^{m}\,
\Big(\frac{-\,q^{(1-R)/2}}{s}\Big)^{m^{+}}
\frac{\big(q^{1-\frac{R}{2}+\frac{|m|}{2}}s^{-1};q\big)_{\infty}}
     {\big(q^{\frac{R}{2}+\frac{|m|}{2}}s;q\big)_{\infty}}\,,
\end{align}
with contour $|s|=1$, $|q|<1$ and no condition on $y$.

The identification \eqref{index_identity} at this rank is the index form of
\eqref{n1_duality}.  Up to a parity transformation switching the sign of the $SU(2)$
Chern-Simons level, this duality has previously been conjectured \cite{Okazaki:2024paq} in the
presence of a boundary, with a check of matching half-indices.

\paragraph{Leading coefficients}
Before quoting the expansions we exhibit the $|A|=1$ sector,
which fixes the contact-term ambiguity discussed above.
On the Abelian side the cheapest contributions \eqref{magnetic_Ap1} and \eqref{magnetic_Am1} are,
restoring $y=\hat y\,q^{R}$ and setting $R=\tfrac12-r_Q$,
\begin{align}
q^{\frac32}\hat y \;=\; q^{\,1+r_Q}\,y\,,
\qquad
q^{\frac32}\hat y^{-1} \;=\; q^{\,2-r_Q}\,y^{-1}\,,
\end{align}
from $V_{+1}$ dressed with one $\bar\psi$ mode and $V_{-1}$ dressed with two $\phi$ modes, respectively.
On the symplectic side the natural candidate contributions with these weights arise already at zero flux:
the modes of $Q$ enter \eqref{ind_SU2} with weight $q^{\frac{r_Q}{2}+\ell}\,a\,s^{\pm1}$
and those of $\bar\psi$ with weight $q^{1-\frac{r_Q}{2}+\ell}\,a^{-1}s^{\mp1}$, $\ell\ge0$,
so the gauge invariants
\begin{align}
J_{ab}\,Q^{a}\partial Q^{b}\;\longmapsto\; q^{\,1+r_Q}\,a^{2}\,,
\qquad
J^{ab}\,\bar\psi_{a}\bar\psi_{b}\;\longmapsto\; q^{\,2-r_Q}\,a^{-2}\,,
\end{align}
match the two Abelian terms in $q$-degree and in $U(1)_a$ charge, for every $r_Q$
(the meson $J_{ab}\, Q^{a} Q^{b}$ itself vanishes by \eqref{meson_vanishes}, so the derivative is essential).
We stress that this identifies index contributions but
the precise operator map, including possible mixing with other states of the same charges,
is fixed only by the full expansion.

\paragraph{The expansions}
We have expanded \eqref{ind_SU2} and \eqref{ind_tadpole1} at the following three R-charge assignments and find
perfect agreement.

At $r_Q=\tfrac12$, corresponding to $R=0$ on the Abelian side, both sides give
\begin{align}
\label{rank1_expansion_rQhalf}
\mathcal{I}^{USp(2)}\Big|_{r_Q=\frac12}
= \mathcal{I}_1\Big|_{R=0}
&= 1 - q + \big(a^{2}+a^{-2}\big)q^{\frac32} - 2q^{2} + \big(a^{2}+a^{-2}\big)q^{\frac52} - 2q^{3}
\nonumber\\
&\quad
+ \big(a^{2}+a^{-2}\big)q^{\frac72} - 2q^{4} + \big(a^{4}+a^{-4}\big)q^{5}
- \big(a^{2}+a^{-2}\big)q^{\frac{11}{2}}
\nonumber\\
&\quad
+ \big(1+a^{4}+a^{-4}\big)q^{6} - 3\big(a^{2}+a^{-2}\big)q^{\frac{13}{2}}
+ \big(5+2a^{4}+2a^{-4}\big)q^{7}
\nonumber\\
&\quad
- 5\big(a^{2}+a^{-2}\big)q^{\frac{15}{2}} + \big(7+2a^{4}+2a^{-4}\big)q^{8}
- 7\big(a^{2}+a^{-2}\big)q^{\frac{17}{2}}
\nonumber\\
&\quad
+ \big(11+3a^{4}+3a^{-4}\big)q^{9} + \dots\,,
\end{align}
which we have verified up to and including order $q^{15}$.

At $r_Q=\tfrac13$, corresponding to $R=\tfrac16$, both sides give
\begin{align}
\label{rank1_expansion_rQthird}
\mathcal{I}^{USp(2)}\Big|_{r_Q=\frac13}
= \mathcal{I}_1\Big|_{R=\frac16}
&= 1 - q + a^{2}q^{\frac43} + a^{-2}q^{\frac53} - 2q^{2} + a^{2}q^{\frac73} + a^{-2}q^{\frac83} - 2q^{3}
+ a^{2}q^{\frac{10}{3}}
\nonumber\\
&\quad
 + a^{-2}q^{\frac{11}{3}}
- 2q^{4} + a^{4}q^{\frac{14}{3}} + \big(a^{-4}-a^{2}\big)q^{\frac{16}{3}}
+ \big(a^{4}-a^{-2}\big)q^{\frac{17}{3}} + q^{6}
\nonumber\\
&\quad
+ \big(a^{-4}-3a^{2}\big)q^{\frac{19}{3}} + \big(2a^{4}-3a^{-2}\big)q^{\frac{20}{3}} + 5q^{7} + \dots\,,
\end{align}
likewise verified up to and including order $q^{15}$.

Finally, at the boundary assignment $r_Q=0$, corresponding to $R=\tfrac12$,
both sides give, in the refined scheme,
\begin{align}
\label{rank1_expansion_rQzero}
\mathcal{I}^{USp(2)}\Big|_{r_Q=0}
= \mathcal{I}_1\Big|_{R=\frac12}
&= 1 + \big(a^{2}-1\big)q + \big(a^{2}-2+a^{-2}\big)q^{2} + \big(a^{2}-2+a^{-2}\big)q^{3}
\nonumber\\
&\quad
+ \big(a^{4}-2+a^{-2}\big)q^{4} + \big(a^{4}-a^{2}\big)q^{5} + \dots\,,
\end{align}
verified up to and including order $q^{7}$.
We have further checked the identity at $r_Q=\tfrac14$, i.e.\ $R=\tfrac14$,
where the grading has $N=4$, to order $q^{15}$.

Several features of \eqref{rank1_expansion_rQhalf}--\eqref{rank1_expansion_rQzero}
deserve comment.
First, the three series are related by the substitution rule below \eqref{yhat_ahat},
$a\to a\,q^{-1/12}$ between the first two and $a\to a\,q^{-1/4}$ from the first to the third,
so their separate agreement is a consistency check on the implementation rather than
independent evidence.  The fractional powers realize the grading \eqref{grading_rule} at $N=2$ and $N=3$ respectively,
with all powers integral at $r_Q=0$, exactly as predicted by the Abelian grading.
Second, the leading non-trivial terms are precisely the $|A|=1$ pair identified above:
$q^{1+r_Q}a^{2}$ and $q^{2-r_Q}a^{-2}$, which merge into $(a^{2}+a^{-2})\,q^{3/2}$ at $r_Q=\tfrac12$,
split into $a^{2}q^{4/3}$ and $a^{-2}q^{5/3}$ at $r_Q=\tfrac13$,
and sit at $a^{2}q$ and $a^{-2}q^{2}$ at $r_Q=0$.
Third, the term $-q$, common to both sides and carrying no $U(1)_a$ charge,
is naturally attributed to the conserved-current multiplet of the single $U(1)$ global symmetry,
consistent with the matching of global symmetries and the absence of any further flavor current.
Fourth, the boundary series \eqref{rank1_expansion_rQzero} collapses to $1$ at $a^{2}=1$,
term by term in $q$, in accordance with \eqref{trivial_index} at $R^{(r)}=+\tfrac12$,
as stated in the corollary \eqref{electric_trivial_prediction} above.
Finally, the first $|A|=2$ term appears at $q^{4+2r_Q}a^{4}$
rather than at the naive minimum $q^{2+2r_Q}a^{4}$ of \eqref{electric_m1_max}:
the top term of the $m=1$ sector cancels against other contributions of the same weight.
Within the verified range the coefficients of \eqref{rank1_expansion_rQhalf}
are symmetric under $a\to a^{-1}$ (equivalently $\hat y\to\hat y^{-1}$).
Pushing the expansion of the Abelian index $\mathcal{I}_1$ at $R^{(1)}=0$ further,
in yet another independent implementation with flux range verified sufficient by the
sector bounds \eqref{magnetic_Am1} and \eqref{magnetic_Ap1},
we find that the symmetry persists exactly through order $q^{25}$,
strongly suggesting that it is exact at rank one.
Neither theory has a manifest charge-conjugation symmetry, and the symmetry is indeed
absent at rank two already at order $q^{5/2}$, see \eqref{rank2_expansion_rQhalf}.
A natural reading of the symmetry is 3d mirror self-duality rather than a statement about
the amount of supersymmetry.  The tadpole theory of rank $r$ is $\mathcal{T}_{1,r}$ in the
notation of section~\ref{sec_general_rank0}, and the indices of that family are proposed in
\cite{Creutzig:2024ljv} to satisfy
$\mathcal{I}[\mathcal{T}_{\ell,n}](\hat y)=\mathcal{I}[\mathcal{T}_{n,\ell}](\hat y^{-1})$ under
3d mirror symmetry.  At $r=1$ the mirror is the theory itself, up to a parity transformation, so
the index is invariant under $\hat y\to\hat y^{-1}$.  For $r\ge2$ the mirror is a different
Abelian theory and no such invariance is implied.  This accounts both for the exact symmetry at
rank one and for its failure from rank two onwards, and does not indicate the amount of
supersymmetry.  Section~\ref{sec_level_pair} places the observation in the family as a whole.

It may be worth addressing an apparent tension with rank $0$.  Since both branches are
points, one might expect nothing to distinguish $SU(2)_H$ from $SU(2)_C$ and the index to be
symmetric under $\hat y\to\hat y^{-1}$ for every $r$.  The index does count local operators, but
not only the ones parametrizing the branches.  Rank $0$ constrains the branch rings alone,
equivalently the two specializations \eqref{trivial_index} at $y=1$ and $R^{(r)}=\pm\tfrac12$,
and away from those points the index is a non-trivial series counting operators of a quite
different kind.  The two contributions at $q^{3/2}$ found above are a case in point.  On the
symplectic side they are $J_{ab}Q^{a}\partial Q^{b}$ and $J^{ab}\bar\psi_{a}\bar\psi_{b}$, neither
of which is a branch operator, the meson vanishing identically by \eqref{meson_vanishes}.  There
is therefore no reason for the flavored index to be symmetric, and indeed
$\mathcal{T}_{1,2}$ and $\mathcal{T}_{2,1}$ are both rank $0$ while their indices are merely
related by $\hat y\to\hat y^{-1}$ rather than equal.  Since $U(1)_y$ is the axial combination of
the two R-symmetry Cartans, $\hat y\to\hat y^{-1}$ interchanges the roles of $SU(2)_H$ and
$SU(2)_C$, so the symmetry at $r=1$ reflects the coincidence of the two twists there and not
rank $0$, which holds for every $r$.

\subsubsection{$USp(4)_{\frac72}$ versus the $U(1)^2$ tadpole theory}
\label{sec_matching_rank2}
\paragraph{The symplectic theory}
The symplectic theory is $USp(4)_{7/2}$ with one fundamental $\mathbf{4}$ and $\mathcal{W}=0$.
This is the first rank at which the full non-Abelian structure enters the index, meaning
short roots and a Weyl group larger than $\Zb_2$.
The Weyl group is $\Zb_2^2\rtimes S_2$ of order $8$,
the fluxes are $(m_1,m_2)\in\Zb^2$,
and besides the long roots $\pm2e_{1,2}$ the short roots $\pm e_1\pm e_2$ appear,
contributing the four off-diagonal factors and the cross terms
$|m_1\pm m_2|$ in the monopole power.
The classical factor carries $2k=7$, so the integrand comes with the sign $(-1)^{m_1+m_2}$
and the theory is again a spin theory.
The vacuum counts \eqref{two_vacuum_counts} are $(1,3)$.
Explicitly,
\begin{align}
\label{ind_USp4}
\mathcal{I}^{USp(4)}
&= \frac{1}{8}\sum_{(m_1,m_2)\in\Zb^{2}} \oint
\prod_{i=1}^{2}\frac{ds_i}{2\pi i\, s_i}\,(-s_i)^{7m_i}\;
q^{-|m_1|-|m_2|-\frac12\big(|m_1+m_2|+|m_1-m_2|\big)}
\nonumber\\
&\quad\times
\prod_{i=1}^{2}\big(1-q^{|m_i|}s_i^{\pm2}\big)\;
\big(1-q^{\frac{|m_1+m_2|}{2}}(s_1 s_2)^{\pm1}\big)
\big(1-q^{\frac{|m_1-m_2|}{2}}(s_1 s_2^{-1})^{\pm1}\big)
\nonumber\\
&\quad\times
\big(q^{\frac{1-r_Q}{2}} a^{-1}\big)^{|m_1|+|m_2|}
\prod_{i=1}^{2}
\frac{(q^{1-\frac{r_Q}{2}+\frac{|m_i|}{2}}a^{-1} s_i^{\mp};q)_{\infty}}
     {(q^{\frac{r_Q}{2}+\frac{|m_i|}{2}}a\, s_i^{\pm};q)_{\infty}}\,,
\end{align}
where, as throughout, a superscript $\pm$ ($\mp$) means that both factors are multiplied.

\paragraph{The Abelian theory}
The rank-two tadpole theory has gauge group $U(1)^2$, level matrix
\begin{align}
K_2 = \begin{pmatrix} 2 & 2 \\ 2 & 4 \end{pmatrix}
= 2\,C(T_2)^{-1}\,,
\qquad
C(T_2) = \begin{pmatrix} 2 & -1 \\ -1 & 1 \end{pmatrix},
\qquad
\det K_2 = 4\,,
\end{align}
two charge-one chirals $\Phi^{(1)},\Phi^{(2)}$ with $\mathbf{Q}=\mathbb{I}$,
and the single monopole superpotential term
\begin{align}
\mathcal{W} = V_{(2,-1)}\,,
\qquad
\mathbf m_1 = C(T_2)\,\mathbf{e}_1 = (2,-1)\,.
\end{align}
One checks from \eqref{tadpole_monopole_gauge_charge} that since
$K_2\mathbf m_1 = (2,0)$ and $\big((\mathbf m_1)^{(J)}\big)^{+} = (2,0)$,
$q_I[V_{(2,-1)}] = 0$ so the superpotential is gauge invariant.
Similarly, from \eqref{tadpole_monopole_charges}, marginality $R[V_{(2,-1)}]=2(1-R^{(1)})=2$ enforces $R^{(1)}=0$,
leaving $R\equiv R^{(2)}$ free, in accordance with \eqref{Rassignment}.
The superpotential breaks $U(1)_{J_1}\times U(1)_{J_2}$ to the single $U(1)_y$ with charge
$A = m^{(1)}+2m^{(2)}$, i.e.\ $\vec A = (1,2)$.
The parity-symmetric level matrix $\widetilde K_2 = K_2 - \tfrac12\mathbb{I}$ has diagonal entries
$(\tfrac32,\tfrac72)$, half-integral as dictated by the parity anomaly.
The Bethe computation \eqref{bethe_count} gives $3$ vacua at generic $\eta$, matching the symplectic $n+1=3$.
The full index \eqref{ind_tadpole_full} specializes to
\begin{align}
\label{ind_tadpole2}
\mathcal{I}_2
&= \sum_{(m^{(1)},m^{(2)})\in\Zb^{2}} \oint
\prod_{I=1}^{2}\frac{ds^{(I)}}{2\pi i\,s^{(I)}}\;
\big(s^{(1)}\big)^{2m^{(1)}+2m^{(2)}}\,
\big(s^{(2)}\big)^{2m^{(1)}+4m^{(2)}}\;
\big(q^{-\frac12}y\big)^{m^{(1)}+2m^{(2)}}
\nonumber\\
&\qquad\times
\Big(\frac{-\,q^{\frac12}}{s^{(1)}}\Big)^{(m^{(1)})^{+}}
\frac{\big(q^{1+\frac{|m^{(1)}|}{2}}(s^{(1)})^{-1};q\big)_{\infty}}
     {\big(q^{\frac{|m^{(1)}|}{2}}s^{(1)};q\big)_{\infty}}
\nonumber\\
&\qquad\times
\Big(\frac{-\,q^{(1-R)/2}}{s^{(2)}}\Big)^{(m^{(2)})^{+}}
\frac{\big(q^{1-\frac{R}{2}+\frac{|m^{(2)}|}{2}}(s^{(2)})^{-1};q\big)_{\infty}}
     {\big(q^{\frac{R}{2}+\frac{|m^{(2)}|}{2}}s^{(2)};q\big)_{\infty}}\,.
\end{align}

\paragraph{The expansions}
As observed at the end of sections~\ref{sec_index} and \ref{sec_tadpole_index},
the leading $|A|=1$ contributions are rank independent:
on the Abelian side the cheapest fluxes are $-\mathbf{e}_1$ and $\mathbf{e}_2-\mathbf{e}_1$,
reproducing the pair $q^{1+r_Q}y$ and $q^{2-r_Q}y^{-1}$ of the previous subsection,
and on the symplectic side the same zero-flux invariants $J_{ab}Q^{a}\partial Q^{b}$ and
$J^{ab}\bar\psi_{a}\bar\psi_{b}$ are present.
The explicit expansion demonstrates where the rank first enters.
We have expanded \eqref{ind_USp4} and \eqref{ind_tadpole2} at the three R-charge assignments
listed below and find
perfect agreement.

At $r_Q=\tfrac12$, corresponding to $R=0$ on the Abelian side, both sides give
\begin{align}
\label{rank2_expansion_rQhalf}
\mathcal{I}^{USp(4)}\Big|_{r_Q=\frac12}
= \mathcal{I}_2\Big|_{R=0}
&= 1 - q + \big(a^{2}+a^{-2}\big)q^{\frac32} - 2q^{2} + a^{2}q^{\frac52}
+ \big({-1}+a^{-4}\big)q^{3}
\nonumber\\
&\quad
 + \big(a^{2}-a^{-2}\big)q^{\frac72}
+ a^{-4}q^{4} - \big(a^{2}+3a^{-2}\big)q^{\frac92}
+ \big(4+a^{4}+2a^{-4}\big)q^{5}
\nonumber\\
&\quad
 - \big(3a^{2}+5a^{-2}\big)q^{\frac{11}{2}}
+ \big(6+a^{4}+a^{-4}\big)q^{6}
\nonumber\\
&\quad
+ \big(a^{-6}-6a^{-2}-6a^{2}\big)q^{\frac{13}{2}}
+ \big(11+3a^{4}\big)q^{7} + \dots\,.
\end{align}

At $r_Q=\tfrac13$, corresponding to $R=\tfrac16$,
\begin{align}
\label{rank2_expansion_rQthird}
\mathcal{I}^{USp(4)}\Big|_{r_Q=\frac13}
= \mathcal{I}_2\Big|_{R=\frac16}
&= 1 - q + a^{2}q^{\frac43} + a^{-2}q^{\frac53} - 2q^{2} + a^{2}q^{\frac73} - q^{3}
+ \big(a^{2}+a^{-4}\big)q^{\frac{10}{3}}
\nonumber\\
&\quad
- a^{-2}q^{\frac{11}{3}} + \big(a^{-4}-a^{2}\big)q^{\frac{13}{3}}
+ \big(a^{4}-3a^{-2}\big)q^{\frac{14}{3}} + 4q^{5}
\nonumber\\
&\quad
+ \big(2a^{-4}-3a^{2}\big)q^{\frac{16}{3}}
+ \big(a^{4}-5a^{-2}\big)q^{\frac{17}{3}} + 6q^{6}
+ \big(a^{-4}-6a^{2}\big)q^{\frac{19}{3}}
\nonumber\\
&\quad
+ \big(3a^{4}-6a^{-2}\big)q^{\frac{20}{3}}
+ \big(11+a^{-6}\big)q^{7} + \dots\,,
\end{align}
both verified up to and including order $q^{7}$.

At the boundary assignment $r_Q=0$, corresponding to $R=\tfrac12$, in the refined scheme,
\begin{align}
\label{rank2_expansion_rQzero}
\mathcal{I}^{USp(4)}\Big|_{r_Q=0}
= \mathcal{I}_2\Big|_{R=\frac12}
&= 1 + \big(a^{2}-1\big)q + \big(a^{2}-2+a^{-2}\big)q^{2} + \big(a^{2}-1\big)q^{3}
\nonumber\\
&\quad
+ \big(a^{4}-a^{2}-a^{-2}+a^{-4}\big)q^{4}
+ \big(4+a^{4}-3a^{2}-3a^{-2}+a^{-4}\big)q^{5} + \dots\,,
\end{align}
verified up to and including order $q^{5}$ and collapsing to $1$ at $a^{2}=1$.

We comment on several features of the expansions.
The three series are related by the substitution rule below \eqref{yhat_ahat},
and their fractional powers realize the grading \eqref{grading_rule} at $N=2$ and $N=3$,
with integral powers at $r_Q=0$.
All begin exactly like the corresponding rank-one series
\eqref{rank1_expansion_rQhalf}--\eqref{rank1_expansion_rQzero},
confirming the rank independence of the leading $|A|\le1$ contributions
and, through the term $-q$, the presence of a single flavor current at this rank as well.
The first rank-dependent coefficients appear immediately afterwards, e.g.\ at order
$q^{5/2}$ for $r_Q = \tfrac{1}{2}$.
The term $a^{-2}q^{5/2}$ of rank one is absent here, and at order $q^{3}$ the coefficient is
$-1+a^{-4}$ instead of $-2$.
At $r_Q=0$ the two statements merge into the $q^{3}$ coefficient,
$a^{2}-1$ in \eqref{rank2_expansion_rQzero} against $a^{2}-2+a^{-2}$ in \eqref{rank1_expansion_rQzero}.
The rank is therefore visible at rather low order, and in the negative-$A$ sector,
not only through the top $a$-charges of the $m=(1,0)$ flux sector estimated in
\eqref{electric_m1_min} and \eqref{electric_m1_max}.
Note also that, in contrast with \eqref{rank1_expansion_rQhalf},
the series \eqref{rank2_expansion_rQhalf} is not symmetric under $a\to a^{-1}$,
the asymmetry setting in already at order $q^{5/2}$.

The term $a^{-4}q^{3}$ (more generally $q^{\,3+2R}\,y^{-2}$, i.e.\ $a^{-4}q^{10/3}$ at $R=\tfrac16$)
has a direct identification on the Abelian side. It is the dressed monopole of the example in
section~\ref{sec_tadpole_monopoles} with $c=2$:
$V_{-\mathbf{e}_2}$ dressed with two modes of $\phi^{(1)}$ and four of $\phi^{(2)}$,
carrying $A=-2$ and $R=4R^{(2)}$.
It has no counterpart at rank one, where the first $A=-2$ contribution occurs only at order $q^{5}$.
On the symplectic side, the natural candidate involves four $\bar\psi$ letters,
of the type $\big(J^{ab}\bar\psi_{a}\bar\psi_{b}\big)^{2}$, with weight
$q^{4-2r_Q}a^{-4} = q^{3}a^{-4}$ at $r_Q=\tfrac12$.
Its lowest component is non-vanishing only for $n\ge2$, by Fermi statistics of the finitely many
lowest $\bar\psi$ modes, consistent with its absence at rank one.
As at rank one, this identifies index contributions rather than a complete operator map.

\subsubsection{$USp(6)_{\frac92}$ versus the $U(1)^3$ tadpole theory}
\label{sec_matching_rank3}
\paragraph{The symplectic theory}
The symplectic theory is $USp(6)_{9/2}$ with one fundamental $\mathbf{6}$ and $\mathcal{W}=0$.
The Weyl group has order $2^3\,3! = 48$, the fluxes are $(m_1,m_2,m_3)\in\Zb^3$,
and there are now three pairs $i<j$ of short-root contributions.
The sign is $(-1)^{m_1+m_2+m_3}$ from $2k=9$,
and the vacuum counts \eqref{two_vacuum_counts} are $(1,4)$.
The full index is
\begin{align}
\label{ind_USp6}
\mathcal{I}^{USp(6)}
&= \frac{1}{48}\sum_{(m_1,m_2,m_3)\in\Zb^{3}} \oint
\prod_{i=1}^{3}\frac{ds_i}{2\pi i\, s_i}\,(-s_i)^{9m_i}\;
q^{-\sum_{i}|m_i|-\frac12\sum_{i<j}\big(|m_i+m_j|+|m_i-m_j|\big)}
\nonumber\\
&\quad\times
\prod_{i=1}^{3}\big(1-q^{|m_i|}s_i^{\pm2}\big)
\prod_{1\le i<j\le3}
\big(1-q^{\frac{|m_i+m_j|}{2}}(s_i s_j)^{\pm1}\big)
\big(1-q^{\frac{|m_i-m_j|}{2}}(s_i s_j^{-1})^{\pm1}\big)
\nonumber\\
&\quad\times
\big(q^{\frac{1-r_Q}{2}} a^{-1}\big)^{\sum_i|m_i|}
\prod_{i=1}^{3}
\frac{(q^{1-\frac{r_Q}{2}+\frac{|m_i|}{2}}a^{-1} s_i^{\mp};q)_{\infty}}
     {(q^{\frac{r_Q}{2}+\frac{|m_i|}{2}}a\, s_i^{\pm};q)_{\infty}}\,.
\end{align}

\paragraph{The Abelian theory}
The rank-three tadpole theory has gauge group $U(1)^3$, level matrix
\begin{align}
K_3 = \begin{pmatrix} 2 & 2 & 2 \\ 2 & 4 & 4 \\ 2 & 4 & 6 \end{pmatrix}
= 2\,C(T_3)^{-1}\,,
\qquad
C(T_3) = \begin{pmatrix} 2 & -1 & 0 \\ -1 & 2 & -1 \\ 0 & -1 & 1 \end{pmatrix},
\qquad
\det K_3 = 8\,,
\end{align}
three charge-one chirals with $\mathbf{Q}=\mathbb{I}$, and the monopole superpotential
\begin{align}
\mathcal{W} = V_{(2,-1,0)} + V_{(-1,2,-1)}\,,
\qquad
\mathbf m_I = C(T_3)\,\mathbf{e}_I\,,\quad I \in \{1, 2\}\,.
\end{align}
Both terms pass the checks of section~\ref{sec_magnetic_level_choice}:
$K_3\mathbf m_1 = (2,0,0)$ and $K_3\mathbf m_2 = (0,2,0)$ with
$\big((\mathbf m_I)^{(J)}\big)^{+} = 2\delta_{IJ}$,
so both are gauge invariant, and marginality enforces $R^{(1)}=R^{(2)}=0$ with $R\equiv R^{(3)}$ free.
The two superpotential terms break $U(1)_J^3$ to $U(1)_y$ with charge
$A = m^{(1)}+2m^{(2)}+3m^{(3)}$, i.e.\ $\vec A = (1,2,3)$.
The diagonal of $\widetilde K_3 = K_3-\tfrac12\mathbb{I}$ is $(\tfrac32,\tfrac72,\tfrac{11}2)$,
and the Bethe count \eqref{bethe_count} gives $4=n+1$ vacua.
The full index specializes to
\begin{align}
\label{ind_tadpole3}
\mathcal{I}_3
&= \sum_{\mathbf m\in\Zb^{3}} \oint
\prod_{I=1}^{3}\frac{ds^{(I)}}{2\pi i\,s^{(I)}}\;
\big(s^{(1)}\big)^{2(m^{(1)}+m^{(2)}+m^{(3)})}\,
\big(s^{(2)}\big)^{2m^{(1)}+4m^{(2)}+4m^{(3)}}\,
\big(s^{(3)}\big)^{2m^{(1)}+4m^{(2)}+6m^{(3)}}
\nonumber\\
&\qquad\times
\big(q^{-\frac12}y\big)^{m^{(1)}+2m^{(2)}+3m^{(3)}}
\prod_{I=1}^{2}
\Big(\frac{-\,q^{\frac12}}{s^{(I)}}\Big)^{(m^{(I)})^{+}}
\frac{\big(q^{1+\frac{|m^{(I)}|}{2}}(s^{(I)})^{-1};q\big)_{\infty}}
     {\big(q^{\frac{|m^{(I)}|}{2}}s^{(I)};q\big)_{\infty}}
\nonumber\\
&\qquad\times
\Big(\frac{-\,q^{(1-R)/2}}{s^{(3)}}\Big)^{(m^{(3)})^{+}}
\frac{\big(q^{1-\frac{R}{2}+\frac{|m^{(3)}|}{2}}(s^{(3)})^{-1};q\big)_{\infty}}
     {\big(q^{\frac{R}{2}+\frac{|m^{(3)}|}{2}}s^{(3)};q\big)_{\infty}}\,.
\end{align}

\paragraph{The expansions}
The pattern of the previous subsection continues.
The $|A|\le1$ terms of \eqref{ind_tadpole3} coincide with those of $\mathcal{I}_1$ and $\mathcal{I}_2$,
the cheapest positive flux now being $\mathbf m = \mathbf{e}_3-\mathbf{e}_2$.
We have expanded \eqref{ind_USp6} and \eqref{ind_tadpole3} at the same three assignments
as before and find perfect agreement.
At $r_Q=\tfrac12$, corresponding to $R=0$, both sides give
\begin{align}
\label{rank3_expansion_rQhalf}
\mathcal{I}^{USp(6)}\Big|_{r_Q=\frac12}
= \mathcal{I}_3\Big|_{R=0}
&= 1 - q + \big(a^{2}+a^{-2}\big)q^{\frac32} - 2q^{2} + a^{2}q^{\frac52}
+ \big(a^{-4}-1\big)q^{3}
\nonumber\\
&\quad
 + \big(a^{2}-a^{-2}\big)q^{\frac72}
+ \big(a^{-6}-2a^{-2}-a^{2}\big)q^{\frac92} + \big(4+a^{4}\big)q^{5}
\nonumber\\
&\quad
+ \big(a^{-6}-3a^{-2}-3a^{2}\big)q^{\frac{11}{2}}
+ \big(5+a^{4}-2a^{-4}\big)q^{6}
\nonumber\\
&\quad
+ \big(2a^{-6}-2a^{-2}-6a^{2}\big)q^{\frac{13}{2}}
+ \big(9+3a^{4}-4a^{-4}\big)q^{7} + \dots\,,
\end{align}
in which the coefficient of $q^{4}$ vanishes identically.

At $r_Q=\tfrac13$, corresponding to $R=\tfrac16$,
\begin{align}
\label{rank3_expansion_rQthird}
\mathcal{I}^{USp(6)}\Big|_{r_Q=\frac13}
= \mathcal{I}_3\Big|_{R=\frac16}
&= 1 - q + a^{2}q^{\frac43} + a^{-2}q^{\frac53} - 2q^{2} + a^{2}q^{\frac73} - q^{3}
+ \big(a^{2}+a^{-4}\big)q^{\frac{10}{3}}
\nonumber\\
&\quad
 - a^{-2}q^{\frac{11}{3}}
- a^{2}q^{\frac{13}{3}} + \big(a^{4}-2a^{-2}\big)q^{\frac{14}{3}} + \big(4+a^{-6}\big)q^{5}
- 3a^{2}q^{\frac{16}{3}}
\nonumber\\
&\quad
 + \big(a^{4}-3a^{-2}\big)q^{\frac{17}{3}}
+ \big(5+a^{-6}\big)q^{6} - \big(6a^{2}+2a^{-4}\big)q^{\frac{19}{3}}
\nonumber\\
&\quad
+ \big(3a^{4}-2a^{-2}\big)q^{\frac{20}{3}} + \big(9+2a^{-6}\big)q^{7} + \dots\,,
\end{align}
both verified up to and including order $q^{7}$,
the two series again being related by $a\to a\,q^{-1/12}$.

At the boundary assignment $r_Q=0$, corresponding to $R=\tfrac12$, in the refined scheme,
\begin{align}
\label{rank3_expansion_rQzero}
\mathcal{I}^{USp(6)}\Big|_{r_Q=0}
= \mathcal{I}_3\Big|_{R=\frac12}
&= 1 + \big(a^{2}-1\big)q + \big(a^{2}-2+a^{-2}\big)q^{2} + \big(a^{2}-1\big)q^{3}
\nonumber\\
&\quad
+ \big(a^{4}-a^{2}-a^{-2}+a^{-4}\big)q^{4}
+ \big(4+a^{4}-3a^{2}-2a^{-2}\big)q^{5} + \dots\,,
\end{align}
verified up to and including order $q^{5}$ and collapsing to $1$ at $a^{2}=1$.

Through order $q^{7/2}$ the series coincide with the rank-two ones,
including the $c=2$ dressed-monopole term $a^{-4}q^{4-2r_Q}$.
The first coefficient distinguishing $n=3$ from $n=2$ appears at order $q^{5-2r_Q}$:
the term $a^{-4}q^{\,5-2r_Q}$, which is $a^{-4}q^{4}$ at $r_Q=\tfrac12$, $a^{-4}q^{13/3}$
at $r_Q=\tfrac13$ and $a^{-4}q^{5}$ at $r_Q=0$, is present with coefficient $1$ at rank two and
absent at rank three,
in precise parallel with the rank-one/rank-two transition,
where the term $a^{-2}q^{3-r_Q}$ dropped.
It is accompanied by an $a^{-2}$ difference at order $q^{5-r_Q}$,
where the coefficient is $-3$ at rank two and $-2$ at rank three,
as visible at $q^{9/2}$ in \eqref{rank2_expansion_rQhalf} and \eqref{rank3_expansion_rQhalf}.
At the boundary assignment the two differences collide at order $q^{5}$.
The boundary series \eqref{rank2_expansion_rQzero} and \eqref{rank3_expansion_rQzero}
coincide through order $q^{4}$ and differ precisely in the $a^{-4}$ and $a^{-2}$
coefficients of $q^{5}$.
Both differences are instances of a single pattern, which we now describe.

Introducing the weight
\begin{align}
w \;\equiv\; q^{\,2-r_Q}\,a^{-2}
\end{align}
of $J^{ab}\bar\psi_{a}\bar\psi_{b}$, the expansions show that,
within the verified range,
the coefficient of $w^{c}$ equals $1$ for all $c\le n$ and vanishes for $c>n$,
while the coefficient of $q\cdot w^{c}$ equals $1$ precisely at rank $n=c$.
The matching of the tower at $w^{c}$ can be understood precisely.
On the Abelian side it is the dressed monopole of the example of section~\ref{sec_tadpole_monopoles},
$V_{-\mathbf{e}_c}$ dressed with $2\min(I,c)$ modes of $\phi^{(I)}$,
whose weight is exactly $w^{c}$ and which exists if and only if $c\le r$.
On the symplectic side it is $\big(J^{ab}\bar\psi_{a}\bar\psi_{b}\big)^{c}$,
non-vanishing if and only if $c\le n$ by Fermi statistics of the $2n$ lowest $\bar\psi$ modes.
The $c=3$ member is the new term $a^{-6}q^{6-3r_Q}$
visible in \eqref{rank3_expansion_rQhalf} and \eqref{rank3_expansion_rQthird}
and absent for the lower ranks.
This single tower therefore detects the rank on both sides of the duality
and its matching requires $n=r$ exactly,
realizing very concretely the off-diagonal statement implicit in \eqref{the_duality}.
Specifically, within the family $\{\mathcal{I}^{USp(2n)}\}_{n\le3}$ versus $\{\mathcal{I}_r\}_{r\le3}$
the indices agree along the diagonal $n=r$
and first disagree off it at $q\,w^{\min(n,r)}$, that is at order
$q^{\,1+\min(n,r)(2-r_Q)}$, which is $q^{\,3-r_Q}$ for $(n,r)=(1,2)$ and $q^{\,5-2r_Q}$ for
$(n,r)=(2,3)$.
\subsubsection{$USp(2)_{\frac72}$ versus $\mathcal{T}_{2,1}$}
\label{sec_matching_12}

The three comparisons above all sit at $\ell=1$.  We now raise the level.  The first case is
$(n,\ell)=(1,2)$, where the symplectic theory is $USp(2)_{7/2}$ and the Abelian theory is
$\mathcal{T}_{2,1}$, which by section~\ref{sec_general_rank0} is $U(1)^{3}$ at level $C(A_{3})$
with one charge-one chiral at each node and the surviving flavor charge $(1,-1,1)$ of
\eqref{S_general}.  The two ranks now differ, and the Abelian side is no longer a tadpole
theory, its superpotential requiring the dressed monopoles of section~\ref{sec_general_W}.  The
symplectic index is \eqref{ind_SU2} at $2k=7$, and \eqref{ind_general_full} specializes to
\begin{align}
\label{ind_T21}
\mathcal{I}\big[\mathcal{T}_{2,1}\big]
= \sum_{\mathbf m\in\Zb^{3}} \oint \prod_{I=1}^{3}\frac{dz_{I}}{2\pi i\,z_{I}}\;
z_{I}^{\,(\kappa\mathbf m)_{I}}\,
\big(q^{-\frac12}y\big)^{m_{1}-m_{2}+m_{3}}
\prod_{I=1}^{3}
\Big(\frac{-\,q^{\frac12}}{z_{I}}\Big)^{m_{I}^{+}}
\frac{\big(q^{1+\frac{|m_{I}|}{2}}z_{I}^{-1};q\big)_{\infty}}
     {\big(q^{\frac{|m_{I}|}{2}}z_{I};q\big)_{\infty}}\,,
\end{align}
with $\kappa=C(A_{3})$, so that
$(\kappa\mathbf m)_{1}=2m_{1}-m_{2}$, $(\kappa\mathbf m)_{2}=-m_{1}+2m_{2}-m_{3}$ and
$(\kappa\mathbf m)_{3}=-m_{2}+2m_{3}$, and with the contours $|z_{I}|=1$.  The R-charges are set
to zero here, which by \eqref{Rcharge_solution} is the choice $\alpha=0$.

At $r_{Q}=\tfrac12$, which by \eqref{family_dictionary} is $\alpha=0$, both indices are
\begin{align}
\label{ell2_expansion}
1-q+(y+y^{-1})q^{\frac32}-2q^{2}+y^{-1}q^{\frac52}+(y^{2}-1)q^{3}
+(y^{-1}-y)q^{\frac72}+y^{2}q^{4}
\nonumber\\
{}-(y^{-1}+3y)q^{\frac92}+(y^{-2}+4+2y^{2})q^{5}+\dots
\end{align}
in the variable $y=a^{2}$.  The series is not symmetric under $y\to y^{-1}$, in accordance with
$\ell\neq n$ and with the table of section~\ref{sec_level_pair}.  The Abelian side required
fluxes up to $|m_{I}|\le14$ at this order.

Changing frame illustrates the point.  At $r_{Q}=\tfrac13$ the dictionary
\eqref{family_dictionary} gives $\alpha=\tfrac16$, and rather than recompute the Abelian index at
that value we use \eqref{yhat_general} to trade it for the frame $\alpha=0$ already computed,
which by $\hat y=y\,q^{-\alpha}$ amounts to $y\to a^{2}q^{-1/6}$.  Both sides then become
\begin{align}
\label{ell2_expansion_rQthird}
1-q+a^{2}q^{\frac43}+a^{-2}q^{\frac53}-2q^{2}+(a^{-2}+a^{4})q^{\frac83}-q^{3}+\dots
\end{align}
The fractional powers realize the grading \eqref{grading_rule} with the same rule as at
$\ell=1$, the power of $q$ being congruent to $r_{Q}A$ modulo one with $A$ the $U(1)_{y}$
charge, which is a check on the identification of the flavor symmetry as much as on the
coefficients.

\subsubsection{$USp(2)_{\frac92}$ versus $\mathcal{T}_{3,1}$}
\label{sec_matching_13}

At $(n,\ell)=(1,3)$ the Abelian theory is $\mathcal{T}_{3,1}$, which is $U(1)^{5}$ at level
$C(A_{5})$.  The symplectic index is \eqref{ind_SU2} at $2k=9$, and the Abelian index is
\eqref{ind_T21} with the three nodes replaced by five, $\kappa$ by $C(A_{5})$ and the exponent of
$q^{-\frac12}y$ by $m_{1}-m_{2}+m_{3}-m_{4}+m_{5}$.  At $r_{Q}=\tfrac12$ both indices are
\begin{align}
\label{ell3_expansion}
1-q+(y+y^{-1})q^{\frac32}-2q^{2}+y^{-1}q^{\frac52}+(y^{2}-1)q^{3}
+(y^{-1}-y)q^{\frac72}
\nonumber\\
{}-(y^{-1}+2y-y^{3})q^{\frac92}+(y^{-2}+4)q^{5}+\dots
\end{align}
The coefficient of $q^{4}$ vanishes on both sides.  This is what separates $\ell=3$ from
$\ell=2$, where \eqref{ell2_expansion} has $y^{2}q^{4}$, so the comparison is sensitive to the
rank of the Abelian gauge group and not only to its leading behaviour.  At $r_{Q}=\tfrac13$ the
two sides agree as well, and the expansion coincides with \eqref{ell2_expansion_rQthird} through
the orders shown, the two levels first differing at $q^{10/3}$.

\subsubsection{$USp(4)_{\frac92}$ versus $\mathcal{T}_{2,2}$}
\label{sec_matching_22}

This is the first case in which neither side has rank one.  The symplectic theory is
$USp(4)_{9/2}$ and the Abelian theory is $\mathcal{T}_{2,2}$, which is $U(1)^{6}$ at level
$C(T_{2})^{-1}\otimes C(A_{3})$, the first member of the family in which both factors of the
tensor product are non-trivial.  The symplectic index is \eqref{ind_USp4} at $2k=9$, and the
Abelian one is \eqref{ind_general_full} with the six nodes labelled by $(i,a)$,
$i=1,2$ and $a=1,2,3$, the level matrix
\begin{align}
\label{kappa_T22}
\kappa=\begin{pmatrix} C(A_{3}) & C(A_{3})\\ C(A_{3}) & 2\,C(A_{3})\end{pmatrix},
\qquad
C(A_{3})=\begin{pmatrix}2&-1&0\\-1&2&-1\\0&-1&2\end{pmatrix},
\end{align}
and the flavor charge
$S=(m_{11}-m_{12}+m_{13})+2(m_{21}-m_{22}+m_{23})$ from \eqref{S_general}.  Writing
$\mathbf m_{i}=(m_{i1},m_{i2},m_{i3})$ for the fluxes of the two rows, the classical factor is
governed by
\begin{align}
\label{kappa_m_T22}
(\kappa\mathbf m)_{1a}=\big(C(A_{3})\,(\mathbf m_{1}+\mathbf m_{2})\big)_{a}\,,
\qquad
(\kappa\mathbf m)_{2a}=\big(C(A_{3})\,(\mathbf m_{1}+2\mathbf m_{2})\big)_{a}\,,
\end{align}
the coefficients $1,1,1,2$ being the entries $\min(i,j)$ of $C(T_{2})^{-1}$, so that
\begin{align}
\label{ind_T22}
\mathcal{I}\big[\mathcal{T}_{2,2}\big]
&= \sum_{\mathbf m\in\Zb^{6}} \oint \prod_{i=1}^{2}\prod_{a=1}^{3}
\frac{dz_{ia}}{2\pi i\,z_{ia}}\;
z_{ia}^{\,(\kappa\mathbf m)_{ia}}\;
\big(q^{-\frac12}y\big)^{S}
\nonumber\\
&\qquad \prod_{i=1}^{2}\prod_{a=1}^{3}
\Big(\frac{-\,q^{\frac12}}{z_{ia}}\Big)^{(m_{ia})^{+}}
\frac{\big(q^{\,1+\frac{|m_{ia}|}{2}}z_{ia}^{-1};q\big)_{\infty}}
     {\big(q^{\,\frac{|m_{ia}|}{2}}z_{ia};q\big)_{\infty}}\,,
\end{align}
at $\alpha=0$.  This is the first index in the paper that is not of tadpole type, both factors of
\eqref{kappa_family} being non-trivial, and the six-fold flux sum is what makes it the most
computationally expensive of our checks.  At $r_{Q}=\tfrac12$ both indices are
\begin{align}
\label{ell22_expansion}
1-q+(y+y^{-1})q^{\frac32}-2q^{2}+(y^{2}+y^{-2})q^{3}-(y+y^{-1})q^{\frac72}
\nonumber\\
{}+(y^{2}+1+y^{-2})q^{4}-3(y+y^{-1})q^{\frac92}+\big(2y^{2}+5+2y^{-2}\big)q^{5}+\dots
\end{align}
and are invariant under $y\to y^{-1}$, as they must be since $\ell=n$ here.  Note that the
coefficient of $q^{5/2}$ vanishes, whereas every tadpole index is non-zero at that order, so
$USp(4)_{9/2}$ is dual to no member of the family treated in
sections~\ref{sec_matching_rank1} to \ref{sec_matching_rank3}.  The Abelian side required
fluxes up to $|m_{I}|\le10$ at this order.

We have not attempted $(n,\ell)$ with $n\ge2$ and $\ell\ge3$, where the symplectic index
becomes computationally expensive for the same reason that the Weyl group and the short roots make
$\mathcal{I}^{USp(6)}$ expensive above.

We remark on how these Abelian indices were evaluated, since the flux sums converge slowly
and a direct expansion in several gauge fugacities is expensive.  Because the gauge group is
Abelian and each chiral carries charge one under a single factor, the one-loop contribution of a
node depends only on its own fugacity and flux, and the nodes are coupled only through the
classical factor $z_{I}^{(\kappa\mathbf m)^{(I)}}$.  Expanding a node by the $q$-binomial theorem
and reading off a single coefficient reduces the contour integrals to one-variable residues,
after which the flux sum multiplies $n$ ordinary $q$-series.  The on-contour poles of
section~\ref{sec_tadpole_index} need no separate treatment in this form, since only one
coefficient of each node is ever required.  Convergence in the flux is the limiting factor, and
it is slower at higher orders in $q$, the coefficient of $q^{5}$ for $\mathcal{T}_{2,2}$
stabilizing only at $|m_{I}|\le10$ while $q^{4}$ has already done so at $|m_{I}|\le9$.  Each
quoted coefficient is stable at the flux range used, and agrees with a symplectic side that has
itself converged at a much smaller range.

\subsubsection{Charge conjugation and the level pairs}
\label{sec_matching_conj}

One further consequence of \eqref{family_main} is present at all ranks, and it concerns the
symmetry discussed in section~\ref{sec_matching_rank1}.  By \eqref{cgk_mirror} the mirror
of $\mathcal{T}_{\ell,n}$ is $\overline{\mathcal{T}}_{n,\ell}$, so the theory is carried to its
own parity conjugate precisely when $\ell=n$, and \eqref{family_main} then predicts that the index of
$USp(2n)_{n+\frac12+\ell}$ is invariant under $\hat y\to\hat y^{-1}$ if and only if $\ell=n$.
That is exactly the pattern tabulated in section~\ref{sec_level_pair}, and it settles the
question left open in section~\ref{sec_matching_rank1}.

Composing \eqref{family_main} with the mirror map explains the relation
\eqref{level_pair} observed in section~\ref{sec_level_pair}.  Since $\mathcal{T}_{\ell,n}$ and
$\overline{\mathcal{T}}_{n,\ell}$ are mirror, the symplectic theories at $(n,\ell)$ and
$(\ell,n)$, which by \eqref{family_main} are dual to those two Abelian theories and which carry
the same level $k=n+\ell+\tfrac12$, must have equal indices up to the inversion of $\hat y$.
The self-conjugate cases $n=\ell$ are the invariance of the index under that inversion, which is
the diagonal of the table there.  That a relation between two symplectic theories of different
rank should follow in this way is a check of \eqref{family_main} rather than a new result, since
by section~\ref{sec_level_pair} the same relation is the Giveon-Kutasov duality with a single
fundamental.  The two routes to it are independent, one through the Abelian theories and their
mirror symmetry and the other through a Seiberg-like duality of the symplectic theories
themselves, and their agreement is among the better pieces of evidence we have.

It is worth asking how the two descriptions can be consistent, since a Seiberg-like duality
and a mirror symmetry are ordinarily different things.  Assuming \eqref{family_main}, the
square
\begin{align}
\label{commuting_square}
\begin{array}{ccc}
USp(2n)_{k} & \longleftrightarrow & \overline{USp(2\ell)_{k}}\\[4pt]
\updownarrow & & \updownarrow\\[4pt]
\mathcal{T}_{\ell,n} & \longleftrightarrow & \overline{\mathcal{T}}_{n,\ell}
\end{array}
\qquad k=n+\ell+\tfrac12\,,
\end{align}
commutes, the vertical arrows being \eqref{family_main} and the horizontal ones the
Giveon-Kutasov duality above and the mirror symmetry \eqref{cgk_mirror}.  The bars are the
parity conjugation of \eqref{parity_conjugate}, and they appear on both sides for the same
reason.  On the Abelian side each node carries an odd number of charged fermions, and on the
symplectic side the level is half-integral by section~\ref{sec:electric_level}, so parity is
broken in both descriptions and neither bar can be dropped.  The same operation is therefore a
Seiberg-like duality read on the symplectic side and a mirror symmetry read on the Abelian
side.

The reason is that the symmetry the two constructions act on is not the same kind of object
in the two descriptions.  A Giveon-Kutasov duality acts on the flavor symmetry of the
fundamentals, and for $2N_{f}\ge2$ that is a genuine flavor symmetry, unrelated to the
R-symmetry, so the duality is not a mirror symmetry.  At $2N_{f}=1$ there is no non-Abelian
flavor symmetry at all, and the surviving $U(1)_{a}$ is not a flavor symmetry either.  By the
$\mathcal{N}=4$ supersymmetry of section~\ref{sec_electric} it is the axial combination
$R_{H}-R_{C}$ of the two R-symmetry Cartans, as we used in
section~\ref{sec_matching_rank1}.  The charge conjugation that accompanies the Giveon-Kutasov
duality, $\hat a\to\hat a^{-1}$, is then the exchange of $SU(2)_{H}$ with $SU(2)_{C}$, which is
what mirror symmetry does.  The two names describe one operation because the symmetry being
inverted is a flavor symmetry in one reading and an R-symmetry in the other.

This is the same degeneracy that we met in the introduction, where the duality
\eqref{family_main} could not be placed in either of the two standard classes.  For theories of
rank $0$ with $\mathcal{N}=4$ supersymmetry the distinction between the classes is not clear,
since the Higgs and Coulomb branches carry no information and all that mirror symmetry can act
on is the axial R-symmetry, which is exactly what a Seiberg-like duality with a single
fundamental acts on as well.

What we have not done is to carry the analysis of sections~\ref{sec_matching_symmetries}
and \ref{sec_vacuum_counting} into the larger family.  The map of operators, and the constraints
of section~\ref{sec_dual_constraints}, which at $\ell=1$ nearly determine the dual, have not been
worked out for $\ell>1$, and the Bethe vacua of $\mathcal{T}_{\ell,n}$ are known only by their
number.  We will return to this.

\subsection*{Acknowledgements}
We would like to thank Yu Nakayama for discussions on related issues. 
The work of TO was supported by the Startup Funding no.\ 4007012317 of the Southeast University. 
The work of DJS was supported in part by the STFC Consolidated grant ST/T000708/1.

\appendix

\section{Bethe vacua of the tadpole theory}
\label{app_bethe}
We describe the details of the computation behind \eqref{bethe_count}.

\subsection{The system}
Clearing the denominator in \eqref{bethe_tadpole} turns the Bethe equations into $r$ polynomial equations
\begin{align}
\label{app_bethe_poly}
E_I(x;\eta) \equiv \eta^{\,I}\prod_{J=1}^{r} x_J^{\,2\min(I,J)} + x_I - 1 = 0,
\qquad I = 1,\dots,r,
\end{align}
in the variables $x_1,\dots,x_r$ and the single FI fugacity $\eta$.  Explicitly, for the
first few ranks,
\begin{align}
r=1:\quad & \eta\,x_1^{2}+x_1-1=0,
\nonumber\\
r=2:\quad & \eta\,x_1^{2}x_2^{2}+x_1-1=0,
\qquad
\eta^{2}x_1^{2}x_2^{4}+x_2-1=0,
\\
r=3:\quad & \eta\,x_1^{2}x_2^{2}x_3^{2}+x_1-1=0,
\qquad
\eta^{2}x_1^{2}x_2^{4}x_3^{4}+x_2-1=0,
\qquad
\eta^{3}x_1^{2}x_2^{4}x_3^{6}+x_3-1=0.
\nonumber
\end{align}
The exponent of $x_J$ in $E_I$ is $(K_r)_{IJ}=2\min(I,J)$, 
so the pattern of exponents is the tadpole level matrix itself.

The variables $x_I=e^{2\pi i u_I}$ are exponentiated Coulomb branch coordinates, 
so they take values in $\mathbb{C}^*=\mathbb{C}\setminus \{0\}$. 
The holonomy part of $u_I$ is periodic, and the two punctures $x_I=0,\infty$ are the limits $\beta\sigma^{(I)}\rightarrow \mp\infty$ 
 in which the vacuum escapes to infinity along the Coulomb branch. 
 Bethe vacua are therefore the solutions of \eqref{app_bethe_poly} lying in the algebraic torus $(\mathbb{C}^{*})^{r}$, not merely in $\mathbb{C}^{r}$. 
We impose this by saturation, adjoining one variable $t$ and one equation,
\begin{align}
\label{app_saturation}
t\prod_{I=1}^{r}x_I - 1 = 0,
\end{align}
so that the affine variety of the enlarged ideal $\mathcal{I}_r = \big(E_1,\dots,E_r,\;t\prod_I x_I-1\big) \subset \mathbb{C}[x_1,\dots,x_r,t]$
is in bijection with the solutions of \eqref{app_bethe_poly} in $(\mathbb{C}^{*})^{r}$.
For \eqref{app_bethe_poly} the saturation is in fact harmless, 
since $x_I=0$ for any $I$ already forces $E_I=-1\neq0$. 
We keep it only to make the count manifestly a count in the torus.

\subsection{Counting}
\label{app_counting}
For generic $\eta$ the ideal $\mathcal{I}_r$ is zero dimensional, and the number of solutions counted with multiplicity is
\begin{align}
\#\{\text{solutions}\} = \dim_{\mathbb{C}} \mathbb{C}[x_1,\dots,x_r,t]/\mathcal{I}_r,
\end{align}
which we evaluate by computing a Gr\"obner basis of $\mathcal{I}_r$ in lexicographic order and counting the standard monomials, 
i.e.\ the monomials not divisible by any leading monomial of the basis. 
Carrying this out for $\eta$ a fixed generic rational number gives
\begin{align}
\label{app_counts}
\begin{array}{c|cccc}
r & 1 & 2 & 3 & 4 \\ \hline
\dim_{\mathbb{C}}\,\mathbb{C}[x,t]/\mathcal{I}_r & 2 & 3 & 4 & 5
\end{array}
\end{align}
that is, $r+1$ in each case, in agreement with the general proof of the
next subsection, which the Gr\"obner computation confirms by an independent, explicit
calculation.

Concretely, eliminating $x_2,\dots,x_r$ in the lexicographic order
$x_r>\dots>x_2>t>x_1$ leaves a single univariate generator, which is the eliminant
$\mathcal{E}_r$ of \eqref{eliminant} below.  The standard monomials are then
$1,x_1,\dots,x_1^{\,r}$, whence $\dim_{\mathbb{C}}\mathbb{C}[x,t]/\mathcal{I}_r=r+1$.
Explicitly,
\begin{align}
\mathcal{E}_1 &= \eta x^{2}+x-1\,, \nonumber\\
\mathcal{E}_2 &= x^{3}+(4\eta-1)x^{2}-4\eta x+\eta\,, \nonumber\\
\mathcal{E}_3 &= \eta x^{4}+(2\eta+4)x^{3}-(\eta+8)x^{2}+(5-2\eta)x+\eta-1\,, \nonumber\\
\mathcal{E}_4 &= x^{5}+(9\eta+1)x^{4}-(12\eta+3)x^{3}+(4\eta-1)x^{2}+3x-1\,, \nonumber\\
\mathcal{E}_5 &= \eta x^{6}+(6\eta+9)x^{5}+(\eta-21)x^{4}+(16-22\eta)x^{3}
+(22\eta-4)x^{2}-8\eta x+\eta\,,
\nonumber
\end{align}
of degrees $2,3,4,5,6$ respectively.  We have also checked numerically, at $\eta=7/3$ and for
every $r\le6$, that each of the $r+1$ roots of $\mathcal{E}_r$ avoids the zeros of
$P_0,\dots,P_r$ and that the tuple reconstructed from it by \eqref{closedform:x} solves
\eqref{app_bethe_poly}, so that the table above extends to $r\le6$.

\subsection{The eliminant and the count for all $r$}
\label{app_eliminant}
Define polynomials $P_I(x)$ by
\begin{align}
P_{-1}&=P_0=1\,,
&
P_{I+1}(x)&=P_I(x)+(x-1)\,P_{I-1}(x)\,,
\label{Prec}
\end{align}
the first few being
\begin{align}
P_1=x\,,\qquad
P_2=2x-1\,,\qquad
P_3=x^2+x-1\,,\qquad
P_4=x\,(3x-2)\,.
\label{app_Ppoly}
\end{align}
Substituting $x=1+t$ turns \eqref{Prec} into $P_{I+1}=P_I+t\,P_{I-1}$, which is solved by
\begin{align}
P_I(x)=\sum_{k\ge0}\binom{I+1-k}{k}\,(x-1)^{k}\,,
\label{app_Pbinom}
\end{align}
which makes $\deg P_I=\lceil I/2\rceil$ manifest.%
\footnote{These are the Fibonacci polynomials in homogenized form.
Writing $F_n$ for the Fibonacci polynomials, $F_1=1$, $F_2=z$,
$F_{n+1}(z)=z\,F_n(z)+F_{n-1}(z)$, one has
$P_I(1+t)=t^{(I+1)/2}\,F_{I+2}\big(t^{-1/2}\big)$;
note that the recursion \eqref{Prec} carries the variable on the second rather than
the first term, so the two families agree only after this rescaling.
At $t=1$, i.e.\ $x=2$, one recovers the Fibonacci numbers themselves, $P_I(2)=F_{I+2}$.}
Two elementary properties follow by
induction, namely $P_I(1)=1$ for all $I$, and $\gcd(P_I,P_{I-1})=1$ (a common divisor
of $P_{I+1}$ and $P_I$ divides $(x-1)P_{I-1}$ and is coprime to $x-1$ since
$P(1)=1$, hence divides $\gcd(P_I,P_{I-1})$). The key identity is of Cassini
type,
\begin{align}
P_{I-1}(x)^2-P_{I-2}(x)\,P_I(x)=(1-x)^I\,,
\label{cassini}
\end{align}
proved by noting that
\begin{align}
D_I\equiv P_{I-1}^2-P_{I-2}P_I
\qquad\text{satisfies}\qquad
D_{I+1}=(1-x)\,D_I\,,
\qquad
D_1=1-x\,,
\label{cassiniproof}
\end{align}
by \eqref{Prec}.

\begin{proposition}
\label{prop_eliminant}
Fix $\eta\neq0$ and set
\begin{align}
\mathcal{E}_r(x;\eta)=\eta\,P_r(x)^2+(x-1)\,P_{r-1}(x)^2\,.
\label{eliminant}
\end{align}
The projection $(x_1,\dots,x_r)\mapsto x_1$ is a bijection from the set of solutions of
\eqref{app_bethe_poly} in $(\mathbb{C}^*)^r$ onto the set of roots $x$ of
$\mathcal{E}_r(x;\eta)$ satisfying $P_I(x)\neq0$ for $I\in\{0,1,\dots,r\}$.  Its inverse is
\eqref{closedform:x}.  Equivalently, for such a root $x$, a tuple
$(x_1,\dots,x_r)\in(\mathbb{C}^*)^r$ with $x_1=x$ solves \eqref{app_bethe_poly} if and only if the following equivalent expressions hold
\begin{subequations}
\label{closedform}
\begin{align}
x_I&=\frac{P_{I-2}(x)\,P_I(x)}{P_{I-1}(x)^2}\,,
\label{closedform:x}
\\
1-x_I&=\frac{(1-x)^I}{P_{I-1}(x)^2}\,,
\label{closedform:onemx}
\\
\prod_{J\le I}x_J&=\frac{P_I(x)}{P_{I-1}(x)}\,.
\label{closedform:prod}
\end{align}
\end{subequations}

\end{proposition}

\begin{proof}
With $w_c=\prod_{J\ge c}x_J$, the equations \eqref{app_bethe_poly} read
\begin{align}
1-x_I&=\eta^I\prod_{c\le I}w_c^2\,,
\qquad I=1,\dots,r\,.
\label{bethew}
\end{align}
Taking the ratio of consecutive equations gives
\begin{align}
1-x_{I+1}&=\eta\,w_{I+1}^2\,(1-x_I)\,,
\label{ratio}
\end{align}
and since $w_{I+1}=w_1/(x_1\cdots x_I)$ while the $I=1$ equation is
$\eta w_1^2=1-x$, equation \eqref{ratio} becomes
\begin{align}
x_{I+1}&=1-\frac{(1-x)(1-x_I)}{(x_1\cdots x_I)^2}\,,
\label{fwd}
\end{align}
a forward recursion determining $x_2,\dots,x_r$ uniquely from $x$. Induction
establishes \eqref{closedform}. The base case is trivial, and assuming
\eqref{closedform} at step $I$ (with $P_{I}(x)\neq0$ following from
$\prod_{J\le I}x_J=P_I/P_{I-1}\in\mathbb{C}^*$), equation \eqref{fwd} gives
\begin{align}
1-x_{I+1}
&=(1-x)\,\frac{(1-x)^{I}}{P_{I-1}^{2}}\,\frac{P_{I-1}^{2}}{P_{I}^{2}}
=\frac{(1-x)^{I+1}}{P_{I}^{2}}\,,
\label{indstepa}
\\
x_{I+1}
&=\frac{P_{I}^{2}-(1-x)^{I+1}}{P_{I}^{2}}
=\frac{P_{I-1}\,P_{I+1}}{P_{I}^{2}}\,,
\label{indstepb}
\end{align}
the last equality by \eqref{cassini}. Substituting
$\prod_{J\le r}x_J=P_r/P_{r-1}$ back into the $I=1$ equation
$\eta w_1^2=1-x$ yields
\begin{align}
\eta\left(\frac{P_r(x)}{P_{r-1}(x)}\right)^{2}=1-x
\qquad\Longleftrightarrow\qquad
\mathcal{E}_r(x;\eta)=0\,.
\label{finalsub}
\end{align}
Conversely, given a root $x$ of \eqref{eliminant} avoiding the zeros of the
$P_I$, defining $x_I$ by \eqref{closedform:x} and running the computation
backwards, using \eqref{cassini} at each step, verifies every equation of
\eqref{app_bethe_poly}.
\end{proof}

Since $\deg P_r=\lceil r/2\rceil$, the two terms of \eqref{eliminant} have
degrees $2\lceil r/2\rceil$ and $1+2\lceil(r-1)/2\rceil$, so that
\begin{align}
\deg_x \mathcal{E}_r
=\max\!\big(2\lceil r/2\rceil,\;1+2\lceil(r-1)/2\rceil\big)
=r+1
\label{degE}
\end{align}
for every $r$ and every $\eta\neq0$. For generic $\eta$ the roots are simple.
Since $P_r(1)=1$ and $\gcd(P_r,P_{r-1})=1$, the two members of the pencil
\eqref{eliminant} are coprime. If the discriminant of $\mathcal{E}_r$ in $x$ vanished
identically in $\eta$, there would be an algebraic family $x(\eta)$ with
$\mathcal{E}_r=\partial_x\mathcal{E}_r=0$, and differentiating $\mathcal{E}_r(x(\eta);\eta)=0$ along the
family, \emph{i.e.}\ with respect to $\eta$, would give $P_r(x(\eta))^2=0$ and hence also $(x-1)P_{r-1}^2=0$,
contradicting coprimality. The same argument shows that for generic $\eta$ no
root of $\mathcal{E}_r$ lies on the fixed, finite set of zeros of $\prod_{I\le r}P_I$.
Hence
\begin{align}
\#\,\text{Bethe vacua}\big[U(1)^r_{K_r}\big]=r+1
\qquad\text{for all }r\ge1\text{ and generic }\eta\,,
\label{count}
\end{align}
proving \eqref{bethe_count} in general.

\paragraph{Remark}
The closed form \eqref{closedform} also lets us analyze the behavior of the vacua in the two
limits of $\eta \to 0, \infty$.
As $\eta\to0$ the roots of \eqref{eliminant} approach the simple zero $x=1$ of
$(x-1)P_{r-1}^2$ together with the zeros of $P_{r-1}$, each of the latter with
multiplicity two, and, for $r$ odd, a single further root running off to $x\to\infty$.  The last
of these is accounted for by the drop in the degree of \eqref{eliminant} at $\eta=0$, in exact
parallel with the $\eta\to\infty$ behavior below.
At $x=1$ one has $x_I=1$ for all $I$ by \eqref{closedform:onemx}, an interior point of
the torus, while at a zero of $P_{r-1}$ the coordinate
$x_r=P_{r-2}P_r/P_{r-1}^2$ diverges and the solution escapes to the boundary of
$(\mathbb{C}^*)^r$.
Exactly one Bethe vacuum therefore remains interior as $\eta\to0$.
In the opposite limit $\eta\to\infty$ the roots approach the zeros of $P_r$, where
$x_r\to0$ by \eqref{closedform:x}, together with, for $r$ even, a single root running
off to $x\to\infty$. Hence all solutions escape.
Which of the escaping solutions still describe vacua of the deformed theory is the
question raised in section~\ref{sec_vacuum_counting}.

\bibliographystyle{utphys}
\bibliography{ref}

\end{document}